\documentclass{amsart}

\usepackage{amsthm}
\usepackage{amssymb}
\usepackage{upgreek}
\usepackage{xurl}
\usepackage[colorlinks=true,
            urlcolor=blue,
            linkcolor=black,
            citecolor=black]{hyperref}
\usepackage{graphicx} 
\usepackage{caption}  
\usepackage{tikz-cd}

\usepackage{amsmath}
\usepackage[margin=1.2in]{geometry}
\usepackage[toc,page]{appendix}
\usepackage{comment}

\AtBeginDocument{%
	}

\usepackage{cleveref}
\usepackage{listings}
\usepackage{enumitem}
\usepackage{microtype}
\usepackage{multicol}
\usepackage[T1]{fontenc}
\usepackage{lmodern}

\makeatletter
\renewcommand\paragraph{\@startsection{paragraph}{4}%
  \z@{3.25ex \@plus 1ex \@minus .2ex}%
  {-1em}
  {\normalfont\bfseries}}%
\makeatother

\usepackage{xurl}

\usepackage[colorinlistoftodos]{todonotes}

\definecolor{keywordcolor}{rgb}{0.7, 0.1, 0.1}   
\definecolor{commentcolor}{rgb}{0.4, 0.4, 0.4}   
\definecolor{errorcolor}{rgb}{0.9, 0.0, 0.0}     
\definecolor{symbolcolor}{rgb}{0.4, 0.4, 0.4}    
\definecolor{sortcolor}{rgb}{0.1, 0.5, 0.1}      
\definecolor{stringcolor}{rgb}{0.31, 0.71, 0.14}      

\usepackage{xspace}
\usepackage[utf8]{inputenc}

\newtheorem{remark}{Remark}

{\theoremstyle{definition}
\newtheorem{example}{Example}}
\newtheorem{definition}{Definition}
\newtheorem{theorem}{Theorem}
\newtheorem{lemma}{Lemma}
\newtheorem{corollary}{Corollary}
\newtheorem*{lemma*}{Lemma}
\newtheorem*{theorem*}{Theorem}

\title{Formally certifying number field invariants}

\author{Alain Chavarri Villarello }
\address{Department of Mathematics \\
 VU Amsterdam \\
De Boelelaan 1111 \\
1081 HV Amsterdam \\
The Netherlands} 
\email{a.chavarri.villarello@vu.nl}

\author{Sander R. Dahmen }
\address{Department of Mathematics \\
 VU Amsterdam \\
De Boelelaan 1111 \\
1081 HV Amsterdam \\
The Netherlands} 
\email{s.r.dahmen@vu.nl}

\begin{document}

	\begin{abstract}

Number fields, which generalize the rational numbers, are fundamental objects in number theory. Many of their key arithmetic properties are captured by invariants whose computation is among the central tasks of computational algebraic number theory and a focus of several computer algebra systems and databases. In this paper, we describe a Lean 4 formalization for certifying several of these number field invariants. Building on previous work on certifying rings of integers, we extend this certification approach to further invariants including the signature, the unit group modulo $p$-th powers, and, ultimately, the class group. We also improve discriminant certification, allowing verifications for higher-degree number fields infeasible in previous work. We introduce structures based on representations of algebraic objects suited to computation, including reusable ones for certifying ideal arithmetic. Along the way, we formalize several underlying mathematical results, for instance on real closed fields and pseudo-remainder sequences, which are of independent interest. We apply our framework to verify hundreds of entries of the \emph{L-functions and modular forms database} (LMFDB) concerning the discriminant, signature, class number, and class group structure of various number fields. To this end, we wrote a SageMath script that computes the certificates and outputs Lean proofs of the corresponding statements.   
	\end{abstract}
    
\maketitle

\section{Introduction}
Number fields are one of the central objects of algebraic number theory. They generalize the field of rational numbers $\mathbb{Q}$ by being finite degree field extensions of it. Attached to each number field is a collection of invariants that govern its key arithmetic properties. These include the ring of integers, discriminant, signature, unit group, and class group. For an explicit number field, computing these invariants forms part of the main computational tasks of algebraic number theory \cite[Section 4.9.3]{Cohen}. 

Involved algorithms to compute these arithmetic data are implemented in different Computer Algebra Systems (CASs), including PARI/GP \cite{PARI2}, SageMath \cite{SageMath}, and Magma \cite{magmasystem}. Using these tools, the invariants of large families of number fields have been determined and collected in mathematical databases such as the well-known \emph{L-Functions and Modular Forms Database} (LMFDB) \cite{lmfdb}.

While CASs are very effective in practice, they rely on large stacks of highly optimized components whose implementations lack formal guarantees of correctness. As with any large and complex piece of software, subtle bugs have been found (e.g. \cite[pag. 162]{magma} reports a bug in class group computations). This motivates formally verifying some of these computations inside a proof assistant such as Lean, where everything is checked against a small trusted kernel. Additionally, the formal verification of number field invariants allows for their use in future formal proofs, such as in the formal resolution of Diophantine equations. For instance, in order to solve some explicit Mordell equations, rings of integers and class group calculations were carried out in Lean in \cite{alex} for a selection of degree two number fields. 

As a first substantial step towards formally verifying invariants of arbitrary number fields, the verification in Lean of rings of integers, together with the development of the necessary computational infrastructure for working in these rings and with polynomials, was done in \cite{rings}. We build on this work, which follows a \emph{certification approach}: rather than formalizing expensive algorithms inside the proof assistant, an external CAS is used to carry out the more computationally intensive tasks and output a certificate which is checked in Lean to guarantee correctness of the computation. The current paper is a substantial continuation of this approach, significantly broadening the scope by verifying several further number field invariants and culminating in the certification of the class group.

Certifying the class group of a number field was the central goal guiding our development. This finite abelian group, whose definition and basic properties were formalized in Lean in  \cite{dedekindforma}, measures the failure of the ring of integers (which generalizes how $\mathbb{Z}$ is contained in the number field $\mathbb{Q}$) to have unique factorization into irreducible elements. Its computation is expensive in general, relying on involved algorithms such as \cite[Algorithm 6.5.9]{Cohen}, \cite[Section 6.5]{zassenhaus} or \cite[Algorithm 3.7.1]{hess}.

The main contribution of this paper is a formalization in Lean 4 for certifying several central computations in algebraic number theory. These include the discriminant and signature of a number field, the unit group modulo $p$-th powers, and, ultimately, the structure of the class group and its order (the \emph{class number}). We also build the infrastructure to certify the decomposition of prime numbers into prime ideals. Class group and unit group computations, together with prime ideal decompositions, are listed among the main computational tasks of algebraic number in theory in \cite[Section 4.9.3]{Cohen}, while discriminant and signature computations are necessary ingredients and of independent interest. 

In order to achieve this, we formalize many underlying mathematical results and definitions, including results on real closed fields and pseudo remainder sequences. To bridge the theory with the computational aspects, we introduce several reusable certifying structures based on representations of the relevant algebraic objects that are suitable for computation, together with proof obligations that can be discharged automatically. We also address challenges in scaling formal verification by partitioning verification into smaller independent batches. 

The structure of this paper is as follows. In Section \ref{sec : prelim} we establish notation and go through some preliminaries of algebraic number theory. 
In Section \ref{sec:ideal} we describe the computational framework used throughout (inherited from \cite{rings}), discuss representations of ideals in Lean, and develop certificates for ideal arithmetic and primality of ideals.

Section \ref{sec:disc} concerns the \emph{discriminant} of a number field, an integer measuring its arithmetic complexity, which is tightly related to the discriminant of a defining polynomial. Polynomial discriminants, and their relation with the discriminant of a number field, were first formalized in Lean in \cite{rings}, however, the computation was infeasible there for degrees higher than 3. Here, we formalize a more efficient method to compute polynomial discriminants using pseudo-remainder sequences, which allowed us to certify the discriminant of higher degree number fields.

In Section \ref{sec : sign}, we deal with the \emph{signature} of a number field, which controls the rank of the unit group. To verify it, we formalize a version of Sturm’s theorem that is suitable for computation, allowing us to efficiently count the number of real roots of an integer polynomial using only arithmetic over $\mathbb{Z}$ and controlling for coefficient explosion. In the process, we formalize foundational results in the theory of real closed fields. 

All this comes together in the certification of the class group. Section \ref{sec : idealgen} is concerned with certifying a factor base and generators for the cyclic factors of the class group. Section \ref{sec:satclass} is about a $p$-saturation certificate, which allows us to verify generators of the unit group modulo $p$-th powers (an important step towards a full certification of the unit group) as well as, after certifying saturation at sufficiently many primes, the full structure of the class group and the class number. In Section \ref{sec:existancesat} we give a proof of the existence of this certificate as well as a heuristic argument for the effort required to find it. 

To illustrate the feasibility of our approach, we formally verified hundreds of entries in the \emph{Number Field} section of the \emph{L-functions and modular forms database} (LMFDB). For a selection of number fields, we verified in Lean their discriminant, signature, class group structure, and class number. To do this, we wrote a SageMath script that, for a given a number field, computes the corresponding certificates and outputs Lean files containing the formal proofs that can be readily checked. A discussion of these results, with verification timings, appears in Section \ref{sec : lmfdb}, and bottlenecks with possible optimizations appear in Section  \ref{sec : optimizations}. 
We conclude the paper with some discussion in Section \ref{sec : discussion}, including alternative approaches, related work, and future work. 

Full source code of our formalization and SageMath script is available in the GitHub repository \href{https://github.com/alainchmt/CertifyingInvariantsNF/tree/v1}{\texttt{CertifyingInvariantsNF}}. 
Throughout the paper, we accompany formal statements, definitions, and examples with links to the corresponding source code, either in Mathlib or in our repository, each pinned to a fixed commit. Some of the code snippets have been lightly edited for readability.

\subsection*{Acknowledgements} Alain Chavarri Villarello was funded by NWO grant No. 613.009.143, Formalizing Diophantine algorithms. We thank Assia Mahboubi, Bhavik Mehta, and John Voight for helpful conversations. 

\subsection*{AI disclosure statement} We did not use generative AI for the Lean formalization, the associated mathematical arguments, or the design of the formalization strategies. To find Mathlib lemmas, we used the semantic search tool LeanSearch. In preparing this manuscript, we occasionally used Overleaf's AI tools for copyediting (grammar and wording suggestions), but not for text generation. We initially wrote a working version of our SageMath code without generative AI, and later used Claude Code with Sonnet 5 for limited code and docstring edits, cleanup, and merging of SageMath~files.

\section{Algebraic number theory preliminaries} \label{sec : prelim}

We presuppose some familiarity with standard notions of graduate algebra (as can be found in \cite{lang2}). These are ideals, fractional ideals, quotient rings, modules and algebras over commutative rings, and basic properties of Dedekind domains and finite fields. An overview is given in Appendix \ref{App:B}. 

Throughout this paper, all rings have an identity $1$ and ring homomorphisms preserve $1$. We fix some notation. Let $R$ be a commutative ring and $I$ an ideal of $R$. When $I$ is clear from context, we denote the image of $x \in R$ under the quotient map $R \to R/I$ by $\bar{x}$. The ideal generated by $\{a_1, \ldots, a_n\} \subseteq R$ is denoted by $\langle a_1, \ldots, a_n\rangle$. For $M$ a module over $R$, the $R$-span of a subset $A$ of $M$ is denoted $\operatorname{span}_{R} A$. 
Here, we review the number-theoretic background needed for this paper. A standard reference is \cite{neuk}.

\subsection{Number fields and rings of integers}
A \emph{number field} $K$ is a finite field extension of $\mathbb{Q}$. Its dimension $[K : \mathbb{Q}]$ as a $\mathbb{Q}$-vector space is called the \emph{degree} of $K$. Given an irreducible polynomial $f$ over $\mathbb{Q}$ of degree $n$, adjoining one of its complex roots $\beta$ to $\mathbb{Q}$ results in a number field $\mathbb{Q}(\beta)$ of degree $n$. Conversely, every number field is isomorphic to some $\mathbb{Q}(\beta)$, with $\beta$ the root of an irreducible polynomial $f$ over $\mathbb{Q}$. Any polynomial with this property is called a \emph{defining polynomial} of $K$. It can always be chosen to be monic (leading coefficient equal to one) and with coefficients in $\mathbb{Z}$.

A number field $K$ of degree $n$ generalizes $\mathbb{Q}$. The analog of $\mathbb{Z}$ is the \emph{ring of integers} of $K$, given by
\begin{align*}
	\mathcal{O}_K = \{x \in K \mid \exists P \in \mathbb{Z}[X] \text{ monic with } P(x) = 0 \}.
\end{align*}
This is a subring of $K$ whose field of fractions equals $K$. Viewed as a $\mathbb{Z}$-module, it has a $\mathbb{Z}$-basis \cite[Proposition 2.10]{neuk}, called an \emph{integral basis}, of cardinality equal to $n$. It follows that, given an integer $m \in \mathbb{Z} \setminus{\{0\}}$, the quotient $\mathcal{O}_K / \langle m\rangle$ has size $|m|^n$. Every nonzero ideal $I$ of $\mathcal{O}_K$ contains a nonzero integer, and thus $\mathcal{O}_K/I$ is a finite ring. Its cardinality is the \emph{norm} of the ideal $I$, denoted by $N(I)$. For convenience, we set $N(\langle 0 \rangle ) = 0$. Number fields and rings of integers have been defined in Mathlib. 

\subsection{Prime ideals and factorization}

A key property, also found in Mathlib, of the ring of integers $\mathcal{O}_K$ of a number field $K$ is that of being a Dedekind domain. This implies that every nonzero proper ideal $I$ of $\mathcal{O}_K$ factors into prime ideals \cite[Theorem 3.3]{neuk}. Consequently, for nonzero ideals $I$ and $J$, $I \subseteq J$ if and only if $I = JH$ for some ideal $H$. Furthermore, for a nonzero prime ideal $\mathfrak{p}$ of $\mathcal{O}_K$, the quotient $\mathcal{O}_K/\mathfrak{p}$ is a finite field of characteristic $p$ for some prime number $p$. Hence, $\langle p \rangle \subseteq \mathfrak{p}$, and we say that $\mathfrak{p}$ lies \emph{above} $p$. The number of prime ideals lying above $p$ is finite, and we can factor $\langle p\rangle$ as
\begin{align*}
	\langle p \rangle = \prod_{\langle p\rangle \subseteq \mathfrak{p}} \mathfrak{p}^{e_\mathfrak{p}}, 
\end{align*}
with integers $e_\mathfrak{p} \geq 1$. The cardinality of $\mathcal{O}_K/\mathfrak{p}$ is $p ^ {f_\mathfrak{p}}$ for some $f_\mathfrak{p} \geq 1$ called the \emph{inertia degree} of $\mathfrak{p}$.

\subsection{Embeddings}
A ring homomorphism $ \sigma : K \hookrightarrow \mathbb{C}$, necessarily injective, is called an embedding of $K$ into the complex numbers. The embeddings whose image is fully contained in $\mathbb{R}$ are called \emph{real embeddings} and the remaining ones are the \emph{non-real complex embeddings}. The latter come in conjugate pairs, as composing with complex conjugation induces an involution on the set of non-real complex embeddings that has no fixed points. We denote by $r_1$ the number of real embeddings (or \emph{number of real places}), and by $r_2$ the number of conjugate pairs of non-real complex embeddings (or \emph{number of complex places}). We have the relation $n = r_1 + 2r_2$, with $n$ the degree of the number field $K$. The tuple $(r_1, r_2)$ is the \emph{signature} of $K$ and $r_1 + r_2$ is the number of \emph{infinite places} of $K$. Definitions of complex embeddings and number of real, complex, and infinite places can be found in Mathlib. 

As an example, adjoin the real cube root $\sqrt[3]{2}$ to $\mathbb{Q}$. This results in a number field $K = \mathbb{Q}(\sqrt[3]{2})$ of degree $3$. Besides the identity embedding, which is real, there are two other (non-real complex) embeddings. One is given by sending $\sqrt[3]{2} \mapsto \zeta \sqrt[3]{2}$, with $\zeta = e^{2\pi i /3}$, and the other by composing with complex conjugation, i.e. by sending $\sqrt[3]{2} \mapsto \zeta^2 \sqrt[3]{2}$. The signature of $K$ is thus $(1,1)$.

\subsection{Failure of unique factorization}
The ring of integers of a number field $K$ may fail to be a PID. A classic example is $\mathbb{Q}(\sqrt{-5})$, where the ideal $\langle 2, 1 - \sqrt{-5}\rangle$ of its ring of integers $\mathbb{Z}[\sqrt{-5}]$ is nonprincipal.
This motivates defining the \emph{class group} of $K$, measuring the failure of $\mathcal{O}_K$ to be a PID.

Let $\mathcal{I}_K$ be the group of invertible (equivalently, nonzero) fractional ideals of $\mathcal{O}_K$. Let $\mathcal{P}_K$ be the group of principal fractional ideals of $\mathcal{O}_K$ (fractional ideals generated, as $\mathcal{O}_K$-modules, by a single $x \in K^{\times}$). The group $\operatorname{Cl}(\mathcal{O}_K) = \mathcal{I}_K/\mathcal{P}_K$ is the class group of $K$. A fundamental theorem \cite[Theorem 6.3]{neuk} states that this group is finite. Its cardinality $h_K$ is the \emph{class number} of $K$. We have $h_K = 1$ if and only if $\mathcal{O}_K$ is a PID. Moreover, $\mathcal{O}_K$ being a PID is equivalent to $\mathcal{O}_K$ having unique factorization into irreducible elements (up to multiplication by units and reordering of factors), so the class group also measures the failure of unique factorization. The class group, its finiteness, and the class number have been formalized in Lean and are available in Mathlib (this is discussed in Section \ref{sec : class}).

\section{Ideal arithmetic}\label{sec:ideal}

Throughout this project, we need to compute in the ring of integers $\mathcal{O}_K$ of a number field $K$ within Lean. In this section, we consider a more general setting with $R$ and $S$ commutative rings, where $S$ is an $R$-algebra which is free and finitely generated as an $R$-module (the example to keep in mind is $S = \mathcal{O}_K$ and $R = \mathbb{Z}$).
We start by reviewing a computational framework, introduced in \cite{rings}. Building on this, we develop certificates for explicit ideal calculations, which will be crucial. 

\subsection{The computational framework} \label{sec : model}
The framework used throughout this project for computing in $S$ is inherited from \cite{rings}, where it is used to certify rings of integers. We recall it here. 

Given an $R$-basis $(b_1, \ldots, b_r)$ for $S$, any element in $S$ is uniquely represented by its coordinate tuple in $R^r$ (of type \lstinline[mathescape]|Fin r $\to R$| in Lean). The multiplication table $T$ corresponding to this basis is the matrix where the $ij$-th entry is $(a_{ijk})_{k = 1}^{r} \in R^r$ with $b_i b_j = \sum_k a_{ijk}\cdot b_k$. Multiplication in $S$ is determined by pairwise multiplication of basis elements, hence can be performed explicitly (with tuple representations) using $T$. The structure \href{https://github.com/alainchmt/CertifyingInvariantsNF/blob/v1/IdealArithmetic/DedekindProject/TimesTable/Defs.lean#L7}{ \lstinline|TimesTable (Fin r) R S|}, taken from \cite{rings}, is used to bundle together an $R$-basis of $S$, the corresponding multiplication table, and a proof of its correctness. Assuming that arithmetic in $R$ is computable, arithmetic operations in $S$ can be translated into corresponding computable operations in terms of tuples and lists over $R$. Further details can be found in \cite[Section 5.2]{rings}.

We apply this framework to the ring of integers of a number field $K$, which is available in Mathlib under the name \href{https://github.com/leanprover-community/mathlib4/blob/1a1a7c154755bbd9c9c092cc41c820bbdbb79e66/Mathlib/NumberTheory/NumberField/Basic.lean#L92-L103}{\lstinline|NumberField.ringOfIntegers K|} and is defined as the subalgebra \lstinline|integralClosure ℤ K| coerced into a type. Indeed, let $K$ be number field of degree $r$ obtained by adjoining a root of an irreducible monic polynomial $f \in \mathbb{Z}[X]$ to $\mathbb{Q}$, and fix an integral basis for $\mathcal{O}_K$. We use \cite[Section 5]{rings} to automatically construct an explicit $\mathcal{O}$ of type \lstinline[mathescape]|Subalgebra $\mathbb{Z}$ K| defined by this basis, together with a term of type \lstinline[mathescape]|TimesTable (Fin r) $\mathbb{Z}$  $\mathcal{O}$| and a proof of \lstinline[mathescape]|$\mathcal{O}$ = integralClosure ℤ K|. Elements in $\mathcal{O}$ are represented as tuples \lstinline[mathescape]|Fin r $\to \mathbb{Z}$|. We perform addition, multiplication, and scalar multiplication in $\mathcal{O}$ by using this tuple representation of elements and the multiplication table. 

Throughout this project, we also compute with polynomials. Mathlib polynomials are implemented in such a way that makes their operations \lstinline|noncomputable|. Instead, we use the computable representation introduced in \cite{rings}, where univariate polynomials are encoded as lists of coefficients (see \cite[Section 4.1]{rings}).

	\subsection{Ideal representation in Lean}
    An ideal of $S$ can be seen both as an $S$-module and an $R$-module. 
    Assuming $R$ is Noetherian, every ideal $I$ of $S$ admits a finite set of $S$-generators $\{v_1, \ldots, v_{m}\} \subseteq S$, such that $I = \langle v_1, \ldots, v_{m}\rangle$, or a finite set of $R$-generators $\{w_1, \ldots, w_{n}\}\subseteq S$, so that $I = \operatorname{span}_{R}\{w_1, \ldots, w_{n}\}$.

	The type \lstinline|Ideal S|, defined in Mathlib, is the type of ideals of $S$ and is definitionally equal to \lstinline|Submodule S S|, the type of $S$-submodules of $S$. In our development, we define explicit ideals in Lean via their $S$-generators using \lstinline|Ideal.span|. In contrast, 
    defining an ideal via their $R$-generators would additionally require us to prove that the resulting $R$-module carries a compatible $S$-module structure.

	As previously noted, we represent elements of $S$ using their $R$-coordinates with respect to an $R$-basis of $S$. In Lean, the tuple of $S$-generators $(v_1, \ldots, v_{m})$ can be represented as a function \lstinline|v : Fin m → Fin r → R| sending each index $i$ to the coordinate tuple of $v_i$. Let \lstinline|B : Basis (Fin r) R S| be an $R$-basis of $S$. Then we define the ideal $I$ generated by the $v_i$ as 
	\begin{lstlisting}
		I := Ideal.span (Set.range (fun i ↦ B.equivFun.symm (v i)))
	\end{lstlisting}
	which has type \lstinline|Ideal S|. The function \lstinline|B.equivFun.symm| is the $R$-linear isomorphism \lstinline[mathescape]!(Fin r → R) $\simeq_l$[R] S! sending a tuple in $R^r$ to the element in $S$ with those coordinates (relative to the basis \lstinline|B|).

	Although we do not define ideals from their $R$-generators, given an ideal \lstinline|I : Ideal S|, we may assert that it is generated by $(w_1, \ldots, w_{n})$ as an $R$-module via the proposition
	\begin{lstlisting}
		I.carrier = Submodule.span R (Set.range (fun i ↦ B.equivFun.symm (w i)))
	\end{lstlisting}
	with \lstinline|w : Fin n → Fin r → R|. Here, \lstinline|I.carrier| denotes the carrier set of \lstinline|I|. A description of $I$ in terms of $R$-generators will be useful for certain computations, such as the norm later described in Section \ref{sec : certnorm}. 
	
	As illustrated in the rest of this section, for computational purposes, ideals of $S$ will simply be represented as collections of tuples of coordinates corresponding to their $R$ or $S$-generators.

	\subsection{Ideal equality}\label{idealeq}
	Given two representations of the same ideal, we want to verify their equality by only performing additions and multiplications inside $S$. Let $I = \langle v_1, \ldots, v_{m}\rangle $ and $J = \langle w_1, \ldots, w_{n}\rangle $. To prove $I = J$, it suffices to verify the mutual inclusions \(I \subseteq J\) and \(J \subseteq I\). For the inclusion $I \subseteq J$, it is enough to exhibit, for every $1 \leq i \leq m$, a witness tuple $g_i \in S^n$ such that $v_i = \sum_k g_{ik}  w_k$ (and vice versa for $J \subseteq I$). Verifying these identities requires a total of $2mn$ multiplications in $S$.

	To implement this certificate in Lean, we define the following \href{https://github.com/alainchmt/CertifyingInvariantsNF/blob/v1/IdealArithmetic/IdealArithmetic/IdealArithmetic.lean#L327}{structure}. 
	\begin{lstlisting}
structure IdealEqCertificateO' {R : Type*} [CommRing R]{r n m : ℕ} 
  (T : Fin r → Fin r → List R)  (v : Fin m → Fin r → R) (w : Fin n → Fin r → R)  where ...
	\end{lstlisting}
	The parameter $T$ encodes the multiplication table, while \lstinline|v| and \lstinline|w| encode the coordinates of $S$-generators $v_i $ of $I$ and $w_j$ of $J$. 
	Included as fields of this structure are \lstinline|g : Fin m → Fin n → Fin r → R|, the tuples $g_i$ witnessing the inclusion $I \subseteq J$  (analogously for $J \subseteq I$), together with the proofs of the identities described above purely in terms of lists and tuples over $R$ (using the multiplication table). 
    
    Recall that our main use case is $R = \mathbb{Z}$, which has decidable equality in Lean. The tuples witnessing the inclusions, which serve as a certificate of equality, can be computed externally using a CAS. The proofs of the corresponding identities can then be automatically discharged using the tactic \lstinline|decide|.  
    
	Importantly, this structure refers purely to data and proofs in terms of tuples and lists over $R$. It makes no explicit mention of ideals or the $R$-algebra $S$. The only typeclass assumption is that of $R$ being a commutative ring, which reduces typeclass inference time when constructing terms of this type.
	
	Under the assumption that $T$ gives a multiplication table for $S$, and that $I$ and $J$ are generated by $v$ and $w$, respectively, we obtain a proof of \lstinline|I = J| from a term of type \lstinline|IdealEqCertificateO' T v w|.

	Sometimes it is convenient to work with a set of $R$-generators of $I$, or better yet, an $R$-basis (if one exists).  Given an ideal described in terms of $S$-generators, say $I = \langle v_1, \ldots, v_{m}\rangle$, and a set of elements $\{w_1, \ldots, w_{n}\} \subseteq S$, we want to certify that $I = \operatorname{span}_{R} \{w_1, \ldots, w_{n}\}$. Note that for the inclusion $I \subseteq \operatorname{span}_{R} \{w_1, \ldots, w_{n}\}$ it is not sufficient to show that each $v_i \in \operatorname{span}_{R} \{w_1, \ldots, w_{n}\}$ as we do not assume that we have, a priori, a proof that $\operatorname{span}_{R} \{w_1, \ldots, w_{n}\}$ is an ideal of $S$.  
	
	Let $(b_1, \ldots, b_{r})$ be an $R$-basis of $S$. Then $$I= \operatorname{span}_R \{v_1b_1,\ldots, v_1b_{r}, \dots, v_{m}b_1, \ldots, v_{m}b_{r} \}.$$ Proving $I \subseteq \operatorname{span}_{R} \{w_1, \ldots, w_{n}\}$ is done by providing, for every $1 \leq i \leq m $ and every $ 1 \leq j \leq r $ a tuple $g_{ij} \in R^n$ such that $v_ib_j = \sum_k g_{ijk} \cdot w_k$. The other inclusion is proved by exhibiting, for every $1 \leq j \leq n$, a matrix $f_j \in R ^ {m \times r}$ such that $w_j = \sum_{i,k} f_{jik} \cdot v_i b_k$. In total, we require $rm$ multiplications in $S$ and $2rmn$ scalar multiplications.  
	We implement this certificate as the \href{https://github.com/alainchmt/CertifyingInvariantsNF/blob/v1/IdealArithmetic/IdealArithmetic/IdealArithmetic.lean#L204}{structure}
	\begin{lstlisting}
structure IdealEqSpanCertificate' {R : Type*} [CommRing R] {r n m : ℕ}
  (T : Fin r → Fin r → List R) (v : Fin m → Fin r → R) (w : Fin n → Fin r → R)  where ...
	\end{lstlisting}
	It includes the tuples $g_{ij}$, the matrices $f_j$, and the coordinates of $v_ib_j$ as part of its fields, together with the proofs of the required identities in terms of lists and tuples over $R$. Under the assumption that $T$ gives a multiplication table for $S$ and that $I$ is generated by $v$, from a term of this type one extracts a proof of $I = \operatorname{span}_R \{w_1,\ldots, w_{n}\}$. 
	
	\begin{remark}
The usual approach to ideal equality in CASs \cite[Section 4.7]{Cohen} is to pass to \(\mathbb{Z}\)-generators. If $I$ is given by $m$ $\mathcal{O}_K$-generators, multiplying them with a $\mathbb{Z}$-basis of $\mathcal{O}_K$ of size $r$ gives a set of $mr$ $\mathbb{Z}$-generators of $I$. These are stored as an $r \times (mr)$ integer matrix. Ideal equality is checked by comparing the Hermite normal forms of these matrices. 
However, computing the $\mathbb{Z}$-generators of $I$ and $J$ (given by $m$ and $n$ $\mathcal{O}_K$-generators, respectively) requires $nr + mr$ multiplications in $\mathcal{O}_K$. By contrast, \lstinline|IdealEqCertificateO'| requires  $2mn$ multiplications. Since $\mathcal{O}_K$ is a Dedekind domain, we can always take $m, n \leq 2$, so our certificate is cheaper as soon as $r > 2$. 

	\end{remark}

	\subsection{The norm of an ideal}\label{sec : certnorm}
    Recall that the norm $N(I)$ of a nonzero ideal $I$ of $\mathcal{O}_K$ is defined as the cardinality of the finite quotient ring $\mathcal{O}_K/I$. 
	It is essential for our purposes to compute $N(I)$ in Lean. We do this by relating it to the more general notion of \emph{index}. 
	
Recall our setting with $S$ an $R$-algebra that is free of rank $r$ as an $R$-module.  Furthermore, suppose that $R$ is a PID and $S$ an integral domain. Under these assumptions, a nonzero ideal $I$ of $S$ admits an $R$-basis of size $r$. Hence, if we know that $I = \operatorname{span}_{R}\{v_1, \ldots, v_r\}$, then we can conclude that $(v_1, \ldots, v_r)$ is an $R$-basis of $I$. 
	 The index $[S : I]$ is an element in $R$ defined (up to units in $R$) as the determinant of the $r \times r$ matrix representing the $R$-linear map given by inclusion $\iota : I \hookrightarrow S$. In Lean, this was formalized in \cite[Section 5.3]{rings}. When $R = \mathbb{Z}$, then the absolute value of $[S : I]$ is precisely $\#S/I$.

     If the coordinates of $v_i$ with respect to an $R$-basis $\mathcal{B}$ are $(a_{ij})_{j=1}^r \in R^r$, then $\iota$ is represented by the matrix $(a_{ij})_{ij} ^ T$, and thus $[S : I] = \det (a_{ij})_{ij}$. In fact, it is always possible to find a basis of $I$ such that $(a_{ij})_{ij}$ is lower triangular. In that case $[S : I] = \prod_i a_{ii}$. 
	This allows us to compute $[S : I]$ cheaply, as we can freely choose a convenient $R$-basis of $I$ found using external means, such as a CAS. After showing that $I$ is indeed spanned by such basis using the techniques in Section \ref{idealeq}, computing $[S : I]$ entails simply computing a product of $r$ elements in $R$, thus allowing us to compute norms of ideals in rings of integers of number fields. We formalize this result as \href{https://github.com/alainchmt/CertifyingInvariantsNF/blob/v1/IdealArithmetic/IdealArithmetic/CertifyPrimeIdeal.lean#L103}{ \lstinline|ideal_index_associated_prod|} and specialize to $\mathbb{Z}$ and the norm in \href{https://github.com/alainchmt/CertifyingInvariantsNF/blob/v1/IdealArithmetic/IdealArithmetic/CertifyPrimeIdeal.lean#L137}{\lstinline|ideal_norm_eq_prod'|}. 
	
	\begin{example}\label{ex:norm}
		Consider the number field $K = \mathbb{Q}(\alpha)$ obtained by adjoining a root $\alpha$ of the polynomial $f = X^5 + 10X^3 - 10X^2 - 15X - 18$. An integral basis for a subring $\mathcal{O}$ is given by $$\mathcal{B} = (b_1, \ldots, b_5) = (1, \alpha,  \alpha^2, 1/2\alpha^3 - 1/2\alpha, 1/24\alpha^4 - 1/8\alpha^3 - 5/24\alpha^2 + 5/24\alpha - 1/4).$$

Take the ideal defined by $I = \langle 3, 2b_1 + 3b_2 -3b_3 -b_4-5b_5\rangle$. In \href{https://github.com/alainchmt/CertifyingInvariantsNF/blob/v1/IdealArithmetic/Examples/Paper/NormExample.lean#L9}{Lean},
		\begin{lstlisting}[mathescape]
def I : Ideal $\mathcal{O}$ :=
    Ideal.span (Set.range (fun i ↦ B.equivFun.symm (![![3, 0, 0, 0, 0], ![2, 3, -3, -1, -5]] i)))
		\end{lstlisting}

		Using \lstinline|IdealEqSpanCertificate'|, we can verify that
		\begin{align*}
			I = \operatorname{span}_{\mathbb{Z}} \left\{ \begin{bmatrix}
				3 & 0 & 0& 0& 0 \\
				0 &3 &0 &0& 0 \\
				2 &2& 1 &0 &0 \\
				1 &1 &0& 1 &0 \\
				0& 1& 0 &0 &1
			\end{bmatrix} \cdot \begin{pmatrix}
				b_1 \\ b_2 \\ b_3 \\ b_4 \\ b_5
			\end{pmatrix} \right\}, 
		\end{align*}
		and conclude that \href{https://github.com/alainchmt/CertifyingInvariantsNF/blob/v1/IdealArithmetic/Examples/Paper/NormExample.lean#L30}{\lstinline[mathescape]|Nat.card ($\mathcal{O}$ $/$ I) = 9|}, where \lstinline|Nat.card| gives the cardinality of a type (its zero if the type is infinite), by using \lstinline|ideal_norm_eq_prod'| and multiplying the diagonal of the matrix. 

	\end{example}

	\subsection{Ideal multiplication} \label{sec : idealmul}

	Given two ideals $I$ and $J$ of $S$, their product $IJ$ is the smallest ideal containing $ \{xy \mid x\in I, y \in J  \}$, or equivalently the $S$-span of this set. If $I,J,$ and $L$ are ideals of $S$, we want to certify $IJ = L$. In fact, for several applications we only need to prove the containment $IJ \subseteq L$. 

    Suppose that $I = \langle u_1, \ldots, u_{m}\rangle$,  $J = \langle v_1, \ldots, v_{n}\rangle$, and $L = \langle w_1, \ldots, w_{s}\rangle $. We have that $IJ = \langle u_1v_1, \ldots, u_1v_{n}, \ldots, u_{m}v_1, \ldots, u_{m}v_{n}\rangle$, and thus proving $IJ \subseteq L$ reduces to certifying ideal containment as described in Section \ref{idealeq}. In total, this requires $mn + mns$ multiplications in $S$. We implement the certificate in Lean as the \href{https://github.com/alainchmt/CertifyingInvariantsNF/blob/v1/IdealArithmetic/IdealArithmetic/IdealArithmetic.lean#L563}{structure} 
	\begin{lstlisting}
structure IdealMulLeCertificate' {R : Type*} [CommRing R] {r n m s : ℕ} (T : Fin r → Fin r → List R)
  (u : Fin m → Fin r → R) (v : Fin n → Fin r → R) (w : Fin s → Fin r → R)  where ...
	\end{lstlisting}
The functions \lstinline|u|, \lstinline|v|, and \lstinline|w| encode the coordinates of the $S$-generators of $I$, $J$, and $I J$, respectively. 
The coordinates of the $u_iv_j$ are included as fields, together with the proofs of the corresponding identities. From a term of this type, and under the assumption that $T$ is a multiplication table for $S$ and \lstinline|u|, \lstinline|v|, and \lstinline|w| encode $S$-generators for $I$, $J$, and $L$, respectively, we obtain a proof of $I  J \subseteq L $. We also define an alternative version \href{https://github.com/alainchmt/CertifyingInvariantsNF/blob/v1/IdealArithmetic/IdealArithmetic/IdealArithmetic.lean#L467}{\lstinline|IdealMulLeCertificate|}, where these assumptions are incorporated into the structure. 
	
	For repeated ideal multiplication, we want to prevent an exponential blow-up in the number of multiplications in $S$. To do this, we think of a goal like $I_0\dots I_k \subseteq J$, as a path starting from $I_0$ and ending at $J$. We proceed iteratively: first prove $I_0 I_1 \subseteq J_1$, then $J_1I_2 \subseteq J_2$, and continue until $J_{k-1}I_k \subseteq J$, while choosing a generating set with as few generators as possible for each intermediate ideal $J_i$. 
    Following this recipe requires chaining multiple multiplication containment certificates. To do this, we \href{https://github.com/alainchmt/CertifyingInvariantsNF/blob/v1/IdealArithmetic/IdealArithmetic/IdealArithmetic.lean#L630}{introduce} an inductive type defined as
	
	\begin{lstlisting}
inductive IdealMulLeChainCertificate {R : Type*} [CommRing R] {r m : ℕ}
   (T : Fin r → Fin r → List R) (u : Fin m → Fin r → R) : {s : ℕ} → (Fin s → Fin r → R) → Type _
  | nil : IdealMulLeChainCertificate T u u
  | cons {v} {uu} {w} : IdealMulLeChainCertificate T u v →
    IdealMulLeCertificate' T v uu w → IdealMulLeChainCertificate T u w
	\end{lstlisting}
A term \lstinline|C : IdealMulLeChainCertificate T u w| is a chain of multiplication containment certificates starting at \lstinline|u| and ending at \lstinline|w|, with no reference to the algebra $S$. Given a term \lstinline|TimesTable (Fin r) R S|, the function \href{https://github.com/alainchmt/CertifyingInvariantsNF/blob/v1/IdealArithmetic/IdealArithmetic/IdealArithmetic.lean#L646}{\lstinline|IdealMulLeChainCertificate.Ideals|} maps \lstinline|C| to the list of corresponding intermediate ideals. We \href{https://github.com/alainchmt/CertifyingInvariantsNF/blob/v1/IdealArithmetic/IdealArithmetic/IdealArithmetic.lean#L671}{prove} that the product of the ideals in this list is contained in the ideal with generators given by \lstinline|w|. 
	
We define analogous structures \href{https://github.com/alainchmt/CertifyingInvariantsNF/blob/v1/IdealArithmetic/IdealArithmetic/IdealArithmetic.lean#L552}{\lstinline|IdealMulEqCertificate'|} and \href{https://github.com/alainchmt/CertifyingInvariantsNF/blob/v1/IdealArithmetic/IdealArithmetic/IdealArithmetic.lean#L481}{\lstinline|IdealMulEqCertificate|} for certifying equality of ideals, where the latter bundles the encoded generators with their correspondence to ideals in $S$. Proofs of goals such as $I_0 \dots I_k = J$ are again done iteratively: $I_0I_1 = J_1$, then $J_1I_2 = J_2$, etc.

\begin{remark}
In CASs, ideals as usually stored via $\mathbb{Z}$-bases, partly because finding a two-generator representation can be costly. To compute $IJ$, one multiplies the basis elements of $I$ and $J$, giving a spanning set with $r^2$ elements, from which a $\mathbb{Z}$-basis is extracted. 
Thus, certifying $IJ=L$ using $\mathbb{Z}$-bases would require a number of multiplications in $\mathcal{O}_K$ that grows with $r$. However, our approach (since every ideal in $\mathcal{O}_K$ can be described with at most two generators) requires a fixed number of multiplications in $\mathcal{O}_K$ independent of $r$. This is advantageous as the degree of the number field increases.

\end{remark}

	\subsection{Primality certificate} \label{sec : primality}

	Determining the prime ideals of $\mathcal{O}_K$ of norm less than a given bound is crucial for the certification of the class group. Thus, given a prime ideal $I$, it is useful to be able to certify in Lean that $I$ is indeed prime, i.e., to construct a proof of \lstinline|Ideal.IsPrime I|. 
    
    Here, we take $R = \mathbb{Z}$ and assume that $S$ is an integral domain. Suppose that $I$ is a nonzero prime ideal of $S$. It can be shown that $S/I$ is finite, and thus 
    a finite field. Hence, $S/I$ is a finite extension of $\mathbb{F}_p$ for some prime number $p$. This implies that $\# S/I = p^n$ for some $n>0$ and $p \in I$. The extension $ (S/I) / \mathbb{F}_p $ is simple, so there exists an $a \in S / I$ with minimal polynomial $f \in \mathbb{F}_p[X]$ such that $S/I = \mathbb{F}_p(a) \cong \mathbb{F}[X]/(f)$, where $\deg f = n$. This suggests the following certificate. 
\newline 

	\paragraph{Primality certificate for ideals}
	\leavevmode
	\begin{multicols}{2}
		\begin{itemize}
			\item A natural number $n > 0$; 
			\item a prime number $p$; 
			\item a monic polynomial $f \in \mathbb{Z}[X]$ of degree $n$; 
			\item an element $a \in S$.
		\end{itemize}
	\end{multicols}

	Let $\bar{f}$ be the reduction of $f$ modulo $p$. Verification of the certificate amounts to checking that
	\begin{multicols}{2}
		\begin{enumerate}[label=(\roman*)]
			\item $p \in I$; 
			\item $\#(S / I) = p ^ n$; 
			\item $\bar{f}$ is irreducible in $\mathbb{F}_p[X]$; 
			\item $f(a) \in I$. 
		\end{enumerate}
	\end{multicols}
	
	The first two conditions guarantee that $S/I$ is an $\mathbb{F}_p$-algebra of dimension $n$. The condition $f(a) \in I$ implies that $\bar{f}(\bar{a}) = 0$ in $S/I$, with $\bar{a}$ being the residue class of $a$ in $S/I$. From $\bar{f}$ being irreducible in $\mathbb{F}_p[X]$, we conclude that $\bar{f}$ is the minimal polynomial of $\bar{a}$, and thus $\mathbb{F}_p(\bar{a})$ is a finite field extension of $\mathbb{F}_p$ of degree $n$. We have the inclusions $\mathbb{F}_p \subseteq \mathbb{F}_p(\bar{a}) \subseteq S/I$. Since both  $\mathbb{F}_p(\bar{a})$ and $S/I$ are of dimension $n$ as $\mathbb{F}_p$-vector spaces, they are equal. This proves that $S/I$ is a field and thus $I$ is a prime ideal. 
	
	To implement this in Lean, we define a \href{https://github.com/alainchmt/CertifyingInvariantsNF/blob/v1/IdealArithmetic/IdealArithmetic/CertifyPrimeIdeal.lean#L342}{structure} that takes as parameters a multiplication table \lstinline|T|, a collection \lstinline|v| corresponding to the $S$-generators of $I$, a collection \lstinline|w| corresponding to a $\mathbb{Z}$-basis of $I$, a term  \lstinline|A : IdealEqSpanCertificate' T v w| certifying this relation, and the natural number \lstinline|p|,

	\begin{lstlisting}
structure CertifiedPrimeIdeal' {r m : ℕ} [NeZero r] {T : Fin r → Fin r → List ℤ}
 {v : Fin m → Fin r → ℤ} {w : Fin r → Fin r → ℤ} (A : IdealEqSpanCertificate' T v w) (p : ℕ) ...
	\end{lstlisting} 

	This structure includes purely computational content, and so it refers exclusively to tuples and lists over $\mathbb{Z}$ and \lstinline|ZMod p|, the type of integers modulo $p$. 
 Among the fields of this structure, we include the nonzero \lstinline|n : ℕ|, the polynomial \lstinline|f : List ℤ| given as a list of coefficients, and the element \lstinline|a : Fin r → ℤ| given as a tuple of coordinates. Proofs of the verification statements are also included. 
	The proof of the condition $\#(S / I) = p ^ n$ is expressed in terms \lstinline|w|, representing the $\mathbb{Z}$-basis of $I$, as described in Section \ref{sec : certnorm}. For the third verification statement, we rely on an irreducibility certificate based on Rabin's test and implemented in Lean for polynomials over $\mathbb{F}_p$ described in \cite[Section 4]{rings}. For the fourth verification statement, a field \lstinline|c : Fin r → ℤ| representing the coordinates of $f(a)$ is included, together with a proof that this tuple corresponds to $f(a)$. This is done using list-based polynomial evaluation via the multiplication table. The membership statement (iv) is incorporated by including the coefficients of the linear combination of tuples in \lstinline|w| that equals \lstinline|c|. These proof obligations can be discharged automatically with \lstinline|decide|. 
    Under suitable assumptions, a term of type \lstinline|CertifiedPrimeIdeal' A p| proves the primality of the ideal generated by $v$. 

	\begin{example}
		Consider the ideal $I$ from Example \ref{ex:norm}. We prove that it is prime using the certificate. Take $n = 2$, the prime number $p=3$, the element in $\mathcal{O}$
        given by $a = -11b_1 + b_2 -10b_3 - 16b_4 - 45b_5$, and the polynomial $f = X^2 + 1$ in $\mathbb{Z}[X]$.
		We easily see that $3 \in I$. Furthermore, as computed in Example \ref{ex:norm}, $\# \mathcal{O} / I =  3 ^ 2$. So the first two verification statements are satisfied.
		Furthermore, one finds $f(a) = 1598b_1 -1642b_2 + 1331b_3 + 3328b_4 + 9954b_5$. We certify that $f(a) \in I$ by noticing that
		\begin{align*}
			\begin{pmatrix}
				-1464 & -5862 &  1331 & 3328 & 9954
			\end{pmatrix}  \cdot
			\begin{bmatrix}
				3 & 0 & 0& 0& 0 \\
				0 &3 &0 &0& 0 \\
				2 &2& 1 &0 &0 \\
				1 &1 &0& 1 &0 \\
				0& 1& 0 &0 &1
			\end{bmatrix}  = \begin{pmatrix}
				1598 & -1642 & 1331 & 3328 & 9954
			\end{pmatrix}.
		\end{align*}
		Finally, the polynomial $X^2 + 1 \pmod 3$ can be shown to be irreducible using Rabin's test or by noticing that it has no roots in $\mathbb{F}_3$.  We conclude that $I$ is a prime ideal of $\mathcal{O}$. 
        In Lean, assuming that \lstinline|v| represents $S$-generators of $I$ and \lstinline|w| a $\mathbb{Z}$-basis, we make the term
        \href{https://github.com/alainchmt/CertifyingInvariantsNF/blob/v1/IdealArithmetic/Examples/Paper/NormExample.lean#L14}{\lstinline|A : IdealEqSpanCertificate' T v w|}, and use it to construct \href{https://github.com/alainchmt/CertifyingInvariantsNF/blob/v1/IdealArithmetic/Examples/Paper/NormExample.lean#L58}{\lstinline[mathescape]|CP : CertifiedPrimeIdeal' A 3|} which allows us to prove \href{https://github.com/alainchmt/CertifyingInvariantsNF/blob/v1/IdealArithmetic/Examples/Paper/NormExample.lean#L77}{\lstinline| I.IsPrime|}.
	\end{example}

    \begin{remark}
    With $S$ as above (e.g. $S = \mathcal{O}_K$), this primality certificate always exists for any nonzero prime ideal of $S$. In the case that $\#(S/I) = p$, it follows trivially that $I$ is prime by taking $f = X$. 
    \end{remark}

	\begin{remark}\label{remarkkummer}
    In practice, one usually wants to find the prime ideals above a prime number $p$. For a number field $K = \mathbb{Q}(\alpha)$, with $T \in \mathbb{Z}[X]$ the minimal polynomial of $\alpha$ and with $p \nmid [\mathcal{O}_K : \mathbb{Z}[\alpha]]$, the Kummer-Dedekind theorem \cite[Theorem 4.8.13.]{Cohen} states a bijection between the irreducible factors $\bar{g}$ of $\bar{T} \in \mathbb{F}_p[X]$ and prime ideals above $p$, mapping $\bar{g} \mapsto \langle p, g(\alpha)\rangle$, whose norm is $p ^ {\deg g}$.
    This could be used to certify primality. Making the Mathlib \href{https://github.com/leanprover-community/mathlib4/blob/6f1c6456ee863edbcf96c4febc52cd1c8d07487f/Mathlib/NumberTheory/KummerDedekind.lean#L134-L176}{version} of this theorem fully explicit, suitable for computation, and integrated into our framework is part of ongoing work. For primes dividing $[\mathcal{O}_K : \mathbb{Z}[\alpha]]$, Kummer-Dedekind no longer applies and more sophisticated algorithms are required to compute the prime ideals above $p$ \cite[Section 6.2]{Cohen}. In that case, our primality certificate provides a more efficient alternative.

	\end{remark}

\begin{remark}
 We apply the certificates in this section to ideals of $\mathcal{O}_K$. To obtain the witness data, we implemented the necessary routines in SageMath, relying heavily on linear algebra over $\mathbb{Z}$. 
    \end{remark}

	\section{Efficient discriminant computation}\label{sec:disc}
	The discriminant $\operatorname{disc}(K) \in \mathbb{Z}$ of a number field $K$ is an important invariant that measures its arithmetic complexity. In previous work \cite{rings}, discriminants of a few degree-3 number fields were verified, but higher degrees were infeasible. 
    In this section, we formalize a more efficient procedure to verify this quantity, allowing the certification of discriminants for considerably higher degrees. 

	Suppose $(\omega_1, \ldots, \omega_n)$ is an integral basis of $\mathcal{O}_K$. Given an element $x \in \mathcal{O}_K$, the $\mathbb{Z}$-linear map $y \mapsto xy$ can be represented as a $n \times n$ matrix with integer entries. The trace of this matrix (independent of the choice of integral basis) is denoted by $\operatorname{Tr}(x)$ and is known as the trace of $x$ in $\mathcal{O}_K$. The discriminant $\operatorname{disc}(K)$ is defined as the determinant of the so-called trace matrix, given by 
	\begin{align*}
		\operatorname{disc}(K) := \det (\operatorname{Tr}(\omega_i\omega_j))_{ij}.
	\end{align*}

	This definition is in Mathlib as \href{https://github.com/leanprover-community/mathlib4/blob/6f1c6456ee863edbcf96c4febc52cd1c8d07487f/Mathlib/NumberTheory/NumberField/Discriminant/Defs.lean#L35-L36}{\lstinline|NumberField.discr K|}. However, computing $\operatorname{disc}(K)$ directly from its definition is inefficient. Following \cite{rings}, where the formula below was formalized, we use the fact that if $K = \mathbb{Q}(\alpha)$, with $\alpha$ a root of a monic irreducible polynomial $T \in \mathbb{Z}[X]$, then $\operatorname{disc}(K)$ can be found using
	\begin{align}\label{formuladisc}
		\operatorname{disc}(T) = [\mathcal{O}_K : \mathbb{Z}[\alpha]]  ^ 2 \operatorname{disc}(K),
	\end{align}
where $\operatorname{disc}(T)$ denotes the discriminant of the polynomial $T$. 
	In previous work \cite[Section 6]{rings}, the definition of the discriminant of a polynomial over any commutative ring $R$ via resultants was formalized in Lean. 
    These definitions have been refactored and added to Mathlib by Andrew Yang, Kenny Lau, and Anne Baanen. We use the Mathlib \href{https://github.com/leanprover-community/mathlib4/blob/6f1c6456ee863edbcf96c4febc52cd1c8d07487f/Mathlib/RingTheory/Polynomial/Resultant/Basic.lean}{definitions} for this formalization.

	The resultant $\operatorname{Res}(G,H)$ of two polynomials $G,H$ with coefficients over a commutative ring $R$ and of degree $m$ and $n$, respectively, is the determinant of an $(m + n) \times (m + n)$ matrix over $R$ called the Sylvester matrix, formed from the coefficients of $G$ and $H$. When $R$ is an integral domain and $G$ and $H$ are monic, one has $\operatorname{Res}(G,H) = \prod_{i,j} (\alpha_i - \beta_j)$, where the $\alpha_i$ and $\beta_j$ are the roots of $G$ and $H$, respectively, counted with multiplicity in an algebraically closed field. For a monic polynomial $P$, we have that $\operatorname{disc}(P) = (-1)^{n(n-1)/2}\operatorname{Res}(P,P')$, with $P'$ the formal derivative of $P$.

	In principle, since the determinant of an explicitly given matrix with integer entries is computable in Lean, the discriminant of an explicit polynomial with integer coefficients is computable.
	However, there is a major bottleneck that quickly leads to stack overflow. The issue arises because, in Mathlib, the determinant is defined via the Leibniz formula, which is not appropriate for efficient computation. As a consequence, the verifications of discriminants of number fields in \cite{rings} was limited to degree three.
    To overcome this, we replace the determinant-based computation with a more efficient approach. 
    
	Similar to how one computes the GCD of two polynomials using the division algorithm, it is possible to use a pseudo-remainder sequence to efficiently compute their resultant \cite[Algorithm 3.3.1]{Cohen}. Let $R$ be an integral domain and $p$ and $q$ polynomials in $R[X]$.
	\begin{definition}\label{def-psuedo}
		A pseudo-remainder sequence starting at $p$ and $q$ is a sequence of polynomials $P_0$, $P_1$, $\ldots$, $P_{k+1}$ of decreasing degree such that $P_0 = p$, $P_1 = q$, $P_{k+1}$ is a constant polynomial, and for every $0 \leq i < k$ there exists nonzero $e_i, f_i$ in $R$ and $Q_i$ in $R[X]$ such that
		\begin{align}\label{pseudo-rem}
			e_i P_i = Q_i  P_{i+1} - f_i  P_{i+2}.
		\end{align}
	\end{definition}

	A sequence of this type can always be constructed starting from any pair of polynomials. Pseudo-remainder sequences have been formalized in proof assistants before. In Coq/Rocq \cite{assia} they are defined similarly to Definition \ref{def-psuedo} but with $f_i = -1$. Sequences of this form are also used in the formalization of the subresultant algorithm in Isabelle/HOL in \cite{subresultant}. However, the construction in our formalization, which will be discussed in Section \ref{sec : sturm_pseudo}, differs from these by allowing flexibility in both families of constants $f_i$ and $e_i$, as does Cohen in \cite{Cohen}. The reason is that this allows us to restrict the computations to the ring $R$ instead of passing to its fraction field, while also controlling for coefficient growth.

	It is possible to show that if $P_{i+2} \neq 0$ (and thus $P_i \neq 0$ and $P_{i+1} \neq 0$), then
	\begin{align*}
		e_i^{\deg{P_{i+1}}} & \operatorname{Res}(P_i, P_{i+1}) =  \\
		& (-1) ^ {(\deg{P_i} + 1) \deg{P_{i+1}}}  f_i ^ {\deg{P_{i+1}}}
		l(P_{i+1}) ^ {\deg{P_i} - \deg{P_{i+2}}}  \operatorname{Res}(P_{i+1}, P_{i+2}),
	\end{align*}
	where $0 \leq i < k$ and $l(P_{i+1})$ is the leading coefficient of $P_{i+1}$.

	We prove this in Lean in three stages. We first show the equivalent equality over a field $F$, dividing both sides by $e_i^{\deg{P_{i+1}}}$ and assuming that $P_i, P_{i+1}$, and $P_{i+2}$ split completely in $F$. This relies on the fact that $\operatorname{Res}(P, \prod_{j=1}  ^ s (X - c_j)) = (-1)^{s \deg P} \prod_j  P(c_j)$ and $\operatorname{Res}( \prod_{j=1}  ^ t (X - u_j), Q) =  \prod_j Q(u_j)$. Second, we extend the result to polynomials over any field (where they do not necessarily split) by passing to an algebraic closure and applying the first result. Finally, we prove the above equation over any integral domain $R$ by embedding $R$ in its fraction field.

	If $P_{k+1} = 0$, then $p$ and $q$ have a common root (in the algebraic closure of the fraction field of $R$) and so $\operatorname{Res}(p,q) =0$. If $P_{k+1} = c \neq 0$, we can use the previous recursive relation to conclude that
	\begin{align}\label{resformula}
		\operatorname{Res}(p,q) = \frac{\left(\prod_{i = 0} ^ {k-1} (-1) ^ {(\deg{P_i} + 1)\deg{P_{i+1}}} f_i ^ {\deg{P_{i+1}}} l(P_{i+1}) ^ {\deg{P_i} - \deg{P_{i+2}}}\right)} {\prod_{i = 0} ^ {k-1} e_i ^ {\deg{P_{i+1}}}} c ^ {\deg{P_k}}.
	\end{align}

	We formalize this result as \href{https://github.com/alainchmt/CertifyingInvariantsNF/blob/v1/IdealArithmetic/Signature/ResultantRecursive.lean#L249}{\lstinline|resultant_eq_div_of_premainder_sequence|}. To encode the sequence of polynomials, we use a list \lstinline|P : List R[X]. | Alternatively, we could have used a tuple \lstinline|P : Fin (k + 2) → R[X]|. However, lists are better suited for the inductive arguments required in the proofs and perform well in practice for computations.
	Similarly, we encode the constants \lstinline|e : List R| and \lstinline|f : List R| as lists over $R$. Working with lists in Lean requires us to constantly include bounds on the indices in the assumptions to avoid index out-of-bound errors. The pseudo-remainder sequence condition is
	\begin{lstlisting}
∀ i, ∀ h : i + 2 < P.length , C (e[i]) * P[i] = Q[i] * P[i + 1] - C (f[i]) * P[i + 2]
	\end{lstlisting}
	After including the assumption \lstinline|h0 :  P[0] = p| and \lstinline|h1 : P[1] = q|, we express the fact that the last entry in $P$ is a nonzero constant by including \lstinline|hcl : P.getLastD 0 = C c| and \lstinline|hcz : c ≠ 0|.

    Together with the appropriate conditions on the degrees of the polynomials and on the lengths of the lists to avoid index out-of-bounds errors, we prove that \lstinline|resultant p q| equals the right-hand-side of \eqref{resformula}. For example, the numerator of the fraction is written in Lean as:

	\begin{lstlisting}
		((∏ i : Fin (P.length - 2), ((-1) ^ ((P[i].natDegree + 1) * P[↑i + 1].natDegree) *
		f[i] ^ (P[↑i + 1].natDegree) * P[↑i + 1].leadingCoeff ^ (P[i].natDegree - P[↑i + 2].natDegree)))
	\end{lstlisting}

	The bookkeeping provided by the proof assistant was particularly valuable in the formalization of this formula, as it helped us catch indexing errors that we had not noticed on paper.

	To actually compute with this formula, coefficient and degree extraction in polynomials must be computable operations. We prove an alternative version of the previous statements using a computable representation of polynomials (as described in Section \ref{sec : model}) via lists of coefficients. The pseudo-remainder sequence is then encoded as a list of lists \lstinline|P : List (List R)| .

	Our construction of pseudo-remainder sequences in Lean will be discussed in Section \ref{sec : sturm_pseudo}. Such sequences have applications beyond resultant computations. In particular, we will use them to count the number of real roots of a polynomial with real coefficients, employing a theorem due to Sturm. The discriminant computation then arises as a simple by-product. This lets us certify polynomial discriminants, and consequently number field discriminants, of substantially higher degree than in previous work (see Section \ref{sec : lmfdb}), including of degree 10. 

	\section{Certifying the signature of a number field} \label{sec : sign}
	In order to certify the class group of a number field $K$, we will need some information about the unit group of $\mathcal{O}_K$. In 1846, Dirichlet described its structure as
	\begin{align}\label{eq : dirichlet}
		\mathcal{O}_K^{\times} \cong \mu(\mathcal{O}_K) \times \mathbb{Z}^ {r_1 + r_2 - 1}, 
	\end{align}
	where $\mu(\mathcal{O}_K) = \{x \in \mathcal{O}_K^{\times} \mid \exists n \in \mathbb{Z}_{\geq 1}, x ^ n = 1 \}$ is the finite and cyclic group of roots of unity in $\mathcal{O}_K$, and  $(r_1, r_2)$ is the signature of the number field.  Dirichlet's unit theorem \cite[Theorem 4.9.5]{Cohen}, which \href{https://github.com/leanprover-community/mathlib4/blob/6f1c6456ee863edbcf96c4febc52cd1c8d07487f/Mathlib/NumberTheory/NumberField/Units/DirichletTheorem.lean#L457-L458}{exists} in Mathlib, shows that the signature determines the rank of the free part of the unit group $\mathcal{O}_K^{\times}$.

	Suppose $K = \mathbb{Q}(\alpha)$, where $\alpha$ has minimal polynomial $T \in \mathbb{Q}[X]$, so that $K \cong \mathbb{Q}[X]/(T)$. If $\xi \in \mathbb{C}$ is a root of $T$, then there is a unique ring homomorphism $K\to \mathbb{C}$ sending $\alpha \mapsto \xi$. This is a real embedding if $\xi \in \mathbb{R}$,  and a non-real complex embedding otherwise. Every embedding $K \hookrightarrow \mathbb{C}$ arises in this form. Consequently, real embeddings are in bijection with the real roots of $T$, and the non-real complex embeddings are in bijection with the non-real roots of $T$.

	It follows that, in order to compute the signature of $K$, it suffices to determine the number of real roots of a defining polynomial $T$. It turns out that from a pseudo-remainder sequence such as \eqref{pseudo-rem} starting with $T$ and $T'$ with positive constants $e_i$ and $f_i$, it is possible to extract two sequences of integers whose difference in the number of changes in sign gives the total number of real roots of $T$. This is a version of Sturm's theorem, which we formalize in Lean in the next section.

	\subsection{Sturm's theorem} \label{sec-sturm}
	In the 19th century, C. Sturm published a theorem for counting the number of real roots of a polynomial with real coefficients \cite{sturm}. The version we are interested in is the following, appearing in \cite[Theorem 4.1.10]{Cohen}. 

	\begin{theorem}\label{Sturm}
		Let $T \in \mathbb{R}[X]$ be a square-free polynomial. Let $P_0,\ldots, P_{k+1}$ be a pseudo-remainder sequence starting with $T$ and $T'$, where $e_i > 0$ and $f_i > 0$ for all $i=0, \ldots, k-1$. Let $l_i = l(P_i)$ be the leading coefficient of $P_i$ and $d_i = \deg(P_i)$. If $s$ is the number of sign changes in the sequence $[l_0, l_1, \ldots, l_{k+1}]$ and $t$ is the number of sign changes in $[(-1)^{d_0}l_0, \ldots, (-1)^{d_{k+1}}l_{k+1}]$, then $t - s$ equals the total number of real roots of $T$.
	\end{theorem}

	Note that since $T$ is square-free, its roots are necessarily distinct. Theorem \ref{Sturm} is a corollary of a more general statement: if $s(a)$ denotes the number of sign changes, ignoring zeros, in the sequence $[P_0(a), \ldots, P_{k+1}(a)]$, and assuming $T(a) \neq 0$ and $T(b) \neq 0$, then the number of real roots in the interval $[a,b]$ is given by $s(a) -s(b)$. This is the classical interval version of Sturm's theorem.

	The proof of this theorem relies on fundamental properties of the real numbers which hold in the more general setting of \emph{real closed fields}. In Section \ref{realclosed}, we define these fields in Lean by following \cite{assia}, and formalize the results needed to prove Theorem \ref{Sturm}. Concurrently with our work, an alternative definition of real closed fields, now in Mathlib, was developed independently. However, the results we derive, including Rolle's theorem and the mean value theorem for polynomials over real closed fields, were not previously available in Lean, to the best of our knowledge. 
    We discuss this and other related work, and its relation to ours, in Section \ref{sturmrel}. In Section \ref{sec : sturm_pseudo} we formalize Sturm sequences over an arbitrary linearly ordered commutative ring, using a more general definition than those in previous formalizations we are aware of, which is suitable for computation, and will allow us to keep all arithmetic in $\mathbb{Z}$ and control coefficient explosion. We then describe our formalization of Theorem \ref{Sturm} using these sequences.

	\subsubsection{Real closed fields}\label{realclosed}

	In \cite{assia}, Cohen and Mahboubi defined real closed fields in Coq/Rocq and formalized several results in real algebraic geometry, including the mean value theorem, Cauchy indices, and Tarski-queries, from which the interval version of Sturm's theorem could, in principle, be derived as a corollary. Formalizing Cauchy indices and Tarski-queries was beyond the scope of this project, although we draw substantial inspiration from their work and their main reference \cite{realbook}. 

	There are many equivalent characterizations of a real closed field. In our formalization, we adopt the one in \cite{assia}, which is the definition that most directly yields the properties that we will need later on. An ordered field is a field $F$ equipped with a linear order $\leq$ such that $a \leq b$ implies $a + c \leq b + c$ and $a <b$ implies $a c < b c$ whenever $c>0$. We express this in Lean with \lstinline|(F : Type*) [Field F]|, together with the instance \lstinline|[LinearOrder F]| providing the total order relation, and the typeclass \lstinline|[IsStrictOrderedRing F]| stating that multiplication by a positive constant and addition are strictly monotone operations. A real closed field is defined as an ordered field satisfying the intermediate value theorem for polynomial functions: for every polynomial $P \in F[X]$ and $a,b \in F$ with $a \leq b$ and $t \in F$ such that $P(a) < t < P(b)$ or $P(b) < t < P(a)$, there exists $c \in (a,b)$ with $P(c) =t$. In Lean, we \href{https://github.com/alainchmt/CertifyingInvariantsNF/blob/v1/IdealArithmetic/Signature/RealClosedField.lean#L38}{define}

	\begin{lstlisting}
def IsRealClosedField (F : Type*) [Field F] [LinearOrder F] [IsStrictOrderedRing F] : Prop :=
   ∀ {a b t : F} , ∀ {P : F[X]},
    a ≤ b → t ∈ Set.Ioo (P.eval a) (P.eval b) → ∃ s, s ∈ Set.Ioo a b ∧ P.eval s = t
	\end{lstlisting}
	The term \lstinline|Set.Ioo (P.eval a) (P.eval b)| denotes the open interval $(P(a), P(b))$, which is empty if $P(a) \geq P(b)$. The term \lstinline|Set.Ioo a b| denotes the interval $(a,b)$. It is unnecessary to include a separate condition with \lstinline|Set.Ioo (P.eval b) (P.eval a)| as this follows from the above one by replacing $P$ with $-P$.

	A proof of \lstinline|IsRealClosedField  ℝ| is obtained directly from the intermediate value theorem available in Mathlib for complete linear orders and the continuity of polynomial functions in a topological ring. Other examples of real closed fields include the algebraic real numbers $\overline{\mathbb{Q}} \cap \mathbb{R}$ and the hyperreal numbers, although we did not formalize this.

	A key ingredient for proving Sturm's theorem is
	the mean value theorem for polynomials over a real closed field. It states that for every $a,b \in F$ with $a < b$, there is $c \in (a,b)$ such that $P(b) - P(a) = P'(c)(b - a)$. \href{https://github.com/alainchmt/CertifyingInvariantsNF/blob/v1/IdealArithmetic/Signature/RealClosedField.lean#L309}{Formally},
	\begin{lstlisting}
theorem mean_value_theorem  (hc : IsRealClosedField F) {a b : F} {P : F[X]} (hab : a < b) :
   ∃ c ∈ Ioo a b , P.eval b - P.eval a = ((derivative P).eval c) * (b - a) := by ...
	\end{lstlisting}
	To prove this, we followed the strategy outlined in \cite{assia}. First proving Rolle's theorem in real closed fields by induction on the maximal number of roots in a given interval. Formulating the precise statement for the induction hypothesis is a subtle matter, and a crucial observation is that it must hold for arbitrary intervals rather than just the one we start with. We refer to \cite[pag. 17]{assia} for details.

	\subsubsection{Related work on Sturm's theorem and real closed fields} \label{sturmrel}
	Sturm's theorem has an extensive history of mechanization. It was first formalized by Harrison in HOL in 1997 for squarefree polynomials \cite{harrison}. Manuel Eberl completed a formal proof in Isabelle/HOL in 2014 \cite{eberl}, and Narkawicz, Muñoz, and Dutle provided another formalization in 2015 in the PVS proof assistant \cite{nasa}. All of these developments concern polynomials over the field of real numbers, rather than the more general framework of a real closed field. As previously mentioned, Cohen and Mahboubi formalized results over real closed fields in Coq/Rocq from which Sturm's theorem could be easily derived. 
	Parallel to our development, a \href{https://github.com/artie2000/real\_closed\_field}{repository} dedicated to real closed fields in Lean, with an alternative \href{https://github.com/leanprover-community/mathlib4/blob/6f1c6456ee863edbcf96c4febc52cd1c8d07487f/Mathlib/FieldTheory/IsRealClosed/Basic.lean#L42-L50}{definition} now in Mathlib, started to be developed by Artie Khovanov. Some of our formalized results, such as the mean value theorem, are now part of this repository. Also in parallel to our development, a version of Sturm's theorem, proven using Cauchy indices, was formalized by Pedro Saccomani, Sarah Pereira, and Tomaz Mascarenhas in Lean. However, their \href{https://github.com/artie2000/real_closed_field/blob/fc1481794180f95657e77fd70f2faa8f86294ec4/RealClosedField/SturmTarski/SturmSeq.lean#L17}{ definition} of a so-called \emph{Sturm sequence} is less general and is not suitable for our computational purposes. In the next section, we use a more general definition via pseudo-remainder sequences over a linearly ordered commutative ring, not necessarily a real closed~field. 
\subsubsection{Sturm sequences}\label{sec : sturm_pseudo}
	The finite sequence of polynomials appearing in Sturm's theorem is called a \emph{Sturm sequence}. 
	Let $R$ be a linearly ordered commutative ring. We say that a pseudo-remainder sequence (see Definition \ref{pseudo-rem}) starting from $p \in R[X]$ and $q \in R[X]$ is a \emph{Sturm sequence} if, additionally, all the constants $e_i, f_i$ are positive and the final polynomial $P_{k+1}$ is a nonzero constant. In Lean, we encode such a sequence as a list of polynomials \lstinline|P : List R[X]| together with the distinguished starting polynomials \lstinline|p q : R[X]|, and define the predicate  \href{https://github.com/alainchmt/CertifyingInvariantsNF/blob/v1/IdealArithmetic/Signature/Sturm.lean#L122}{\lstinline|IsSturmSequence P p q|} as
	\begin{lstlisting}
structure IsSturmSequence [LinearOrder R] (P : List R[X]) (p q : R[X])  where
   hlen : 2 ≤ P.length
   h0 : P[0] = p
   h1 : P[1] = q
   hc : ∃ c : R, c ≠ 0 ∧ P.getLastD 0 = C c
   hmono : ∀ i, ∀ h : i + 1 < P.length , P[i + 1].natDegree < P[i].natDegree
   hrem : ∀ i, ∀ h2 : i + 2 < P.length ,
    (∃ e f : R, ∃ Q : R[X], 0 < e ∧ 0 < f ∧ C e * P[i] = Q * P[i + 1] - C f * P[i + 2] )
	\end{lstlisting}
	The polynomials $p$ and $q$ are treated separately from the sequence for easier manipulation. We recall that our definition includes both sets of constants $e_i$ and $f_i$, allowing us to restrict computations to the ring $R$ and control coefficient growth. This differs from other formalizations of Sturm sequences. Harrison employs the relation $P_i = Q_i P_{i+1} - f_iP_{i+2}$ instead, omitting the constant $e_i$, and formulates it only for polynomials over the real numbers. Eberl, when defining an explicit Sturm sequence \cite{eberl}, uses the standard remainder sequence  $P_i = Q_i P_{i+1} - P_{i+2}$. Saccomani, Pereira, and Mascarenhas defined the sequence for polynomials over real closed fields using again the standard remainder sequence. 

	For a linearly ordered \lstinline|R : Type*| equipped with a zero and a list \lstinline|L : List R|, we define \lstinline|signChanges L| by first removing all zeros entries of \lstinline|L| and then counting the number of times a positive element is followed by a negative one and vice versa. 
The number of sign changes of a polynomial sequence $P_i$ at $a \in R$ is denoted by $s(a)$ and is, \href{https://github.com/alainchmt/CertifyingInvariantsNF/blob/v1/IdealArithmetic/Signature/Sturm.lean#L95}{formally}, \lstinline|signChanges| applied to the list $[P_0(a), \ldots, P_{k+1}(a)]$:
	\begin{lstlisting}
def signChangesPolySeq (P : List R[X]) (a : R) : ℕ := signChanges (List.map (fun x => x.eval a) P)
	\end{lstlisting}

	We also define \href{https://github.com/alainchmt/CertifyingInvariantsNF/blob/v1/IdealArithmetic/Signature/Sturm.lean#L101}{\lstinline|signChangesInfty P|} as the number of sign changes  `at positive infinity', the number of sign changes in the sequence of leading coefficients $[l_0, \ldots, l_{k+1}]$. 
	Similarly, \href{https://github.com/alainchmt/CertifyingInvariantsNF/blob/v1/IdealArithmetic/Signature/Sturm.lean#L107}{ \lstinline|signChangesNInfty P|}, the number of sign changes `at negative infinity', is defined as the number of sign changes in $[(-1)^{d_0}l_0, \ldots, (-1)^{d_{k+1}}l_{k+1}]$.

	Let us now consider the case $R = F$, where $F$ is a real closed field and let $a,b \in F$ with $a \leq b$. We follow a strategy similar to \cite{harrison} to prove Sturm's theorem for intervals. We start with a Sturm's sequence $P_0, \ldots, P_{k+1}$ with $P_0 = p$ and $P_1 = p'$. The set of roots of the polynomials $P_i$ in $[a,b]$ is given by $S = \{\xi \in [a,b] \mid \prod_i P_i (\xi) = 0\}$. It is finite and we can order its elements as $a \leq \xi_0 < \xi_1 < \ldots < \xi_{N-1} \leq b$.

	Suppose that $[c,d] \subseteq [a,b]$ is a sub-interval containing exactly one of the roots $\xi \in S$. We formally prove that if $p(\xi) \neq 0$ then  $s(c) = s(d)$, and if $p(\xi) = 0$ then $s(c) = s(d) + 1$. This step relies on the mean value theorem. In our formalization, we additionally assumed that $\xi \neq c$ and $\xi \neq d$ to avoid some technical difficulties when dealing with lists containing zeros. This simplifies the argument and is sufficient for our purposes, but it is not strictly necessary. By induction on $N$, with the induction hypothesis universally quantified over every interval (as done in \cite{assia} for Rolle's theorem) and using the previous result, one can show that if, for every index $i$, $P_i(a) \neq 0$ and $P_i(b) \neq 0$, then the number of zeros of $p$ in the interval $[a,b]$ is equal to $s(a) - s(b)$.
	\begin{remark}
		The assumption that $a$ and $b$ are not roots of any of the polynomials in the Sturm sequence is not strictly necessary (it is enough to require that $p(a) \neq 0$ and $p(b) \neq 0$). The stronger condition arises as a consequence of the assumption described above for a sub-interval $[c,d]$. In any case, this does not affect our ability to prove Theorem \ref{Sturm}.
	\end{remark}

	The sign of a nonzero polynomial $P\in F[X]$ evaluated at a sufficiently large $a \in R$ is equal to the sign of the leading coefficient $l(P)$ of $P$, while the sign of $P(b)$ for a sufficiently small $b$ is equal to the sign of $(-1)^{\deg P} l(P)$. By choosing a large enough interval $[a,b]$ containing all the roots of $p$, we can formally \href{https://github.com/alainchmt/CertifyingInvariantsNF/blob/v1/IdealArithmetic/Signature/Sturm.lean#L1476}{prove} Theorem \ref{Sturm} for polynomials over real closed fields. 
	\begin{lstlisting}
theorem sturm_theorem_total  (hc : IsRealClosedField F) {P : List F[X]} {p : F[X]}
      (hs : IsSturmSequence P p (derivative p)) :
   #(Multiset.toFinset p.roots) = signChangesNInfty P - signChangesInfty P := ...
	\end{lstlisting}
	The term \lstinline|Multiset.toFinset p.roots| is the finite set of roots of $p$ in $F$, and $\#$ denotes its cardinality.

	\subsection{Computing the number of real roots}
	We recall that our goal is to count the number of real roots of a square-free polynomial $T \in \mathbb{Z}[X]$. However, note that we cannot use \lstinline|sturm_theorem_total| with $F = \mathbb{R}$ directly. First, it uses Mathlib polynomials, whose operations are \lstinline|noncomputable| regardless of the base ring. Second, the type of real numbers $\mathbb{R}$ does not have decidable equality. We would like to carry out all of the arithmetic operations in the type $\mathbb{Z}$. To achieve this, we prove a variant of \lstinline|sturm_theorem_total| in the following setting: $R$ is a linearly ordered commutative ring, $F$ is a real closed field, and $f : R \to F$ is a strictly monotone ring homomorphism. For \lstinline|p : R[X]| and \lstinline|P : List R[X]| with \lstinline|h : IsSturmSequence P p (derivative p)| (note that this is defined over the base ring), we prove that \href{https://github.com/alainchmt/CertifyingInvariantsNF/blob/v1/IdealArithmetic/Signature/Sturm.lean#L1568}{\lstinline|#(Multiset.toFinset (map f p).roots) = signChangesNInfty P - signChangesInfty P|}.

	The term \lstinline|map f p| is the polynomial in $F[X]$ obtained by applying $f$ to each of the coefficients of $p$. The Sturm sequence $P$ is entirely defined over $R$. By taking $R = \mathbb{Z}$,  $F = \mathbb{R}$, and $f : \mathbb{Z} \hookrightarrow \mathbb{R}$ we obtain a version that relies only on integer arithmetic.
	Furthermore, in order to build a Sturm sequence and compute the corresponding sign changes, we use the representation of polynomials as lists. We \href{https://github.com/alainchmt/CertifyingInvariantsNF/blob/v1/IdealArithmetic/Signature/Sturm.lean#L1585}{define}
	\begin{lstlisting}
structure SturmBuilderOfList (P : List (List R)) (p : List R) (q : List R) where ...
	\end{lstlisting}
	Here, $P$ represents a sequence of polynomials, with each polynomial being encoded as a list of coefficients, while $p$ and $q$ correspond to the initial polynomials in the sequence. This structure is used to produce a proof that $P$ is a Sturm sequence, while also carrying the data of the constants $e_i$ and $f_i$. Indeed, the fields in this structure include
	\begin{lstlisting}
		hmono :  ∀ i , ∀ h : i + 1 < P.length, P[i + 1].length < P[i].length
		e : List R
		f : List R
		Q : List (List R)
		hrem : ∀ i, ∀ h2 : i + 2 < P.length ,
		   P[i].mulPointwise e[i] = Q[i] * P[i + 1] - P[i + 2].mulPointwise (f[i])
	\end{lstlisting}
	Here, \lstinline|hmono| corresponds to the condition of decreasing degrees of the polynomials in the sequence. The fields \lstinline|e| and \lstinline|f| are lists of positive elements in $R$, corresponding to the constants $e_i, f_i$ in the definition of a pseudo-remainder sequence, and \lstinline|hrem| is the key pseudo-remainder relation between these polynomials written in terms of lists. When $R = \mathbb{Z}$, all the proof obligations can be discharged with \lstinline|decide|. From a term of this type, we construct a term of type \lstinline|IsSturmSequence (List.map (ofList) P) (ofList p) (ofList q)| where \lstinline|ofList p| is the term of type \lstinline|R[X]| with coefficients given by the list \lstinline|p|. Furthermore, since a term \lstinline|SturmBuilderOfList P p q| stores the constants $e_i$ and $f_i$, it can be directly used to compute $\operatorname{Res}(p,q)$ using the formula in \eqref{resformula}, and hence to determine the discriminant of polynomials.

	Finally, by defining analogues of \lstinline|signChangesNInfty| and \lstinline|signChangesInfty| for polynomials represented as lists, we are able to compute the number of real roots of a polynomial defined over $\mathbb{Z}$.

	\begin{example}\label{sturm}
		Consider the irreducible polynomial over the integers $T = X ^ 5 - 3 X ^ 3 + 9 X - 8$. Its derivative is $T' = 5X^4 - 9X^2 + 9$. To construct a pseudo-remainder sequence with $P_0 = T$ and $P_1 = T'$, we note that if $A,B \in \mathbb{Z}[X]$, with $\deg(B) \leq \deg(A)$, then the quotient and remainder of $l(B)^{\deg(A) - \deg(B) + 1} A$ after division by $B$ lie in $\mathbb{Z}[X]$. We can set $e_i = l(P_{i+1})^{\deg(P_i) - \deg(P_{i+1}) + 1}$ and let $R_{i}$ be the remainder of dividing $e_iP_{i}$ by $P_{i+1}$. To reduce the size of the coefficients, we can let $f_i$ be the gcd of the coefficients of $R_i$ and set $P_{i+2} = R_{i}/f_i$ (see \cite[Algorithm 4.1.11]{Cohen} for an alternative way of choosing $f_i$). 
		Using a CAS, we compute 
\[
\begin{aligned}[t]
P_0 &= T &\quad& P_1 = T' &\quad& P_2 = 3X^3 - 18X + 20 \\
P_3 &= -63X^2 + 100X - 27 &\quad& P_4 = 3103X - 4752 &\quad& P_5 = 1
\end{aligned}
\]

		With quotient polynomials $Q = (5X, 15X, -300 - 189X, 10924 - 195489 X)$ and tuples of positive constants $e = (25, 9, 3969, 9628609)$ and $f = (10, 3, 15, 208061595)$. 
	Using these data, we define in Lean a term \href{https://github.com/alainchmt/CertifyingInvariantsNF/blob/v1/IdealArithmetic/Examples/Paper/SturmExample.lean#L11}{\lstinline|S : SturmBuilderOfList P p q|}
	with
		\begin{align*}
			P &\equiv [[-8, 9, 0, -3, 0, 1], [9, 0, -9, 0, 5], [20, -18, 0, 3],
			[-27, 100, -63], [-4752, 3103], [1]]  \\
			p &\equiv [-8, 9, 0, -3, 0, 1] \\
			q &\equiv [9, 0, -9, 0, 5]
		\end{align*}
		For this, we must verify that indeed $e_i P_i = Q_i P_{i + 1} - f_i P_{i + 2}$.
The number of sign changes in the sequence $(1, 5, 3, -63, 3103, 1)$ of leading coefficients of the $P_i$ is equal to $s = 2$, while the number of sign changes in the sequence $(-1, 5, -3, -63, -3103, 1)$ is $t = 3$. By Theorem \ref{Sturm}, we conclude that the total number of real roots of $T$ is $t-s = 3 - 2 =1$. Using \lstinline|S| in Lean, we \href{https://github.com/alainchmt/CertifyingInvariantsNF/blob/v1/IdealArithmetic/Examples/Paper/SturmExample.lean#L40}{prove} that

		\begin{lstlisting}
example : #(Multiset.toFinset (X ^ 5 - 3 * X ^ 3 + 9 * X - 8 : ℝ[X]).roots) = 1 := by ...
		\end{lstlisting}

		Additionally, using the formula in \eqref{resformula}, we find that $\operatorname{disc} T = \operatorname{Res}(T, T') = 2516240$. In \href{https://github.com/alainchmt/CertifyingInvariantsNF/blob/v1/IdealArithmetic/Examples/Paper/SturmExample.lean#L48}{Lean}, 

		\begin{lstlisting}
example : discr T = 2516240 :=  by ...
		\end{lstlisting}

	\end{example}
For a number field $K$ and a \lstinline[mathescape]| $\mathcal{O}$ : Subalgebra ℤ K|, we bundle in the structure \href{https://github.com/alainchmt/CertifyingInvariantsNF/blob/v1/IdealArithmetic/Saturation/CertifyTorsionOrder.lean#L318}{\lstinline[mathescape]|RankUnitsCertificate|} a proof asserting that $\mathcal{O} = \mathcal{O}_K$ together with a Sturm sequence \lstinline|P : List (List ℤ)| starting with a defining polynomial for $K$ and its derivative using the computable expressions described above. We use it to certify the signature of $K$, and hence the rank of the unit group of $\mathcal{O}$.

	\begin{example}
		For the number field $K$ with defining polynomial $T=X ^ 5 - 3 X ^ 3 + 9 X - 8$ from Example \ref{sturm}, after proving an explicit integral basis for its ring of integers $\mathcal{O}$, and computing a Sturm sequence starting from $T$ and its derivative, we obtain efficient \href{https://github.com/alainchmt/CertifyingInvariantsNF/blob/v1/IdealArithmetic/Examples/Paper/UnitRankExample.lean#L33}{Lean proofs} of the value of the discriminant of $K$, its signature, as well as the rank of its unit group modulo torsion.  

		\begin{lstlisting}[mathescape]
theorem K_discr : NumberField.discr K = 157265 := by ...

lemma K_nrComplexPlaces : NumberField.InfinitePlace.nrComplexPlaces K = 2 := by ...

lemma K_nrRealPlaces : NumberField.InfinitePlace.nrRealPlaces K = 1 := by ...

lemma K_units : Module.finrank ℤ (Additive ($\mathcal{O}^{\times}$ / (CommGroup.torsion $\mathcal{O}^{\times}$))) = 2 := by ...
		\end{lstlisting}
	\end{example}

	\section{Certifying generators of the class group}\label{sec : idealgen}

	The previous invariants will play a key role in the certification of the class group. Recall that the class group of a number field $K$ is a finite abelian group. The procedure we introduce goes as follows. 
    \begin{enumerate}
        \item First, certify generators. This consists of two parts: 
        \begin{enumerate}
            \item Ensure that $\operatorname{Cl}(\mathcal{O}_K)$ is generated by the classes of an explicit (finite) collection of prime ideals $\mathfrak{B}$ called a \emph{factor base}. 
            \item Using $\mathfrak{B}$, certify that the classes of a (typically much smaller) collection $\{J_1, \ldots, J_m\}$ of nonzero ideals generate $\operatorname{Cl}(\mathcal{O}_K)$. 
        \end{enumerate}
        \item Secondly, with $n_i$ the order of $[J_i]$ in $\operatorname{Cl}(\mathcal{O}_K)$, prove that $\mathbb{Z}/n_1\mathbb{Z} \times \ldots \times \mathbb{Z}/n_m\mathbb{Z} \cong \operatorname{Cl}(\mathcal{O}_K)$ with the isomorphism being $(\bar{a}_1, \ldots, \bar{a}_m)  \mapsto [J_1]^{a_1} \dots [J_m]^{a_m}$. 
    \end{enumerate}
    This gives us an explicit description of the structure of $\operatorname{Cl}(\mathcal{O}_K)$. 
    Here, the $J_i$ are chosen so that their classes are generators of the cyclic factors of $\operatorname{Cl}(\mathcal{O}_K)$ (although this is only proven in the second step).

	In this section, we focus on the first step. We begin by reviewing how the class group has been defined in Mathlib. Then, we discuss factor base bounds and introduce certificates for part (a) and (b).  

	\subsection{The class group in Lean}\label{sec : class}
	We recall the definition of the class group of an integral domain and go over its current implementation in Mathlib.
	The class group of an integral domain $R$ is an abelian group that measures how far $R$ is from being a principal ideal domain (PID). Its definition was first formalized in Lean 3 (and subsequently ported to Lean 4) in \cite{dedekindforma}, where the finiteness of the class group of the ring of integers of a number field (simply referred to as the class group of the number field) was also formally established.
	To introduce this concept, we recall the notion of fractional ideals. If $F$ is the fraction field of $R$, a subset $I \subseteq F$ is a fractional ideal if there is $x \in R \setminus \{0\}$ such that $xI \subseteq R$ is an ideal of $R$. The collection of fractional ideals has a semiring structure with $R$ being the identity. The units in here are the \emph{invertible} fractional ideals. 
	
    In Mathlib, the fractional ideals of $R$ are denoted by \lstinline|FractionalIdeal (nonZeroDivisors R) F|. There is a map \lstinline[mathescape]|$F^{\times}$ →* (FractionalIdeal (nonZeroDivisors R) F)$^{\times}$|, called \href{https://github.com/leanprover-community/mathlib4/blob/bfc0bff5a036ddf63093b58716f09b074e9e360d/Mathlib/RingTheory/ClassGroup/Basic.lean#L47-L55}{\lstinline|toPrincipalIdeal R F|}, sending $\alpha \in F^{\times}$ to the invertible fractional ideal $R\cdot \alpha \subseteq F$. Its image consists of the nonzero principal fractional ideals (fractional ideals generated, as $R$-modules, by a single nonzero element). The class group $\operatorname{Cl}(R)$ is \href{https://github.com/leanprover-community/mathlib4/blob/6f1c6456ee863edbcf96c4febc52cd1c8d07487f/Mathlib/RingTheory/ClassGroup/Basic.lean#L86-L90}{defined} in Mathlib as the quotient of the group of invertible fractional ideals of $R$ by the subgroup consisting of the nonzero principal fractional ideals. That is, 
	\begin{lstlisting}[mathescape]
def ClassGroup R := (FractionalIdeal (nonZeroDivisors R) (FractionRing R))$^{\times}$ $/$
   (toPrincipalIdeal R (FractionRing R)).range
	\end{lstlisting}

	If $R$ is a Dedekind domain, then every nonzero fractional ideal $I$ is invertible. In that case, 
it is convenient to have a map sending a nonzero ideal $I \subseteq R$ to the corresponding class $[I]$ in $\operatorname{Cl}(R)$. We use the map \href{https://github.com/leanprover-community/mathlib4/blob/bfc0bff5a036ddf63093b58716f09b074e9e360d/Mathlib/RingTheory/ClassGroup/Basic.lean#L253-L255}{\lstinline|ClassGroup.mk0|}, available in Mathlib, of type \lstinline|nonZeroDivisors (Ideal R) →* ClassGroup R|. The subtype \lstinline|nonZeroDivisors (Ideal R)| consists of the ideals of $R$ which are not zero divisors in the semiring of ideals of $R$ (in this context, these are the nonzero ideals of $R$). There is a coercion from \lstinline|nonZeroDivisors (Ideal R)| to \lstinline|Ideal R|. 
    
    For two nonzero ideals $I$ and $J$ of $R$, we have that $[I] = [J]$ in $\operatorname{Cl}(R)$ if and only if there exists $\alpha, \beta \in R \setminus \{0\}$ such that $\langle \alpha \rangle * I = \langle \beta \rangle * J$. This criterion, formulated purely in terms of (integral) ideals of $R$, allows us to prove relations in $\operatorname{Cl}(R)$ without making any reference to fractional ideals.

	\subsection{Minkowski's and other bounds}
    
    We will say that $D \in \mathbb{R}$ is a \emph{factor base bound} for $K$ if the classes of the ideals in the set $\{\mathfrak{p} \subseteq \mathcal{O}_K \mid \mathfrak{p} \text{ is a nonzero prime ideal with } N(\mathfrak{p}) \leq D \}$ generate $\operatorname{Cl}(\mathcal{O}_K)$. We formalize this \href{https://github.com/alainchmt/CertifyingInvariantsNF/blob/v1/IdealArithmetic/Generation/ClassGroupGeneration.lean#L64}{definition} as the predicate
	\begin{lstlisting}[mathescape]
def isFactorBaseBound (K : Type*) [Field K] [NumberField K] (D : ℝ) :=
  Subgroup.closure (ClassGroup.mk0 '' {I : $\uparrow$ (nonZeroDivisors (Ideal $\mathcal{O}_K$))
      | (I : Ideal $\mathcal{O}_K$).IsPrime ∧ Ideal.absNorm (I : Ideal $\mathcal{O}_K$) ≤ D }) = ⊤
	\end{lstlisting}
	Here, $\mathcal{O}_K$ stands for \lstinline|NumberField.ringOfIntegers K|. The notation \lstinline|ClassGroup.mk0 '' S| denotes the image of a set of nonzero ideals $S$ of $\mathcal{O}_K$ in $\operatorname{Cl} (\mathcal{O}_K)$. The notation \lstinline|Subgroup.closure G| with \lstinline[mathescape]|G : Set (ClassGroup $\mathcal{O}_K$)| is the subgroup of $\operatorname{Cl} (\mathcal{O}_K)$ generated by the elements in $G$.  Furthermore, \lstinline|Ideal.absNorm I| refers to the norm of an ideal of a Dedekind domain, defined as the cardinality of $\mathcal{O}_K/I$. Finally, 
    \lstinline|⊤| is the entire group $\operatorname{Cl}(\mathcal{O}_K)$ regarded as a subgroup of itself.

	A well-known theorem of Minkowski is the following \cite[Theorem 10.5]{stevenhagen}.
	\begin{theorem}\label{mink}
		Consider an ideal class $C \in \operatorname{Cl}(\mathcal{O}_K)$. There exists a nonzero (integral) ideal $I$ of $\mathcal{O}_K$ such that $[I] = C$ and
		\begin{align*}
			\operatorname{N}(I) \leq \frac{n!}{n^n} \left(\frac{4}{\pi}\right)^{r_2} \sqrt{|\operatorname{disc}(K)|}
		\end{align*}.
	\end{theorem}
	The upper bound in Minkowski's theorem is known as Minkowski's bound, which we denote by $\mathfrak{M}$. The number $r_2$ is the number of pairs of non-real complex embeddings introduced before. 
    Since $\mathcal{O}_K$ is a Dedekind domain, every nonzero proper ideal of $\mathcal{O}_K$ is the product of prime ideals. Consequently, together with Theorem \ref{mink}, this gives that $\mathfrak{M}$ is a factor base bound for $K$.  
	Zimmer \cite{zimmer} later proved a similar theorem to Theorem \ref{mink}, with upper bounds that improve the constant $\frac{n!}{n^n} \left(\frac{4}{\pi}\right)^{r_2}$ in Minkowski's bound. However, the dependency on $\sqrt{\operatorname{disc}(K)}$ remains. There is no known asymptotically better bound for a set of generators of $\operatorname{Cl}(\mathcal{O}_K)$ \cite[Remark 4.3]{magma} unless one is willing to assume unproven conjectures such as the \emph{Generalized Riemann Hypothesis}.

	The Generalized Riemann Hypothesis (GRH) is a generalization of the famous Riemann Hypothesis to a broader class of so-called $L$-functions. The terminology in the literature regarding this conjecture is not completely uniform. Variants of the name include \emph{Extended Riemann Hypothesis} (ERH), and different authors have used these names to mean different things. In this work, as in many contemporary treatments, GRH refers to the open conjecture stating that all Hecke $L$-functions (these extend the notion of Dirichlet $L$-functions from $\mathbb{Q}$ to arbitrary number fields) are zero-free in the half plane $\Re(s) > 1/2$ (see \cite[pag. 360]{bach} for the precise formulation, where this conjecture is referred to as ERH).

	Eric Bach proved in \cite{bach} that, assuming GRH (called ERH in his article), one can take $D = 12\log ^ 2 |\operatorname{disc}(K)| $ as a factor base bound.
    Asymptotically, as the discriminant of $K$ increases, this bound is much smaller than Minkowski's bound. 
	In \cite{belabas}, Karim Belabas, Francisco Diaz y Diaz, and Eduardo Friedman give another conditional bound which is smaller than Bach's bound in practice, but not asymptotically.

	Since we intend to work with explicit number fields, we need an explicit factor base bound. Theorem \ref{mink} has been formalized in Lean by Xavier Roblot and is \href{https://github.com/leanprover-community/mathlib4/blob/6f1c6456ee863edbcf96c4febc52cd1c8d07487f/Mathlib/NumberTheory/NumberField/ClassNumber.lean#L77-L98}{available} in Mathlib. Using it together with the unique factorization in \lstinline[mathescape]|Ideal $\mathcal{O}_K$|, the type of ideals of $\mathcal{O}_K$, we \href{https://github.com/alainchmt/CertifyingInvariantsNF/blob/v1/IdealArithmetic/Generation/ClassGroupGeneration.lean#L71}{prove} that Minkowski's bound $\mathfrak{M}$ is indeed a factor base bound.

	\begin{lstlisting}
lemma minkowskiBoundFB_isFactorBaseBound (K : Type*) [Field K] [NumberField K] :
   isFactorBaseBound K (minkowskiBoundFB K) := by ...
	\end{lstlisting}

	The term \lstinline|minkowskiBoundFB K : ℝ| is the quantity $\mathfrak{M}$ from before. By combining the decimal approximations for $\pi$ provided in Mathlib with the efficient computation of the discriminant of a number field presented in Section \ref{sec:disc}, as well as suitable approximations of its square root, we can employ tactics such as \lstinline|norm_num| to carry out explicit numerical evaluations and prove a decimal upper bound for \lstinline|minkowskiBoundFB K|. 

	\begin{example}
		Let $K$ be the number field with defining polynomial $T = X^4 - X^3 - 80X^2 - 332X - 383 $. In Lean, we \href{https://github.com/alainchmt/CertifyingInvariantsNF/blob/v1/IdealArithmetic/Examples/Paper/MinkowskiBoundExample.lean#L8}{prove}

		\begin{lstlisting}
theorem K_minowski : minkowskiBoundFB K ≤ 691.357086451742  := by ...
		\end{lstlisting}

	\end{example}

	Minkowski's bound is the only explicit and general factor base bound that, to our knowledge, has been formalized in a proof assistant. In this project, we restrict ourselves to using Minkowski's bound for the certification of class groups. Nonetheless, we formalize the relevant theorems (see the next section) for an arbitrary factor base bound. The formalization of GRH-conditional bounds would be a worthwhile direction for future work that could enable certification of the class group, and other invariants, for a larger class of number fields (under the assumption of GRH). We note that to formulate the statement of GRH, one must formalize some of the theory of Hecke characters and Hecke $L$-functions. For the case $K = \mathbb{Q}$, David Loeffler and Michael Stoll have formalized Dirichlet $L$-functions in Lean \cite{loeffler} and formulated the classical Riemann Hypothesis.
	\begin{remark}
		The fastest known procedures to compute the class group of a number field are based on a probabilistic algorithm originally due to Buchmann \cite{buchmann}. It is conditional on GRH in two fundamental ways: first, it starts by determining a factor base relying on Bach's bound, which is proved under GRH; second, the algorithm's stopping criterion is based on error bounds for approximations of certain truncated Euler products, which are known to hold assuming GRH. Even using CASs with highly optimized implementations in a high-performance language, computing class groups unconditionally is computationally prohibitive for number fields of large discriminant.

	\end{remark}

	\subsection{The collection of prime ideals of norm below some bound}\label{sec: collection}
	In this section, we describe the construction of a collection containing all nonzero prime ideals of $\mathcal{O}_K$ with norm less than $C \in \mathbb{N}$ and explain how we certify this in Lean. For $D$ a factor base bound and $D < C$, this, together with the verification of relations between ideals, is required for the first step in Section \ref{sec: redgen}.  Concretely, for a subalgebra $\mathcal{O}$ of $K$ with a multiplication table and such that $\mathcal{O} = \mathcal{O}_K$, we discuss, for an appropriate tuple \lstinline[mathescape] | $\mathfrak{p}$: m → Ideal $\mathcal{O}$|, how to construct a proof of
	\begin{lstlisting}[mathescape]
		{I : Ideal $\mathcal{O}$ | 0 < I.absNorm ∧ I.IsPrime ∧ I.absNorm < C} ⊆ Set.range $\mathfrak{p}$
	\end{lstlisting}
	Given an explicit prime ideal $\mathfrak{p} \subseteq \mathcal{O}$, we can prove that it is prime and of norm below $C$ following the methods in Sections \ref{sec : certnorm} and \ref{sec : primality}. Ensuring that a given collection of ideals contains \emph{all} such ideals requires additional work. We achieve this from a factorization of $p\mathcal{O}$ for every prime number $p < C$.

	Since $\mathcal{O}_K$ is a Dedekind domain, the ideal $p\mathcal{O}_K$ admits a factorization into (not necessarily distinct) prime ideals 
	$p\mathcal{O}_K = \mathfrak{p}_1 \cdots \mathfrak{p}_t.$ Let $\mathcal{F}_p$ be a tuple (with possible repetitions) containing all of these prime ideal factors. Every prime ideal of $\mathcal{O}_K$ that contains $p$ occurs among the factors $\mathfrak{p}_i$ and so is a member of $\mathcal{F}_p$. Moreover, as previously observed, every nonzero prime ideal of $\mathcal{O}_K$ necessarily contains some rational prime. This implies that
	\begin{align*}
		\{I \subseteq \mathcal{O}_K \mid I \text{ is prime with } 0 < N(I) < C  \} \subseteq \bigcup_{p < C} \{\mathfrak{p} \in \mathcal{F}_p \mid N(\mathfrak{p}) < C\},
	\end{align*}
	where $p$ ranges over all prime numbers less than $C$.

For a prime number $p$ and a tuple of nonzero prime ideals  $\mathcal{F} := (\mathfrak{p}_1, \ldots, \mathfrak{p}_t)$ of $\mathcal{O}_K$, the condition that $\mathcal{F}$ contains all the prime ideals (with multiplicity) above $p$ is equivalent to $\prod_i \mathcal{F}_i \subseteq \langle p \rangle $ (since, in $\mathcal{O}_K$, to contain is to divide). 
This leads to the following certificate. 

\paragraph{Certificate of a factor base}
Let $\mathfrak{B}$ be a tuple of ideals of $\mathcal{O}_K$ and $P = (p_1, \ldots, p_m)$ the collection of all prime numbers $p$ with $p < C$. 
The following data certifies that $\mathfrak{B}$ is a factor base. For each $i \in \{1, \ldots, m\}$: 
\leavevmode
\begin{multicols}{2}
	\raggedcolumns
	\begin{itemize}
        \item A natural number $t_i$; 
		\item a $t_i$-tuple of ideals $\mathcal{F}_i$ of $\mathcal{O}_K$; 
        \item a $t_i$-tuple of natural numbers $N_i$. 
	\end{itemize}
\end{multicols}
Verification of the certificate amounts to checking, for each $i \in \{1, \ldots, m\}$, that
\begin{multicols}{2}
	\raggedcolumns
	\begin{enumerate}[label=(\roman*)]
        \item for every $j$, $\# (\mathcal{O}_K /\mathcal{F}_{ij}) = N_{ij}$; \label{itemnorm}
        \item for every $j$, $N_{ij} \neq 0$; \label{itemnezero}
		\item for every $j$, $\mathcal{F}_{ij}$ is a prime ideal; \label{itemp}
        \item $\prod_{j} \mathcal{F}_{ij} \subseteq \langle p_i \rangle$; \label{itemprod}
        \item for every $j$, $N_{ij} < C$ implies $\mathcal{F}_{ij} \in \mathfrak{B}$. 
        
	\end{enumerate}
\end{multicols}
The statement \ref{itemnezero} is included to guarantee, in combination with \ref{itemnorm}, that the ideals in $\mathcal{F}_i$ are nonzero (we use the convention that $\# A$ for an infinite $A$ equals zero). Together with \ref{itemp} and \ref{itemprod}, this implies that each prime ideal above a prime number $p < C$ is a member of the set $\{\mathcal{F}_{ij} | 1 \leq i \leq m, 1 \leq j \leq t_i\}$. 

We collect these data and proofs of the verification statements in the \href{https://github.com/alainchmt/CertifyingInvariantsNF/blob/v1/IdealArithmetic/Generation/ClassGroupGeneration.lean#L544}{structure}
	\begin{lstlisting}[mathescape]
structure PrimesBelowBoundCertificate ($\mathcal{O}$ : Type*) [CommRing $\mathcal{O}$] [Nontrivial $\mathcal{O}$] (C : ℕ) ...
    \end{lstlisting}
    
 It is parametrized over a general commutative ring $\mathcal{O}$ and bound $C \in \mathbb{N}$. In our applications, $\mathcal{O}$ will be a subalgebra of $K$ equipped with a multiplication table and propositionally equal to $\mathcal{O}_K$. 
 As part of its fields, we include 
\begin{lstlisting}[mathescape]
    t : Fin m → ℕ
    P : Fin m → ℕ
    F : Π i,  Fin (t i) → Ideal $\mathcal{O}$
    N :  Π i,  Fin (t i) → ℕ
    β : Fin r → Ideal $\mathcal{O}$
\end{lstlisting}
 Here \lstinline|P| is the tuple of the prime numbers below $C$. The dependent functions \lstinline|F| and \lstinline|N| correspond to the collections of ideals and norms, respectively. Since \lstinline|F| and \lstinline|N| have the same dependent shape, their indexing type coincides definitionally. This is convenient for relating ideals with their corresponding norms. Furthermore, the factor base $\mathfrak{B}$ is given as the tuple  \lstinline|β|. 
 
We bundle the statements \ref{itemp} and \ref{itemprod} with the predicate \href{https://github.com/alainchmt/CertifyingInvariantsNF/blob/v1/IdealArithmetic/Generation/ClassGroupGeneration.lean#L538}{ \lstinline|ContainsPrimesAboveP (P i) (F i)|}. We prove a goal of the form \lstinline|(F i j).IsPrime| using the certificate in Section \ref{sec : primality}. Furthermore, we verify \lstinline|∏ j, F i j ≤ Ideal.span {↑(P i)}| by using a chain of containments as described in Section \ref{sec : idealmul}.

 In practice, we choose \lstinline|F i| to be precisely the collection of prime ideals whose product equals $\langle p_i\rangle$, however, note that we only certify one of the containments and we can avoid the reverse containment entirely, reducing the amount of computations needed for verification.

	\begin{remark}
		Suppose $K = \mathbb{Q}(\alpha)$, with $\alpha \in \mathcal{O}_K$ of minimal polynomial $T \in \mathbb{Z}$. 
        For $ p \nmid [\mathcal{O}_K : \mathbb{Z}[\alpha]]$, the Kummer-Dedekind theorem (see Remark \ref{remarkkummer}) reduces the certification of the factorization of $p\mathcal{O}_K$ to that of the factorization of $T \pmod p$ in $\mathbb{F}_p[X]$, while providing the norms of the resulting prime ideals.  As previously noted, incorporating this approach into our framework is part of ongoing work.
	\end{remark}

	Given a commutative ring $\mathcal{O}$ with $\mathcal{O} \cong \mathcal{O}_K$ and a term \lstinline[mathescape]| A : PrimesBelowBoundCertificate $\mathcal{O}$ C| we prove the following \href{https://github.com/alainchmt/CertifyingInvariantsNF/blob/v1/IdealArithmetic/Generation/ClassGroupGeneration.lean#L934}{lemma}.  
	\begin{lstlisting}[mathescape]
lemma le_primes_below_bound_of_PrimesBelowBoundCertificate_le_of_eq {$\mathcal{O}$ : Type*} 
    [CommRing $\mathcal{O}$] [IsDedekindDomain $\mathcal{O}$] [Module.Free ℤ $\mathcal{O}$]  {C : ℕ} (φ : $\mathcal{O}$ ≃+* $\mathcal{O}_K$) 
    (A : PrimesBelowBoundCertificate $\mathcal{O}$ C) :
  {I : Ideal $\mathcal{O}$ | 0 < I.absNorm ∧ I.IsPrime ∧ I.absNorm < C} ⊆ Set.range (A.β) := by ...
	\end{lstlisting}

	We wrote a SageMath routine that, given a subalgebra $\mathcal{O}$  equipped with a multiplication table and equal to $\mathcal{O}_K$, produces, for every prime number $p < C$, Lean code defining a tuple \lstinline|F| of the prime ideals above $p\mathcal{O}$, and a proof of \lstinline|ContainsPrimesAboveP p F|.

	If the bound $C$ is large, this results in a substantial number of primality, multiplication, and containment proofs. Assembling them to produce a term of type \lstinline[mathescape]|PrimesBelowBoundCertificate $\mathcal{O}$ C| leads to timeout errors when the number of ideals is large. For this reason, we introduce the intermediate certificate \href{https://github.com/alainchmt/CertifyingInvariantsNF/blob/v1/IdealArithmetic/Generation/ClassGroupGeneration.lean#L570}{\lstinline[mathescape]|PrimesBelowBoundCertificateInterval  $ \hspace{0.1cm} \mathcal{O} \hspace{0.1cm} l$ u C|}, which collects all the prime ideals appearing in the factorization of $p\mathcal{O}$ for $ l < p \leq u$ with norm less than $C$.
	The fields of this structure are almost identical to \lstinline[mathescape]|PrimesBelowBoundCertificate $\mathcal{O}$ C|, with the difference being that \lstinline|P| is now the tuple of primes in the interval $(l, u]$, as asserted by \lstinline[mathescape]|hP : ∀ p, p ∈ Set.range P ↔ Nat.Prime p ∧ $l$ < p ∧ p ≤ u|.
We therefore split the primes in the interval $(0,C)$ into intervals $(0, e_0], (e_0, e_1], \ldots, (e_{r-1}, e_{r}]$ with $e_r = C - 1$ and construct an interval certificate for each. All of these are independent of each other, which allows parallelization in the verification. To combine them, we concatenate the corresponding tuples.

    In the next section, we discuss the problem of certifying that a specified list consists of all the prime numbers within a given interval.

	\subsection{Efficient retrieval of primes}
	A key component of the certificates discussed in the previous section is solving goals of the form
	\begin{lstlisting}[mathescape]
		∀ p, p ∈ L ↔ Nat.Prime p ∧ $l$ < p ∧ p ≤ u
	\end{lstlisting}
	with \lstinline|L : List ℕ|, \lstinline[mathescape]|$l$ : ℕ| and \lstinline|u : ℕ|. Since such goals arise repeatedly, we handle them by precomputing and proving the list of prime numbers up to 20,000 and storing it in a way that allows for efficient retrieval of the primes in a given interval by using a binary tree.

Our initial method for generating the prime numbers up to $20,000$ followed a pattern commonly presented in functional programming literature as a version of the ``sieve of Eratosthenes'' where, starting from an initial list, multiples of each known prime $p$ are repeatedly removed by testing $x \mod p = 0$ for every remaining $x$. As noted by O'Neill in \cite{sieve}, this implementation is more accurately viewed as a form of trial division rather than a genuine sieve. We are thankful to Bhavik Mehta for pointing us to this reference and for suggesting several optimizations, many of which we adopted. 

We define the function \href{https://github.com/alainchmt/CertifyingInvariantsNF/blob/v1/IdealArithmetic/Computation/PrimeSieve.lean#L29}{ \lstinline|primesBetween A B|}, which implements trial division more efficiently. Given a list of integers $P$, it constructs the list of natural numbers in the interval $(A , B]$
and filters out all elements $x$ for which there is $p \in P$ with $p < x$ and $x \mod p = 0$. If the list \lstinline|P| contains all prime numbers less than $\sqrt{B} + 1$, satisfies $1 \notin P$ (e.g. the list of all primes less than $\sqrt{B} + 1$), and $A \geq 1$,  then \lstinline|primesBetween A B P| is guaranteed to return the list of prime numbers in the interval $(A, B]$.

One of the limitations encountered when working with large lists in Lean is the recursion depth limit. The type \lstinline|List α| is an inductive type and is represented internally by a nested sequence of \lstinline|List.cons| applications. For sufficiently long lists, elaborating such a term exceeds Lean's \lstinline|macRecDepth| internal parameter, which limits the depth of recursive elaboration. Although there are ways around this, like manually increasing this limit,
we prefer to partition the interval $(1, 20000]$ into subintervals of length 300 and certify the list of primes in each of these with \lstinline|primesBetween|. This results in a more modular approach, with the possibility of parallelization. We then prove that the concatenation of all of these lists yields the list of primes up to $20,000$. The resulting list $L$ has length $2,262$.

	Given $L$ of prime numbers up to $20,000$, we want to quickly retrieve the primes $p \in L$ such that $ l < p \leq u$. Even assuming that the list is sorted, this would require a linear traversal of the list. A more efficient approach is to store $L$ in a balanced tree.

	Mathlib's type \href{https://github.com/leanprover-community/mathlib4/blob/6f1c6456ee863edbcf96c4febc52cd1c8d07487f/Mathlib/Data/Ordmap/Ordnode.lean#L70-L77}{\lstinline|Ordnode α|} denotes the type of finite sets of elements of \lstinline|α|, represented as a binary tree. The constant \lstinline|nil : Ordnode α| denotes the empty tree. For trees \lstinline|t1 t2 : Ordnode α|, an element \lstinline|x : α|, and a natural number \lstinline|n : ℕ|, the term \lstinline|node n t1 x t2| denotes the tree of size \lstinline|n| whose root node is labeled by \lstinline|x|, with left and right subtrees given by \lstinline|t1| and \lstinline|t2|, respectively.

	The property \lstinline|t.Balanced| states that, for each node of \lstinline|t|, either the left and right subtrees have at most one element between them, or their sizes differ by a factor of at most $3$. This guarantees that a balanced tree of size $n$ has height $O(\log n)$. If \lstinline|α| has an order $<$ defined, the property \lstinline|t.Bounded ⊥ ⊤| means that for each node \lstinline|t| labeled \lstinline|x|, the elements on the left subtree are all smaller than \lstinline|x|, and the elements in the right subtree are all bigger than \lstinline|x|. Mathlib includes a function \href{https://github.com/leanprover-community/mathlib4/blob/6f1c6456ee863edbcf96c4febc52cd1c8d07487f/Mathlib/Data/Ordmap/Ordnode.lean#L1377}{\lstinline|Ordnode.ofList|} which sends a list over \lstinline|α| to the corresponding tree of type \lstinline|Ordnode α|. If the order in \lstinline|α| is total, then the resulting tree is balanced and bounded.
	\begin{example}
		Consider the list $L = [2,3,5,7,11,13,17]$ of primes less than $18$. We have

		\begin{lstlisting}
Ordnode.ofList L = node 7 (node 3 (node 1 nil 2 nil) 3 (node 1 nil 5 nil) ) 7 (node 3 (node 1 nil 11 nil) 13 (node 1 nil 17 nil))
		\end{lstlisting}
		which corresponds to the balanced tree:
		\begin{center}
\scalebox{0.7}{
			\begin{tikzpicture}[
				level distance=1.4cm,
				sibling distance=3cm,
				every node/.style={circle,draw,minimum size=5mm},
				level 2/.style={sibling distance=1.5cm},
				level 3/.style={sibling distance=8mm}
				]

				\node {7}
				child {
					node {3}
					child { node {2} }
					child { node {5} }
				}
				child {
					node {13}
					child { node {11} }
					child { node {17} }
				};

			\end{tikzpicture}}

		\end{center}

	\end{example}

	Given a type \lstinline|α| with a \lstinline|[LinearOrder α]| instance, we define the function \href{https://github.com/alainchmt/CertifyingInvariantsNF/blob/v1/IdealArithmetic/Computation/PrimeSieve.lean#L265}{\lstinline|Ordnode.extractRangeTree|} recursively. It takes as input a balanced tree \lstinline|t : Ordnode α| and upper and lower bounds \lstinline[mathescape]|$l$₁ $l$₂ : α|, and extracts a (not necessarily balanced) tree containing the elements of \lstinline|t| in the interval $[l_1, l_2]$.

	Let $L$ be the certified list of primes up to a 20,000.  We define \href{https://github.com/alainchmt/CertifyingInvariantsNF/blob/v1/IdealArithmetic/Computation/PrimeSieve.lean#L542}{\lstinline|PTreeE : Ordnode α|}, by explicitly writing down the balanced tree with all the primes up to 20,000, and then prove that \lstinline|Ordnode.ofList L = PTreeE| by \lstinline|rfl|. This proof does require that we increase the \lstinline|maxRecDepth|, but this one-time increase is an acceptable trade-off for efficient repeated retrieval. We then formalize the following \href{https://github.com/alainchmt/CertifyingInvariantsNF/blob/v1/IdealArithmetic/Computation/PrimeSieve.lean#L552}{result}
	\begin{lstlisting}
lemma primes_range (l₁ l₂ : ℕ) (hle : l₂ ≤ 20000) {p : ℕ} :
	 p ∈ (Ordnode.extractRangeTree PTreeE (l₁ + 1) l₂).toList  ↔ Nat.Prime p ∧ l₁ < p ∧ p ≤ l₂ := ...
	\end{lstlisting}
	with \lstinline|toList| being the map that sends a tree to the corresponding list (in increasing order). This allows us to prove statements such as the following \href{https://github.com/alainchmt/CertifyingInvariantsNF/blob/v1/IdealArithmetic/Examples/Paper/RetrievalOfPrimes.lean#L5}{example}
	\begin{lstlisting}
example {p} : p ∈ [13877, 13879, 13883, 13901, 13903, 13907, 13913, 13921, 13931, 13933, 13963,
      -- 52 explicit primes omitted 
      14537, 14543, 14549, 14551, 14557, 14561, 14563]  ↔ Nat.Prime p ∧ 13876 < p ∧ p ≤ 14565 := 
    primes_range 13876 14565 (by omega)
	\end{lstlisting}
	in a fraction of a second (less than 0.05s in our hardware), without encountering recursion depth errors.

\subsection{An alternative generating set}\label{sec: redgen}
	After enumerating the prime ideals of norm below a certain factor base bound, we obtain a set of prime ideals $\mathfrak{B}$, which is called a \emph{factor base}, whose classes are guaranteed to generate the whole class group. Our objective in this section is to prove another generating set, which is typically much smaller. The main strategy is to use relations in the class group that involve the ideals in the factor base. 
	Finding enough of these relations is one of the most time-consuming parts of the algorithms for computing class groups. 
    The approaches to do this rely on involved subroutines such as algorithms for factoring ideals of $\mathcal{O}_K$ \cite[Section 6.5]{zassenhaus} and the LLL lattice reduction algorithm \cite{cohenalgo}.

	In our certification-based framework, however, we can delegate the task of \emph{finding} these relations to an efficient CAS, while limiting the computations inside the proof assistant to \emph{verifying} that the relations indeed hold. We consider the following certificate: 

\paragraph{Certificate of an alternative set of generators}
Let $\mathfrak{B} = (\mathfrak{p}_1, \ldots, \mathfrak{p}_t)$ be a factor base and $J = (J_1, \ldots, J_m)$ a collection of ideals of $\mathcal{O}_K$. The following data certifies that $\operatorname{Cl}(\mathcal{O}_K) = \langle [J_1], \ldots [J_m]\rangle$. 
\leavevmode
	\begin{itemize}
        \item For each $i \in \{1, \ldots, t\}$, a pair $(\alpha_i, \beta_i) \in \mathcal{O}_K \times \mathcal{O}_K$ of nonzero elements;   
        \item a $t \times m$ matrix $M$ with coefficients in $\mathbb{Z}_{\geq 0}$. 
	\end{itemize}
Verification simply involves checking that 
\begin{enumerate}[label=(\roman*)]
\item $\langle \alpha_i \rangle \cdot \mathfrak{p}_i = \langle \beta_i \rangle \cdot \prod_k J_k ^ {M_{ik}}$ for every $i \in \{1, \ldots, t\}$. \label{relations}
\end{enumerate}

The previous identity is written purely in terms of integral ideals, and implies that $[\mathfrak{p}_i] = \prod_k [J_k]^{M_{ik}}$ in $\operatorname{Cl}(\mathcal{O}_K)$ for every $i$. It follows that every element in the factor base is generated by $\{[J_1], \ldots, [J_m]\}$. This proves $\operatorname{Cl}(\mathcal{O}_K) = \langle [J_1], \ldots [J_m]\rangle$.

To state the underlying lemma that allows us to use the previous certificate, we do as follows. We encode the tuples $\mathfrak{B}$ and $J$ as \lstinline[mathescape]|$\mathfrak{p}$ : m → Ideal $\mathcal{O}_K$| and \lstinline[mathescape]|J : n → Ideal $\mathcal{O}_K$|, respectively, and the matrix \lstinline|M : Matrix m n ℕ|, with \lstinline|m| and \lstinline|n| finite types.
	In order to use \lstinline|ClassGroup.mk0| to map their image to their respective classes in $\operatorname{Cl}(\mathcal{O}_K)$, the inputs must be of type \lstinline[mathescape]|nonZeroDivisors (Ideal $\mathcal{O}_K$)|. We handle this by introducing \lstinline[mathescape]|$\mathfrak{p}$' : m → nonZeroDivisors (Ideal $\mathcal{O}_K$)| and \lstinline[mathescape]|J' : n → nonZeroDivisors (Ideal $\mathcal{O}_K)$| together with the coercion hypothesis \lstinline[mathescape]|hg' : ∀ i, ↑($\mathfrak{p}$' i) = $\mathfrak{p}$ i| and \lstinline|hx' : ∀ i, ↑(J' i) = J i|, which in particular ensures that all the $J_i$ and $\mathfrak{p}_i$ are nonzero ideals. Additionally, we fix a factor base bound \lstinline|D : ℝ| and a natural number \lstinline|C : ℕ| with \lstinline|hB : D < C|. Under the previous assumptions, we \href{https://github.com/alainchmt/CertifyingInvariantsNF/blob/v1/IdealArithmetic/Generation/ClassGroupGeneration.lean#L373}{prove} 
\enlargethispage{\baselineskip}
\refstepcounter{equation}
\label{eq:subgroup-closure}
\begin{center}
\begin{minipage}{0.95\linewidth}
   \begin{lstlisting}[mathescape]
lemma subgroup_closure_eq_classGroup_bound' (hbg : isFactorBaseBound K D)
    (hg : {I : Ideal $\mathcal{O}_K$ | 0 < I.absNorm ∧ I.IsPrime ∧ I.absNorm < C} ⊆ Set.range $\mathfrak{p}$) 
    (hgmul : ∀ i, ∃ (α β : $\mathcal{O}_K$), ∃ (_ : α ≠ 0), ∃ (_ : β ≠ 0),
       Ideal.span {α} * ($\mathfrak{p}$ i) = Ideal.span {β} * ∏ j, (J j) ^ (M i j)) :
  Subgroup.closure (Set.range (fun i => ClassGroup.mk0 (J' i))) = ⊤ := by
   \end{lstlisting}
\end{minipage}
\hfill(\theequation)
\end{center}
\vspace{-0.2cm}
	The hypothesis \lstinline|hg| guarantees that classes of the ideals $\mathfrak{p}_i$ generate $\operatorname{Cl}(\mathcal{O}_K)$ and \lstinline|hgmul| implies that each class $[\mathfrak{p}_i]$ is generated by the classes of ideals in $J$. For decidability reasons, it is convenient to work with a natural number $C$ bounding $D$, instead of using directly the real number $D$. The lemma concludes that the set of classes $[J_i]$ in $\operatorname{Cl}(\mathcal{O}_K)$ generates the full group. 

	In practice, we choose $D$ to be the Minkowski bound of $K$, $C$ the floor of $D + 1$, $\mathfrak{p}$ the prime ideals of norm below $C$, and $J$ a set of ideals whose classes generate the cyclic factors of $\operatorname{Cl}(\mathcal{O}_K)$, which we find using a CAS. We also use a CAS to find the matrix $M$ that expresses each $\mathfrak{p}_i$ in terms of the $J_j$ and the corresponding pair $(\alpha_i, \beta_i) \in \mathcal{O}_K \times \mathcal{O}_K$.  
    Goals of the form given in \lstinline|hgmul| are proven using the ideal arithmetic described in Section \ref{sec:ideal}. In fact, the elements $\alpha_i \in \mathcal{O}_K$ can always be chosen in $\mathbb{Z}$, which simplifies the computation. For more details on how this is handled, we refer to Section \ref{sec: verrel}.

	There is, however, a mismatch with our computation framework. Recall that we work with an explicitly given subalgebra \lstinline[mathescape]|$\mathcal{O}$ : Subalgebra ℤ K| equipped with a multiplication table. 
    This subalgebra is propositionally equal to \lstinline|integralClosure ℤ K|, but not definitionally equal. Since we carry out all computations relative to the type \lstinline[mathescape]|$\mathcal{O}$|, using the previous lemma written in terms of $\mathcal{O}_K$, which is \lstinline|integralClosure ℤ K| coerced into a type, will lead to type errors.

	More generally, if $S$ is any commutative ring isomorphic to $\mathcal{O}_K$, the previous result still holds. We therefore prove a \href{https://github.com/alainchmt/CertifyingInvariantsNF/blob/v1/IdealArithmetic/Generation/ClassGroupGeneration.lean#L435}{version} of \lstinline|subgroup_closure_eq_classGroup_bound'| where $\mathcal{O}_K$ is replaced with a general \lstinline|{S : Type*} [CommRing S]| together with a ring isomorphism \lstinline[mathescape]| φ : S ≃+* $\mathcal{O}_K$|. The instances \lstinline|[Nontrivial S] [IsDedekindDomain S] [Module.Free ℤ S]|, are also included because they are required by \lstinline|ClassGroup.mk0| and \lstinline|Ideal.absNorm|. The proof proceeds by transporting the statement along an explicit isomorphism $\phi : \operatorname{Cl}(S) \xrightarrow{\sim} \operatorname{Cl}(\mathcal{O}_K)$ induced by $\varphi : S \xrightarrow{\sim} \mathcal{O}_K$.  We briefly discuss explicit induced maps between class groups, which we formalize.  

	\subsubsection{The induced map between class groups}\label{sec : iso class}
	Let \lstinline[mathescape]|$\mathcal{O}$ : Subalgebra ℤ K| with \lstinline[mathescape]|h : $\mathcal{O}$ = integralClosure ℤ K|. Suppose we have an explicit morphism \lstinline[mathescape]|φ : ZMod 3  →+ Additive (ClassGroup $\mathcal{O}$)|, where \lstinline|Additive| views the class group as an additive group. That is, we know an explicit \lstinline[mathescape]|J : Ideal $\mathcal{O}$| such that $1 \mapsto  [J]$ under $\varphi$.
	We would like to obtain the corresponding map \lstinline[mathescape]|$\psi$ : ZMod 3  →+ Additive (ClassGroup $\mathcal{O}_K$)|. 
    However, simply rewriting with \lstinline|h| in the term $\varphi$ would result in a term involving a cast along \lstinline|h|, from which we cannot recover an explicit ideal of $\mathcal{O}_K$ whose class is $\psi(1)$.
	A better approach is to compose $\varphi$ with an explicit map $\operatorname{Cl}(\mathcal{O}) \to \operatorname{Cl}(\mathcal{O}_K)$ induced by the inclusion $\mathcal{O} \hookrightarrow \mathcal{O}_K$ coming from the equality $\mathcal{O} = \mathcal{O}_K$. 
    More generally, consider the integral domains $R$ and $S$ together with an injective map $\varphi : R \to S$. By defining appropriate maps between fractional ideals and groups of units, we construct an explicit map \href{https://github.com/alainchmt/CertifyingInvariantsNF/blob/v1/IdealArithmetic/ClassGroupEquiv/ClassGroupEquiv.lean#L214}{$\operatorname{Cl}(R) \to \operatorname{Cl}(S)$} induced by $\varphi$:
	\begin{lstlisting}
noncomputable def ClassGroup.map [IsDomain R] [IsDomain S] {φ : R →+* S} 
    (hinj : Function.Injective φ) : ClassGroup R →* ClassGroup S := by ...
	\end{lstlisting}
    \vspace{-0.3cm}
	If \lstinline|φ : R ≃+* S| is an isomorphism, we prove that \lstinline|φ.symm|, the inverse of $\varphi$, induces the inverse map on class groups. This yields an isomorphism \footnote{After this work was carried out, an isomorphism \href{https://github.com/leanprover-community/mathlib4/blob/6f1c6456ee863edbcf96c4febc52cd1c8d07487f/Mathlib/RingTheory/ClassGroup/Basic.lean\#L486-L495}{\lstinline|ClassGroup.mulEquiv|} induced by a ring isomorphism was added to Mathlib, constructed by composing quotient isomorphisms. Ours, however, is obtained from the more general \lstinline|ClassGroup.map| we define for any injective morphism, and comes with an explicit description on integral ideals.} between class groups which we name \href{https://github.com/alainchmt/CertifyingInvariantsNF/blob/v1/IdealArithmetic/ClassGroupEquiv/ClassGroupEquiv.lean#L325}{\lstinline|ClassGroup.congr φ|}.

	The map \lstinline|ClassGroup.map| can be explicitly described. In our main case of interest, namely over Dedekind domains, where every nonzero ideal is invertible, the following \href{https://github.com/alainchmt/CertifyingInvariantsNF/blob/v1/IdealArithmetic/ClassGroupEquiv/ClassGroupEquiv.lean#L262}{lemma} states that the class of a nonzero (integral) ideal $I$ of $R$ is mapped to the class of the ideal $S\varphi(I)$, denoted by \lstinline|Ideal.map φ I|, 
	\begin{lstlisting}
lemma ClassGroup.map_apply' [IsDedekindDomain R] [IsDedekindDomain S] (φ : R →+* S) 
    (hinj : Function.Injective φ) (I : nonZeroDivisors    (Ideal R)) (J : nonZeroDivisors (Ideal S))
	  (hI : ↑J = Ideal.map φ I) : ClassGroup.map hinj (ClassGroup.mk0 I) = ClassGroup.mk0 J:= ...
	\end{lstlisting}

	\subsubsection{Verification of relations} \label{sec: verrel}
	Our SageMath script, given a prime ideals $\mathfrak{p}_i$ of a factor base and the ideals $J_j$ whose classes generate the cyclic factors of the class group, finds the corresponding $\alpha_i \in \mathbb{Z}$ and $\beta_i \in \mathcal{O}_K$ and row of the matrix $M$ in the certificate. It then writes Lean proofs of the relations in the verification statement \ref{relations} using the ideal multiplications described in Section \ref{sec:ideal} in combination with the structure \href{https://github.com/alainchmt/CertifyingInvariantsNF/blob/v1/IdealArithmetic/IdealArithmetic/IdealArithmetic.lean#L846}{\lstinline|RelationCertificate|} to certify equalities of the form $\langle a \rangle * I = \langle \beta\rangle * J$ with $a \in \mathbb{Z}$ and $\beta \in \mathcal{O}_K$. If the ideal $\mathfrak{p}_i$ is given using a single generator, this step is omitted as the relation is trivial. 

	The number of ideals $J_i$ is typically small as class groups rarely have high $p$-ranks. This is predicted, under suitable assumptions on the number field and the prime $p$, by the refined version of the Cohen-Lenstra-Martinet heuristics \cite[Proposition 2.2]{malle} (note, however, that enumerating number fields by absolute discriminant can pose problems, and other orderings have been considered \cite{lenstra}). We observe that over 98\% of all degree 3 number fields of absolute discriminant at most $10^6$ in the LMFDB (complete list) have cyclic class group. This number is 97\% in degree 4. For both degree 5 and 6 number fields of absolute discriminant at most $10^7$, more than 99\% have cyclic class group. In all of these cases, the class group requires at most three generators. The story is different for quadratic number fields, where the 2-rank of the class group, according to genus theory \cite[Section 7]{genus}, increases with the number of distinct prime factors of the discriminant.

 On the other hand, the number of ideals $\mathfrak{p}_i$ is large if the Minkowski bound is large, and collecting the ideals and the relation proofs in a single Lean term leads to timeouts. We handle this by processing the ideals in batches using the interval certificates \lstinline[mathescape]|PrimesBelowBoundCertificateInterval $\mathcal{O}$ $l$ $u$ C|, introduced in Section \ref{sec: collection}, each describing the prime ideals of norm in the subinterval $(e_{k}, e_{k+1}]$ contained in $(0, C]$. We collect the prime ideals and proofs of the relations separetely for each interval, then concatenate the resulting tuples to produce \lstinline[mathescape]|$\mathfrak{p}$ : m → Ideal $\mathcal{O}$| and a proof of the version of lemma \eqref{eq:subgroup-closure} with $\mathcal{O}$.

\section{Certifying saturation of the class group}\label{sec:satclass}
Recall the two steps of our certification procedure for the class group, described at the start of Section \ref{sec : idealgen}. 
At this point, we assume that we have certified a factor base $\mathfrak{B}$ and the relations proving that a collection of ideals $J_i$ satisfy $\operatorname{Cl}(\mathcal{O}_K) = \langle [J_1], \ldots, [J_m]\rangle$. In this section, we focus on the second step, which certifies the structure of the class group. We introduce a certificate, based on the concept of $p$-saturation, proving that
$\mathbb{Z}/n_1\mathbb{Z}\times \ldots \times \mathbb{Z}/n_m\mathbb{Z}  \cong \operatorname{Cl}(\mathcal{O}_K)$, where the factor $\mathbb{Z}/n_i\mathbb{Z}$ corresponds to the subgroup generated by $ [J_i]$.

Suppose $n_i \in \mathbb{Z}_{>0}$ is such that $[J_i] ^ {n_i} = 1 $ in $\operatorname{Cl}(\mathcal{O}_K)$. This means that the surjective homomorphism 
\begin{align*}
    \psi : \mathbb{Z}^m &\to \operatorname{Cl}(\mathcal{O}_K) \\ (a_1, \ldots, a_m)  &\mapsto [J_1]^{a_1} \dots [J_m]^{a_m}
\end{align*}
factors through $\varphi : \mathbb{Z}/n_1\mathbb{Z} \times \ldots \times \mathbb{Z}/n_m\mathbb{Z} \to \operatorname{Cl}(\mathcal{O}_K) $. The kernel of $\varphi$ catches any extra relations between the generators $[J_i]$ beyond those coming from $[J_i]^{n_i} = 1$. 

For a prime number $p$, we will say that the collection of pairs $([J_i], n_i)$ is \emph{$p$-saturated} if $p \nmid \# \ker \varphi$. This condition is equivalent to that of $n_1 \mathbb{Z} \times \ldots \times n_m\mathbb{Z}$ being a $p$-saturated subgroup of $\ker \psi$, explaining our choice of terminology. An equivalent formulation, and the one we will use, can be stated at the level of integral ideals. The collection $\{([J_i], n_i)\}_{i}$ is $p$-saturated if and only if 
\begin{align}\label{eq : psat_ideal}
	\forall a \in \mathbb{N}^m \text{ and } \forall b \in \mathcal{O}_K \text{, if  } \prod_{i = 1} ^ m J_i ^ {a_i} = \langle b \rangle \text{, then } \forall i, n_i \mid pa_i \implies \forall i, n_i \mid a_i.
\end{align}

Our strategy consists of showing that $\{([J_i], n_i)\}_{i}$ is $p$-saturated for every prime number $p$, and thus $\ker \varphi$ must be trivial. In that case, $\varphi$ is an isomorphism and $\operatorname{Cl}(\mathcal{O}_K) \cong \mathbb{Z}/n_1\mathbb{Z}\times \ldots \times \mathbb{Z}/n_m\mathbb{Z}$. Note that if $p$ does not divide $\prod_{i} n_i$, then $\{([J_i], n_i)\}_{i}$ is automatically $p$-saturated. 

The above argument applies more generally to a finite commutative group $G$ and a collection of generators. We \href{https://github.com/alainchmt/CertifyingInvariantsNF/blob/v1/IdealArithmetic/Saturation/ClassGroupSaturation.lean#L173}{ formalize} it in Lean in this generality, and then specialize to the case of the class group using the integral ideal formulation of $p$-saturation in \eqref{eq : psat_ideal}. Using this technique on the explicit order $\mathcal{O} = \mathcal{O}_K$ we will be able to construct an explicit isomorphism 
\begin{lstlisting}[mathescape]
  (∀ i : ι , (ZMod (n i))) ≃+ Additive (ClassGroup $\mathcal{O}$)
\end{lstlisting}
with $\mathbb{Z}/n_1\mathbb{Z} \times \ldots \times \mathbb{Z}/n_m\mathbb{Z}$ represented as the dependent function type \lstinline|∀ i : ι , (ZMod (n i))|, which \lstinline|i : ι| to an element in \lstinline|ZMod (n i)|, the integers modulo $n_i$. Furthermore, \lstinline[mathescape]|Additive $\mathcal{O}$| is the type \lstinline[mathescape]|$\mathcal{O}$| with the multiplicative structure turned additive. This isomorphism is then composed with the induced $\operatorname{Cl}(\mathcal{O
}) \xrightarrow{\sim} \operatorname{Cl}(\mathcal{O}_K)$ to get an explicit isomorphism between $\mathbb{Z}/n_1\mathbb{Z} \times \ldots \times \mathbb{Z}/n_m\mathbb{Z}$ and $\operatorname{Cl}(\mathcal{O}_K)$.

\subsection{The discrete logarithm approach} \label{sec : disclog}
In this section, we discuss the tools we use to certify that a $\{([J_i], n_i)\}_{i}$ is $p$-saturated. These are based on a technique that uses discrete logarithms in residue fields, and which underlies some of the methods implemented in various CASs, including PARI/GP and Magma \cite[Section 7]{magma}. For the underlying theory, our main reference is \cite[Section 3.6]{hess} which uses a different formulation in the language of $S$-units. Here, we adapt the relevant definitions and results to our setting. We first introduce some machinery to identify $p$-th powers in $\mathcal{O}_K$. 

Let $\mathfrak{q}$ be a prime ideal of $\mathcal{O}_K$ such that $p \mid N(\mathfrak{q}) - 1$. The quotient $\mathcal{O}_K/\mathfrak{q}$ is a finite field and its unit group $(\mathcal{O}_K/\mathfrak{q}) ^ {\times}$ is cyclic and of order $N(\mathfrak{q}) - 1$. Let $\xi \in (\mathcal{O}_K/\mathfrak{q}) ^ {\times}$ be a primitive root (i.e. $\xi$ generates $(\mathcal{O}_K/\mathfrak{q}) ^ {\times}$). Take $x \in \mathcal{O}_K\setminus \mathfrak{q}$. Since $\xi$ is a primitive root, there exists a natural number $k$ such that $\xi ^ k = \bar{x}$ in $\mathcal{O}_K / \mathfrak{q}$. Because of the condition $p \mid N(\mathfrak{q})-1$, the number $k$ is well-defined modulo $p$. We define 
\begin{align*}
	\log_{\xi,p} : \mathcal{O}_K \setminus \mathfrak{q} &\to \mathbb{F}_p \\
	x &\mapsto \bar{k},
\end{align*}
where $\bar{k}$ is $k$ modulo $p$. The map satisfies $\log_{\xi,p}(xy) = \log_{\xi,p}(x) + \log_{\xi,p}(y)$ for all $x,y \in \mathcal{O}_K \setminus \mathfrak{q}$. In particular, $\log_{\xi,p}(x ^ m) = m\log_{\xi,p}(x)$ for $m \in \mathbb{N}$. Consequently, $\log_{\xi,p} (x ^ p) = 0$ and thus this map provides a way to detect $p$-th powers.

\begin{remark}
	For $x \in \mathcal{O}_K$ not a $p$-th power and a randomly chosen prime ideal $\mathfrak{q}$ with $p \mid N (\mathfrak{q}) - 1$, it can be taken as reasonable heuristic that $\log_{\xi,p}(x) \neq 0$ with probability $(p-1)/p$ (see Section \ref{heuristic}). 
\end{remark}

To define this map in Lean, we first introduce a more general version for a finite commutative ring $R$ with cyclic unit group generated by an element \lstinline|ξ : R|. This is asserted with the assumption \lstinline[mathescape]|h : IsPrimitiveRoot ξ (Fintype.card R$^{\times}$)|. We define a total map \href{https://github.com/alainchmt/CertifyingInvariantsNF/blob/v1/IdealArithmetic/Saturation/LogMatrix.lean#L72}{ \lstinline|LogFiniteRing h p : R → ZMod p|} that sends a unit $x = \xi ^ k$ to $\bar{k} \in \mathbb{Z}/p\mathbb{Z}$ and nonunits to $0$.

For a commutative ring $S$ and a ring homomorphism \lstinline|φ : S →+* R|, the composition
\lstinline|fun i => LogFiniteRing h p (φ i)| defines a map of type \lstinline|S → ZMod p|. In the special case where $S = \mathcal{O}_K$ and $\varphi : \mathcal{O}_K \to \mathcal{O}_K/\mathfrak{q}$ is the quotient map, this composition corresponds to $\log_{\xi, p}$ extended to a total function on $\mathcal{O}_K$ sending elements in $\mathfrak{q}$ to $0$. To certify statements of the form $\log_{\xi, p}(x) = \overline{t}$ 
 we need to verify identities of the form $\xi^k = x$ in the finite field $\mathcal{O}_K/\mathfrak{q}$. However, working directly with this quotient in Lean is inconvenient for automatic computation since equality depends on ideal membership.
An alternative is to transport the corresponding identities to another type \lstinline|(F : Type*) [CommRing F]| such that $F \cong \mathcal{O}_K/\mathfrak{q}$ and which has better computational support. Currently, Mathlib supports automatic arithmetic in finite fields of prime order via the type \lstinline|ZMod q|. Therefore, we restrict ourselves to prime ideals $\mathfrak{q}$ with inertia degree one (i.e. $N(\mathfrak{q}) = q$, with $q$ a prime number), so that $\mathcal{O}_K / \mathfrak{q} \cong \mathbb{Z}/q\mathbb{Z}$.

\begin{remark}
	This restriction is not a big limitation: the proportion of prime ideals (ordered by norm) with inertia degree one among all prime ideals $\mathfrak{q}$  such that $p \mid N(\mathfrak{q}) -1$ tends to $1$ (see Section \ref{sec:existancesat}).
\end{remark}

Let $\mathfrak{q}$ be a prime ideal with $N(\mathfrak{q}) = q$ a prime number, so that $\mathcal{O}_K / \mathfrak{q} \cong \mathbb{Z}/q\mathbb{Z}$. Using the \lstinline|LogFiniteRing| construction with \lstinline|ZMod q| as the finite ring, and composing with the reduction map $\mathcal{O}_K \to \mathbb{Z}/q\mathbb{Z}$, we \href{https://github.com/alainchmt/CertifyingInvariantsNF/blob/v1/IdealArithmetic/Saturation/PrincipalityCertificate.lean#L71}{define} a specialized version of $\log_{\xi, p} : \mathcal{O}_K \to \mathbb{F}_p$ as 

\begin{lstlisting}[mathescape]
def LogFiniteZMod {$\mathcal{O}$ : Type*} [CommRing $\mathcal{O}$] {q ξ : ℕ} [hq : Fact (Nat.Prime q)] {$\mathfrak{q}$ : Ideal $\mathcal{O}$} 
    (hcard : Nat.card ($\mathcal{O}$ / $\mathfrak{q}$) = q) (h : IsPrimitiveRoot ($\xi$ : ZMod q) (q - 1)) (p : ℕ) :
	   $\mathcal{O}$ → ZMod p := by ...
\end{lstlisting}
The definition is valid for any commutative ring $\mathcal{O}$ and ideal $\mathfrak{q}$ of $\mathcal{O}$ with a residue ring of prime order. Observe that the primality of $\mathfrak{q}$ is an immediate consequence of its residue ring having prime order. Furthermore, since we include the instance condition that $q$ is prime, we know that the cardinality of \lstinline[mathescape]|(ZMod q)$^\times$| is $q - 1$. The assumption \lstinline|h| thus asserts that $\bar{\xi}$ generates $(\mathbb{Z}/q\mathbb{Z})^{\times}$. For an element $x \in \mathcal{O} \setminus \mathfrak{q}$, the following certifies that $\log_{\xi,p}(\bar{x}) = t$ with $ t \in \mathbb{F}_p$.

\paragraph{Certificate of discrete logarithms}
\leavevmode
\begin{multicols}{2}
	\begin{itemize}
		\item An integer $m$;
		\item a natural number $k$.
	\end{itemize}
\end{multicols}

Verification of the certificate amounts to checking that
\begin{multicols}{3}
	\begin{enumerate}[label=(\roman*)]
		\item $x - m \in \mathfrak{q} $ \label{ver : mem};
        \item $\xi ^ k = \bar{m} $ in $\mathbb{Z}/q\mathbb{Z}$; \label{ver : mem2}
        \item $ \bar{k} = t$ in $\mathbb{F}_p$. \label{ver : mem3}
	\end{enumerate}
\end{multicols}
The first verification statement proves that $x \to \bar{m}$ under the reduction map $\mathcal{O} \to \mathbb{Z}/q\mathbb{Z}$. Statement \ref{ver : mem2} shows that $k$ is the discrete logarithm of $\bar{m}$ with base $\xi$, and \ref{ver : mem3} verifies the reduction modulo $p$. 

In order to prove a goal of the form \lstinline|LogFiniteZMod hcard h p x = t | with \lstinline[mathescape]|x : $\mathcal{O}$| and \lstinline|t : ZMod p|, we introduce the structure \href{https://github.com/alainchmt/CertifyingInvariantsNF/blob/v1/IdealArithmetic/Saturation/PrincipalityCertificate.lean#L84}{ \lstinline|DiscreteLogCertificate hcard h p x t|}, which 
includes in its fields the previous certificate with proofs of the verification statements. The coordinate tuples of the relevant elements are included, and the membership statement \ref{ver : mem} can be proven using a membership certificate for ideals. 

\begin{remark}
	To specify the exponent $k$ such that $\xi ^ k = m$ in $\mathbb{Z}/q\mathbb{Z}$, it is necessary to solve a discrete logarithm problem. We delegate this task to a CAS. Although the security of many cryptographic systems relies on the intractability of this computation \cite[Section 3.6]{handbook}, this will not be an issue for us since the values for $q$ are of moderate size in practice.
\end{remark}

A key step in the above discussion is proving \lstinline[mathescape]|IsPrimitiveRoot $\xi$ (q - 1)|, i.e., that \lstinline[mathescape]|$\xi$ : ZMod q| has multiplicative order $q - 1$ in \lstinline|ZMod q|. To automate this, we introduce a structure \href{https://github.com/alainchmt/CertifyingInvariantsNF/blob/v1/IdealArithmetic/Saturation/PrincipalityCertificate.lean#L33}{ \lstinline|IsOrderOf x n|}, where \lstinline|x : G| lies in a monoid $G$ and \lstinline|n : ℕ|. This structure certifies that the order of $x$ in $G$ is $n$ by bundling together the factorization data of $n$ together with proofs that $x^{n/p} \neq 1$ for every prime $p$ dividing $n$. In our setting, $G$ is \lstinline|ZMod q| with its multiplicative structure.
To handle expressions like \lstinline|(a : ZMod q) ^ n|, we define in Lean the function \href{https://github.com/alainchmt/CertifyingInvariantsNF/blob/v1/IdealArithmetic/Computation/ExponentiationZMod.lean#L70}{\lstinline|squareAndMultiply|}, implementing the tail-recursive \cite[Chapter 8.1]{funprog} square-and-multiply exponentiation algorithm in \lstinline|ZMod q|, and a tactic \href{https://github.com/alainchmt/CertifyingInvariantsNF/blob/v1/IdealArithmetic/Computation/ExponentiationZMod.lean#L81}{ \lstinline|zmod_pow|} that automatically simplifies such exponentiations in goals. For \href{https://github.com/alainchmt/CertifyingInvariantsNF/blob/v1/IdealArithmetic/Computation/ExponentiationZMod.lean#L102}{example}, it quickly solves:
\begin{lstlisting}
	example : (3 : ZMod 10020313) ^ 11002324556787 ≠ 1  := by zmod_pow
\end{lstlisting}

\subsection{The logarithmic matrix}
Our objective in this section is to prove that the collection $\{([J_i], n_i)\}_{i}$ is $p$-saturated. We combine information on the unit group of $\mathcal{O}_K$ with the logarithmic map introduced in the previous section. 

Let $(r_1, r_2)$ be the signature of $K$ and $\omega(K)$ the size of the torsion subgroup $\mu(\mathcal{O}_K)$ of $\mathcal{O}_K ^ {\times}$. From Dirichlet's unit theorem \eqref{eq : dirichlet}, we know $\mathcal{O}_K ^ {\times}$ decomposes as the direct product of a cyclic finite group of size $\omega(K)$ and a free part of rank $r_1 + r_2  - 1$. The quotient $\mathcal{O}_K ^ {\times} / (\mathcal{O}_K ^ {\times}) ^ p $ can be regarded as an $\mathbb{F}_p$-vector space. Its dimension equals $r_1 + r_2 -1$ when $p \nmid \omega(K)$, and $r_1 + r_2$ when $p \mid \omega(K)$. Let $\bar{u}_1, \ldots, \bar{u}_t$ be an $\mathbb{F}_p$-basis of $\mathcal{O}_K ^ {\times} / (\mathcal{O}_K ^ {\times}) ^ p $. In particular, this means that every unit $\gamma \in \mathcal{O}_K^{\times}$ can be written as $\gamma = \eta ^ p\prod_{i}^t u_i^{m_i}$ with $m_i \in \mathbb{Z}$ and $\eta \in \mathcal{O}_K^{\times}$. The following theorem underlies a certificate of $p$-saturation.

\begin{theorem}\label{theo: log}
	Let $(J_1, \ldots, J_k)$ be nonzero ideals of $\mathcal{O}_K$ such that $J_i^{n_i} = \langle \alpha_i\rangle $ for some $n_i \in \mathbb{Z}_{>0}$ and $\alpha_i \in \mathcal{O}_K$. Consider the subcollection $\alpha_{\psi_1}, \ldots, \alpha_{\psi_d}$ consisting of those $\alpha_{\psi_i}$ for which $p$ divides $n_{\psi_i}$.

	Let $u_1, \ldots, u_t$ be a collection of units in $\mathcal{O}_K$, and let $\mathfrak{q}_1, \ldots, \mathfrak{q}_{r}$ be a collection of prime ideals of $\mathcal{O}_K$ with $r \geq t + d$ and with the property that $\alpha_{\psi_i} \not \in \mathfrak{q}_l$ for every $i$ and every $l$. For each $l$, pick a primitive root $\xi_l \in \mathcal{O}_K/\mathfrak{q_l}$, and form the $r \times (t + d)$ matrix over $\mathbb{F}_p$ given by 
	\begin{align}\label{logmatrix}
		\begin{pmatrix}
			\log_{\xi_1, p} (u_1) & \dots & \log_{\xi_1, p} (u_{t}) & \log_{\xi_1, p} (\alpha_{\psi_1}) & \dots & \log_{\xi_1, p} (\alpha_{\psi_d}) \\
			\vdots & \ddots & \vdots & \vdots & \ddots & \vdots\\
			\log_{\xi_{r}, p} (u_1) & \dots & \log_{\xi_{r}, p} (u_{t}) & \log_{\xi_{r}, p} (\alpha_{\psi_1}) & \dots & \log_{\xi_{r}, p} (\alpha_{\psi_d})
		\end{pmatrix}.
	\end{align}
	Suppose that the units $u_i$ are chosen so that every $\gamma \in \mathcal{O}_K^{\times}$ can be written as $\gamma = \eta ^ p\prod_{i}^t u_i^{m_i}$ with $m_i \in \mathbb{Z}$ and $\eta \in \mathcal{O}_K^{\times}$. Under this assumption, if the logarithmic matrix \eqref{logmatrix} has full rank, then the collection $\{([J_i], n_i)\}_{i}$ is $p$-saturated. 
\end{theorem}

\begin{proof}
	Recall the characterization of $p$-saturation in \eqref{eq : psat_ideal}. Suppose $a \in \mathbb{N}^k$ and $b \in \mathcal{O}_K$ satisfy $\prod_i J_i ^{a_i} = \langle b \rangle$ and $n_i \mid pa_i$ for every $i$. For the sake of contradiction, assume there is $j$ such that $n_j \nmid a_j$.

	We have that $\langle b \rangle ^ p = \prod_i J_i^{pa_i} = \prod \langle \alpha_i\rangle^{r_i}$ with $r_i$ such that $pa_i = r_i n_i$. This implies that $b^p = \gamma \prod_i \alpha_i^{r_i}$ with $\gamma \in \mathcal{O}_K^{\times}$. Note that if $p \nmid n_i$ then $p \mid r_i$. Furthermore, $\gamma = \eta ^ p \prod _i u_i ^ {m_i} $ for some $m_i \in \mathbb{Z}$ and $\eta \in \mathcal{O}_K^{\times}$. Therefore, we have $b^p = y ^ p \prod_i u_i ^ {m_i} \prod \alpha_{\psi_i}^{r_{\psi_i}}$ with some element $y \in \mathcal{O}_K$. This implies that $(b/y) ^ p \in \mathcal{O}_K$, and since $\mathcal{O}_K$ is integrally closed, we can conclude that $b/y \in \mathcal{O}_K$. Thus, $\prod_i u_i ^ {m_i} \prod \alpha_{\psi_i}^{r_{\psi_i}} = x ^ p$ for some element $x \in \mathcal{O}_K$. Taking logarithms, we see that for $i \in \{1, \ldots, r\}$ it holds that
	\begin{align*}
		m_1 \log_{\xi_i, p}(u_1) + \ldots + m_t \log_{\xi_i, p}(u_t) + r_{\psi_1} \log_{\xi_i, p}(\alpha_{\psi_1}) + \ldots + r_{\psi_d} \log_{\xi_i, p}(\alpha_{\psi_d}) = 0.
	\end{align*}
	Thus, $(m_1, \ldots, m_t,r_{\psi_1}, \ldots, r_{\psi_d}) \in \mathbb{F}_p ^ {t + d}$ is in the kernel of matrix \eqref{logmatrix}. Since $n_j \nmid a_j$ we must have that  $p \nmid r_j$ and $p \mid n_j$, hence this tuple is not the zero tuple. We conclude that the kernel of matrix \eqref{logmatrix} is nontrivial, which is a contradiction because we assumed the matrix is of full rank.
\end{proof}

\begin{remark}
	The assumption that $\mathcal{O}_K$ is integrally closed is not required if one imposes additional conditions on the prime ideals by demanding that $\alpha_i \notin \mathfrak{q}_l$ for all $i$ and all $l$. Nevertheless, this would make our search for a certificate more restrictive and would introduce additional proof obligations.
\end{remark}

For a commutative ring $\mathcal{O}$, a collection of prime numbers $q_i$, ideals $\mathfrak{q}_i$ of $\mathcal{O}$ such that $\# (\mathcal{O} / \mathfrak{q}_i) = q_i$, and primitive roots $\xi_i \in \mathbb{Z}/q_i\mathbb{Z}$, we define the logarithmic matrix \href{https://github.com/alainchmt/CertifyingInvariantsNF/blob/v1/IdealArithmetic/Saturation/PrincipalityCertificate.lean#L163}{ \lstinline|MatrixLogZMod|} associated to a collection of elements $x_j \in \mathcal{O}$ as the matrix over $\mathbb{F}_p$ whose $ij$-th entry the value of $\log_{\xi_i, p}(x_j)$, where $\log_{\xi_i, p}: \mathcal{O} \to \mathbb{F}_p$ is the total map implemented by \lstinline|LogFiniteZMod|.

We formalize the corresponding version of Theorem \ref{theo: log} as \href{https://github.com/alainchmt/CertifyingInvariantsNF/blob/v1/IdealArithmetic/Saturation/PrincipalityCertificate.lean#L192}{ \lstinline|pSaturated_of_full_rank_matrixLogZMod|}, for an integral domain $\mathcal{O}$ which is integrally closed and $p$ is such that \lstinline|∀ i, p ∣ q i - 1|. The tuple of ideals is given as \lstinline[mathescape]|J : κ → Ideal $\mathcal{O}$| with the property that \lstinline|h : ∀ i, (J i) ^ (n i) = Ideal.span {a i}| with tuples \lstinline|n : κ → ℕ| and \lstinline[mathescape]|a : κ → $\mathcal{O}$|. The function \lstinline|ψ : γ → κ| extracts the subcollection of indices such that $p \mid n_{\psi_i}$. The collection of units is given as a tuple \lstinline[mathescape]|u : ι → $\mathcal{O}$| with the assumption \lstinline|hu : ∀ i, IsUnit (u i)|. The property that they generate the unit group $\mathcal{O}^{\times}$ up to $p$-th powers is stated as:
\begin{lstlisting}[mathescape]
	hugen : ∀ (w : $\mathcal{O}^{\times}$), (∃ (e : ι → ℕ) , (∃ (t : $\mathcal{O}^{\times}$) , w = (∏ (i : ι), (u i) ^ (e i)) * t ^ p))
\end{lstlisting}
The tuple \lstinline|(Sum.elim (fun i => u i) (fun i => a (ψ i)))| corresponds to $(u_1, \ldots, u_t, a_{\psi_1}, \ldots, a_{\psi_d})$ and has type \lstinline[mathescape]|ι ⊕ γ → $\mathcal{O}$|, indexed by disjoint union of \lstinline|ι| and \lstinline|γ|. The resulting logarithmic matrix has type \lstinline|Matrix τ (ι ⊕ γ) (ZMod p)| and being full-rank means it has rank $\# \iota + \# \gamma$.

On the way to formally proving the theorem above, we first prove a more general \href{https://github.com/alainchmt/CertifyingInvariantsNF/blob/v1/IdealArithmetic/Saturation/LogMatrix.lean#L439}{version} of Theorem \ref{theo: log} for an integrally closed domain $S$, using the \lstinline|LogFiniteRing| function introduced before. Rather than working with a family of prime ideals with prime-order residue fields, we instead consider a family of types \lstinline|F : τ → Type*|, where each type \lstinline|F i| is a finite commutative ring with cyclic unit group and with a morphism \lstinline|φ i : S →+* F i|. This abstracts the role of the quotient rings and the reduction maps, allowing greater flexibility in how they are represented.

\subsection{A certificate for \texorpdfstring{$p$}{p}-saturation}\label{sec:certunits}
To apply Theorem \ref{theo: log}, we must show that a collection of units $u_i$ generates $\mathcal{O}_K^{\times}$ modulo $p$-th powers. This can be deduced from the full-rank condition on the logarithmic matrix \eqref{logmatrix} as follows.

Let $\mu(\mathcal{O}_K)$ be the torsion subgroup of $\mathcal{O}_K^{\times}$ so that $U_K = \mathcal{O}_K^{\times}/\mu(\mathcal{O}_K)$ is free abelian of rank $r_1 + r_2 - 1$. Assume for the moment that $p \nmid \omega(K)$, so that every torsion unit is a $p$-th power. Let $t = r_1 + r_2 -1$ and suppose that a collection $u_1, \ldots, u_t$ in $\mathcal{O}_K^{\times}$ is such that the matrix \eqref{logmatrix} is of full rank. This, in particular, implies that the submatrix consisting of the first $t$ columns
\begin{align} \label{logmatrixunits}
	\begin{pmatrix}
		\log_{\xi_1, p} (u_1) & \dots & \log_{\xi_1, p} (u_{t}) \\
		\vdots & \ddots & \vdots\\
		\log_{\xi_{r}, p} (u_1) & \dots & \log_{\xi_{r}, p} (u_{t})
	\end{pmatrix}
\end{align}
is also of full rank. This can be used to show that the collection $u_i$ is $\mathbb{Z}$-linearly independent in $U_K$. In Lean, under weaker assumptions on an integral domain $\mathcal{O}$, and assuming that \lstinline[mathescape]| ¬ p ∣ Nat.card (CommGroup.torsion $\mathcal{O}^{\times}$)| and \lstinline|(MatrixLogZMod p hcard u hr).rank = Fintype.card ι|, with units \lstinline[mathescape]|u : ι → $\mathcal{O}$| such that \lstinline|hu : ∀ i, IsUnit (u i)|, we \href{https://github.com/alainchmt/CertifyingInvariantsNF/blob/v1/IdealArithmetic/Saturation/PrincipalityCertificate.lean#L219}{establish} that
\begin{lstlisting}[mathescape]
	LinearIndependent ℤ 
     (fun i => Additive.ofMul (QuotientGroup.mk (s := (CommGroup.torsion $\mathcal{O}^{\times}$)) (hu i).unit))
\end{lstlisting}
The term \lstinline|(hu i).unit| is the unit \lstinline|u i| as a term of type $\mathcal{O}^{\times}$. The map \lstinline|QuotientGroup.mk| sends it to its image in $\mathcal{O}^{\times}/\mu(\mathcal{O})$, and \lstinline|Additive.ofMul| identifies this multiplicative group with its additive version.

The image of $\{u_1, \ldots, u_t\}$ in $U_K$ generates a subgroup $W$ of rank $r_1 + r_2 - 1$, so the index $[U_k : W]$ is finite. From the full rank condition of \eqref{logmatrixunits}, it can be shown that $p \nmid [U_k : W]$, from which it follows that the images of the $u_i$ generate $\mathcal{O}_K^{\times}/ (\mathcal{O}_K ^ {\times})^p$.

Instead of employing linear algebra algorithms for computing the rank of a matrix, we certify that a $r \times t$ matrix has full rank by exhibiting an invertible $t \times t$ submatrix and providing its inverse. The previous discussion, together with this last observation, gives rise to the following certificate. 

\paragraph{Certificate for unit group generators modulo $p$-th powers}
For an integral domain $\mathcal{O}$, let $\mu(\mathcal{O})$ be the torsion subgroup of $\mathcal{O}^{\times}$. Let $p$ be a prime number with $p \nmid \# \mu(\mathcal{O})$. Suppose that $\mathcal{O}^{\times}/ \mu(\mathcal{O})$ has rank $t$ as a $\mathbb{Z}$-module. Let $u_1, \ldots, u_t$ be elements of $\mathcal{O}$. A certificate that the $u_i$ are units such that their image generates $\mathcal{O}^{\times}/ (\mathcal{O}^{\times}) ^ p$ consists of the following data:
\leavevmode
\begin{multicols}{2}
	\begin{itemize}
		\item A natural number $r$ with $t \leq r$; 
		\item a tuple of prime numbers $(q_1, \ldots, q_r)$;
		\item a tuple of natural numbers $(\zeta_1, \ldots, \zeta_r)$;
		\item a tuple of ideals $(\mathfrak{q}_1, \ldots, \mathfrak{q}_r)$ of $\mathcal{O}$;
		\item an $r \times t$ matrix $M$ over $\mathbb{F}_p$;
		\item a $t \times t$ matrix $M'$ over $\mathbb{F}_p$.
	\end{itemize}
\end{multicols}

Verification of the certificate amounts to checking that 
\begin{multicols}{2}
	\begin{enumerate}[label=(\roman*)]
		\item the element $u_i$ is a unit for every $i$;
		\item $p \mid q_i - 1$ for every $i$;
		\item $\bar{\zeta}_i$ is a primitive root of $\mathbb{Z}/q_i\mathbb{Z}$ for every $i$;
		\item $\# \mathcal{O}/\mathfrak{q}_i = q_i$ for every $i$;
		\item $M_{ij} = \log_{\zeta_i, p}(u_j)$ for every $i$ and $j$;
		\item $M_{[t]} * M' = \operatorname{Id}_t$. 
	\end{enumerate}
\end{multicols}
Here, $M_{[t]}$ denotes the submatrix of $M$ formed by the first $t$-rows and $\operatorname{Id}_t$ the identity $t \times t$ matrix. Statement (iv) directly proves that the $\mathfrak{q}_i$ are prime ideals.
Statements (ii), (iii), and (iv) together assert that $\log_{\zeta_i,p}(u_j)$ is defined for every $i$. Statement (v) states that $M$ is the logarithmic matrix \eqref{logmatrixunits}. Verifying (vi) amounts to a simple matrix multiplication and proves that $M$ is of full rank without the need of gaussian elimination-type algorithms.  From the discussion above, this shows that the $u_i$ is a system of independent units, the subgroup generated by the images of the $u_i$ in $\mathcal{O}^{\times} / \mu(\mathcal{O})$ is $p$-maximal, and they generate $\mathcal{O}^{\times}$ modulo $p$-th powers.
We \href{https://github.com/alainchmt/CertifyingInvariantsNF/blob/v1/IdealArithmetic/Saturation/PrincipalityCertificate.lean#L306}{formalize} this certification scheme in Lean as
\begin{lstlisting}[mathescape]
	structure pMaximalUnitsCertificateNDvdT ($\mathcal{O}$ : Type*) [CommRing O] (p : ℕ)   where ...
\end{lstlisting}

It includes as part of its fields a proof of \lstinline[mathescape]| ¬p ∣ Nat.card (CommGroup.torsion $\mathcal{O}^{\times}$)|, which is the subject of Section \ref{sec:torsiondiv}, a proof that the $\mathbb{Z}$-rank of \lstinline[mathescape]|Additive ($\mathcal{O}^{\times}$ /(CommGroup.torsion $\mathcal{O}^{\times}$))| is bounded by $t$, together with the tuple of units $u_i$ and the certifying data and verification statements above.

For a subring $\mathcal{O}$ of a number field $K$ equipped with a multiplication table and such that $\mathcal{O}= \mathcal{O}_K$, we can construct a term of type \lstinline[mathescape]|pMaximalUnitsCertificateNDvdT $\mathcal{O}$ p| by first finding a system of $\mathbb{Z}$-independent units $u_i$ using a CAS, and representing them in Lean as a tuple of coordinates with respect to an integral basis. Statement (i) can be proven by providing the inverse $u_i ^{-1}$ and verifying $u_i u_i ^{-1} = 1$ using the multiplication table.
Statements (ii) and (vi) can be solved by \lstinline|decide|, while (iii) and (v) are handled with the \lstinline|IsOrderOf| and \lstinline|DiscreteLogCertificate| structures, respectively, discussed in Section \ref{sec : disclog}. The statement (iv) is verfied as described in Section \ref{sec : certnorm}. Mathlib already includes the fact that the ring of integers of a number field modulo torsion is a finite and free $\mathbb{Z}$-module of rank equal to the number of infinite places minus one, and transporting this to $\mathcal{O}$ is straightforward.

Given a term of type \lstinline[mathescape]|C : pMaximalUnitsCertificateNDvdT $\mathcal{O}$ p|, a proof that the units $u_i$ are $\mathbb{Z}$-linearly independent in $\mathcal{O}^{\times}/\mu(\mathcal{O})$ is given as \href{https://github.com/alainchmt/CertifyingInvariantsNF/blob/v1/IdealArithmetic/Saturation/PrincipalityCertificate.lean#L344}{
\lstinline|units_linear_independent_of_pMaximalUnitsCertificateNDvdT C|}, and a proof that every unit of $\mathcal{O}$ can be written as a $p$-th power multiplied by a product of the units $u_i$ is given by \href{https://github.com/alainchmt/CertifyingInvariantsNF/blob/v1/IdealArithmetic/Saturation/PrincipalityCertificate.lean#L364}{ \lstinline|units_of_pMaximalUnitsCertificateNDvdT C|}.

\begin{remark}
	The rows of a full-rank $r \times t$ matrix can always be reordered so that the first $t$ form an invertible matrix. Consequently, the ideals \(\mathfrak{q}_i\) may be reordered so that \(M_{[t]}\) is invertible.
\end{remark}

We introduce the parameter $r$ in the certificate above, rather than fixing $r = t$, so that we can have the option to augment the columns of $M$ to form the matrix in \eqref{logmatrix}, obtaining a certificate that the collection $\{([J_i], n_i)\}_{i}$ is $p$-saturated. Indeed, we extend the previous certificate to one for $p$-saturation.

\paragraph{Certificate for $p$-saturation}
Let $(J_1, \ldots, J_m)$ be a tuple of ideals of $\mathcal{O}$ and $(n_1, \ldots, n_m) \in \mathbb{Z}_{>0}^m$. Let $d$ be the number of indices $i$ such that $ p \mid n_i$.
To certify that $\{([J_i], n_i)\}_{i}$ is $p$-saturated, extend the previous certificate for unit group generators modulo $p$-th powers by setting $r = t + d$ and including the following data:
\leavevmode
\begin{multicols}{2}
	\raggedcolumns
	\begin{itemize}
		\item A tuple $(\alpha_1, \ldots, \alpha_m) \in \mathcal{O}^m$;
		\item a tuple $(\psi_1, \ldots, \psi_d) \in \{1,\ldots,m\}^d$;
		\item a tuple $(\nu_1, \ldots, \nu_m) \in \{1, \ldots, d\} ^ m$;
		\item an $r \times d$ matrix $N$ over $\mathbb{F}_p$;
		\item an $r \times r$ matrix $N'$ over $\mathbb{F}_p$. 
	\end{itemize}
\end{multicols}
Recall that $\mathcal{O}/\mathfrak{q}_i \cong \mathbb{Z}/q_i\mathbb{Z}$. For $x \in \mathcal{O}$, let $\bar{x}$ be its image in $\mathbb{Z}/q_i\mathbb{Z}$. Verification of the certificate amounts to additionally checking that 
\begin{multicols}{2}
	\raggedcolumns
	\begin{enumerate}[label=(\roman*)]
		\item $J_i^{n_i} = \langle \alpha_i\rangle$ for every $i$;
		\item for every $i$ with $p \mid n_i$, $\psi_{\nu_i} = i$;
		\item $\overline{\alpha_{\psi_j}} \neq 0$ in $\mathbb{Z}/q_i\mathbb{Z}$ for every $i$ and $j$;
		\item $\alpha_i \neq 0$ for every $i$ with $p \nmid n_i$;
		\item $\log_{\zeta_i, p}(\alpha_{\psi_j}) = N_{ij}$ for every $i$ and $j$;
		\item $(M || N) * N' = \operatorname{Id}_r$.
	\end{enumerate}
\end{multicols}
The notation $(M || N)$ denotes the $r \times r$ matrix resulting from concatenating the matrix $M$, from the previous certificate, with the columns of $N$. By (v), it corresponds to the logarithmic matrix \eqref{logmatrix}. Statement (ii) ensures that $\psi$ indexes all $i$ with $p \mid n_i$. Statement (iii) proves that $\alpha_{\psi_j} \not \in \mathfrak{q}_i$, and thus $\log_{\zeta_i, p}(\alpha_{\psi_j})$ is defined. Additionally, it proves that $\alpha_{\psi_i} \neq 0$. Together with statements (i) and (iv), this guarantees that $J_i^{n_i}$, and hence $J_i$, is nonzero. Statement (vi) asserts that this matrix has full rank. Theorem \ref{theo: log} then implies that $\{([J_i], n_i)\}_{i}$ is $p$-saturated. We \href{https://github.com/alainchmt/CertifyingInvariantsNF/blob/v1/IdealArithmetic/Saturation/PrincipalityCertificate.lean#L562}{ implement} this certification scheme in Lean by extending \lstinline|pMaximalUnitsCertificateNDvdT|. 

\begin{lstlisting}[mathescape]
	structure pSaturatedClassGroupCertificateNDvdT {$\mathcal{O}$ : Type*} [CommRing $\mathcal{O}$] (p : ℕ) {m : ℕ}
	(J : Fin m → Ideal $\mathcal{O}$) (n : Fin m → ℕ) extends pMaximalUnitsCertificateNDvdT O p where ...
\end{lstlisting}

Given a term \lstinline|C : pSaturatedClassGroupCertificateNDvdT p J n|, if $\mathcal{O}$ is an integrally closed domain with $\mathcal{O}^{\times}/\mu(\mathcal{O)}$ free and finite as a $\mathbb{Z}$-module (all required as instance hypothesis), then the term \href{https://github.com/alainchmt/CertifyingInvariantsNF/blob/v1/IdealArithmetic/Saturation/PrincipalityCertificate.lean#L599}{\lstinline|pSaturated_of_CertificateNDvdT C|} is a proof of $p$-saturation of the collection $\{([J_i], n_i)\}_{i}$ using the formulation in \eqref{eq : psat_ideal}:
\begin{lstlisting}[mathescape]
	∀ a : Fin m → ℕ, ∀ b : $\mathcal{O}$, ∏ i, (J i) ^ (a i) = Ideal.span {b} → (∀ i, n i ∣ p * (a i)) → ∀ i, n i ∣ a i
\end{lstlisting}

For a subring $\mathcal{O}$ of a number field $K$ equipped with a multiplication table and such that $\mathcal{O}= \mathcal{O}_K$, we can prove statement (i) in the previous certification scheme with ideal multiplication described in Section \ref{sec : idealmul}. The statements (ii) and (vi) are proven with \lstinline|decide|. Statements (iii) and (v) are solved using the \lstinline|DiscreteLogCertificate|, which stores the value of $\overline{\alpha_{\psi_i}}$ in $\mathbb{Z}/q_j\mathbb{Z}$. Finally, (iv) is solved using the coordinate representation of the $\alpha_i$.

\begin{remark}
	Let $K$ be a number field of signature $(r_1, r_2)$ and $\{u_1, \ldots, u_{t}\}$, with $t = r_1 + r_2 - 1$, a collection of units whose image generate $U_K$. Suppose $\operatorname{Cl}(\mathcal{O}_K) = \langle [J_1] \rangle \times \ldots \times \langle [J_k] \rangle $ and that the order of $[J_i]$ in $\operatorname{Cl}(\mathcal{O}_K)$ is equal to $n_i$.
In Section \ref{sec:existancesat}, we see that the certificate described above always exists. By employing a CAS to search among the prime ideals of inertia degree one lying above primes $q$ satisfying $p \mid q - 1$, we efficiently find these certificates in practice .
\end{remark}

In the previous discussion, we focused on the case where $p \nmid \omega(K)$, so that every torsion unit is a $p$-th power. However, the certificates can be readily adapted to the case where $p \mid \omega(K)$. The main difference is that the matrix \eqref{logmatrixunits} is constructed with $t = r_1 + r_2$, and the last unit $u_t$ is taken to be a  generator $v$ of the cyclic group $\mu(\mathcal{O}_K)/\mu(\mathcal{O}_K)^p$. More generally, for a commutative ring $\mathcal{O}$ and prime $p$, the certifying structure \href{https://github.com/alainchmt/CertifyingInvariantsNF/blob/v1/IdealArithmetic/Saturation/PrincipalityCertificate.lean#L388}{ \lstinline[mathescape]|pMaximalUnitsCertificateDvdT $\mathcal{O}$ p|} incorporates the data of an element $v$ of $\mu(\mathcal{O})$ and a proof of $v^m = 1$ for an $m > 0$, and \href{https://github.com/alainchmt/CertifyingInvariantsNF/blob/v1/IdealArithmetic/Saturation/PrincipalityCertificate.lean#L639}{ \lstinline|pSaturatedClassGroupCertificateDvdT|} extends this to a certificate of $p$-saturation. The lemma \href{https://github.com/alainchmt/CertifyingInvariantsNF/blob/v1/IdealArithmetic/Saturation/PrincipalityCertificate.lean#L691}{ \lstinline|pSaturated_of_CertificateDvdT|} then proves $p$-saturation from such a certificate, where additional instance hypothesis such as \lstinline[mathescape]|[IsCyclic (CommGroup.torsion $\mathcal{O}^{\times}$)]|, asserting that $\mu(\mathcal{O})$ is cyclic, are required.

\subsection{The torsion divisibility condition}\label{sec:torsiondiv}
The structure \lstinline | pMaximalUnitsCertificateNDvdT| includes as part of its fields a proof of \lstinline[mathescape]| ¬p ∣ Nat.card (CommGroup.torsion $\mathcal{O}^{\times}$)|. We describe how we verify this.

Firstly, consider the case where $\mathcal{O}$ is the ring of integers of a number field $K$ of odd degree. Any defining polynomial of $K$ has a real root, and thus $K$ admits a real embedding $\sigma : K \to \mathbb{R}$. Any torsion unit in $K$ is injectively mapped to a torsion unit in $\mathbb{R}$. Consequently, $\# \mu(\mathcal{O}_K) = 2$. 
Within this context and when $p \neq 2$, we can immediately prove \lstinline[mathescape]| ¬p ∣ Nat.card (CommGroup.torsion $\mathcal{O}^{\times}$)|.

When $p = 2$ or when the degree of the number field is even, the above argument does not apply. We certify the divisibility condition using an ideal witness.
Let $\mathcal{O}$ be an integral domain with $\mu(\mathcal{O})$ finite. Suppose that $I$ is an ideal of $\mathcal{O}$ such that $\#\mathcal{O}/I = q$, with $q$ a prime number, and $p \neq q$. If $p \mid \# \mu(\mathcal{O})$, then there is an element $x \in \mathcal{O}^{\times}$ of order $p$. Since $q \neq p$, it can be shown that the multiplicative order of $\bar{x}$ in $\mathcal{O}/I$ remains equal to $p$. Therefore, $p \mid q -1$ as $\# (\mathcal{O}/I)^{\times} = q -1$.
This gives a certification strategy for $p \nmid \#\mu(\mathcal{O})$: provide an ideal $I$ of $\mathcal{O}$ and a prime number $q$, verify that $\#\mathcal{O}/I = q$ and that $p \neq q$ and $p \nmid q -1$. We formalize this as \href{https://github.com/alainchmt/CertifyingInvariantsNF/blob/v1/IdealArithmetic/Saturation/CertifyTorsionOrder.lean#L169}{ \lstinline|prime_not_dvd_torsion_of_not_dvd|} (together with a slightly more general variant), and apply it to the case where $\mathcal{O}$ is the ring of integers of a number field.

It can be shown (see Remark \ref{torsionexists}) that whenever $p \nmid \omega(K)$, an ideal $I$ of $\mathcal{O}_K$ with the required properties always exists. 

\section{Existence of the saturation certificate}\label{sec:existancesat}
Let $K$ be a number field of signature $(r_1, r_2)$, and let $u_1, \ldots, u_{t}$ be units whose image forms an $\mathbb{F}_p$-basis of $\mathcal{O}_K^{\times}/(\mathcal{O}_K^{\times}) ^ p$ (for instance, take a $\mathbb{Z}$-basis of $U_K = \mathcal{O}_K^{\times}/\mu (\mathcal{O}_K)$ with the inclusion of a generator of $\mu(\mathcal{O}_K)$ in case $p \mid \omega(K)$ ). Suppose $\operatorname{Cl}(\mathcal{O}_K) = \langle [J_1] \rangle \times \ldots \times \langle [J_k] \rangle $ and that the order of $[J_i]$ in $\operatorname{Cl}(\mathcal{O}_K)$ is equal to $n_i$. In particular $J_i ^ {n_i} = \langle \alpha_i \rangle $ for some $\alpha_i \in \mathcal{O}_K$ . In this section, we sketch a proof that, under the above assumptions, a collection of prime ideals $\mathfrak{q}_i$ satisfying the conditions of Theorem \ref{theo: log} making the matrix \eqref{logmatrix} full-rank always exists. Moreover, the $\mathfrak{q}_i$ may be chosen of inertia degree one. This shows that the certificates of Section \ref{sec:certunits} always exist.
The existence of such a collection was already established in \cite[Theorem 3.6.3]{hess} using a formulation in terms of $S$-units. Here, we give an alternative proof (both relying fundamentally on a theorem by Frobenius) that fits better with our presentation and offers explicit density estimates, and include a heuristic argument that estimates the complexity of finding such a collection. We did not formalize this result in Lean, but rather offer a meta-level guarantee that a formal certificate always exists. We begin with an observation.

\begin{lemma}\label{lemma:powers}
With $\alpha_i$ and $u_j$ as above. Suppose that $\prod_{i, p \mid n_i} \alpha_i ^ {m_i} \prod_{j} u_j^{t_j}$ is a $p$-th power in $\mathcal{O}_K$ for some $m_i \in \mathbb{N}$ and $t_j\in \mathbb{N}$. Then $p \mid m_i$ for all $i$ with $p \mid n_i$, and $p \mid t_j$ for all $j$.
\end{lemma}
\begin{proof}
Suppose $\prod_{i, p \mid n_i} \alpha_i ^ {m_i} \prod_{j} u_j^{t_j} = y^p$ for some $y \in \mathcal{O}_K$. Then $\prod_{i, p \mid n_i} J_i ^ {n_im_i} = \langle y \rangle ^ p$. For $p \mid n_i$, write $n_i = p s_i$, and thus $(\prod_{i, p \mid n_i} J_i ^ {s_im_i}) ^ p = \langle y \rangle ^ p$. It follows that $\prod_{i, p \mid n_i} J_i ^ {s_im_i} = \langle y \rangle $, which is principal. Consequently, from the assumed class group structure, we have $n_i \mid s_im_i$ for all $i$ with $p \mid n_i$. Hence, $p \mid m_i$ for all $i$. This implies that $\prod_j u_j ^ {t_j} = x^p$ for some $x \in \mathcal{O}_K$. Hence, $\prod_j \bar{u_j} ^ {t_j} = 1$ in $\mathcal{O}_K^{\times}/(\mathcal{O}_K^{\times}) ^ p$. It follows, from $\bar{u_1}, \ldots, \bar{u_{t}}$ being an $\mathbb{F}_p$-basis of $\mathcal{O}_K^{\times}/(\mathcal{O}_K^{\times}) ^ p$, that $p \mid t_j$ for every $j$.
\end{proof}

For a set of prime ideals $S$ of $\mathcal{O}_K$, the \emph{natural density} of $S$ is defined as
\begin{align*}
d(S) = \lim_{x \to \infty} \frac{\#\{ \mathfrak{q} \mid N(\mathfrak{q}) \leq x \text{ and } \mathfrak{q} \in S\}}{\#\{\mathfrak{q} \mid N(\mathfrak{q}) \leq x \}},
\end{align*}
if this limit exists. Strictly speaking, this is not a probability measure. It fails, for instance, to be $\sigma$-additive. However, intuitively, it might be thought of as the probability that a randomly selected prime ideal belongs to $S$. If a set has positive natural density, then it is necessarily infinite. A tool used to calculate the density of sets of prime ideals of certain shape is the Frobenius density theorem \cite{frob}. We refer to Appendix \ref{app:A} for a full statement of this theorem. We remark that one can also consider the \emph{analytic} (Dirichlet) density of a set of prime ideals, which equals the natural density if the latter exists. However, in this text we only use the natural density. 

\begin{lemma}\label{lemma:inerden}
	The natural density of the set of prime ideals of $\mathcal{O}_K$ of inertia degree one is equal to $1$.
\end{lemma}
\begin{proof}
This is a standard fact that follows from estimations using the Prime Ideal Theorem \cite[Chapter XV, Theorem 4]{lang3} to count prime ideals of $\mathcal{O}_K$ and primes in $\mathbb{Z}$. 
\end{proof}

 Let $p$ be a prime number and $\zeta_p$ a primitive $p$-th root of unity in an algebraic closure of $K$. Set $\delta_p = [K(\zeta_p) : K]$. It can be shown that $\delta_p$ divides $p-1$. 
 
\begin{lemma}\label{lemma:dens1}
Let $\alpha \in \mathcal{O}_K \setminus  (\mathcal{O}_K) ^ p$. Then the natural density of prime ideals of $\mathcal{O}_K$ for which $\bar{\alpha} \in \mathcal{O}_K/\mathfrak{q}$ is a $p$-th power is equal to $1-(p-1)/(p\delta_p)$.
\end{lemma}
\begin{proof}
Apply the Frobenius density theorem to the polynomial $X^p - \alpha$. See Appendix \ref{app:A}. 
\end{proof}

\begin{lemma}\label{lemma:rootun}
The natural density of the set  $ \{\mathfrak{q} \mid p | N(\mathfrak{q}) - 1\}$ of prime ideals of $\mathcal{O}_K$ equals $1/\delta_p$.
\end{lemma}
\begin{proof}
Apply the Frobenius density theorem to the $p$-th cyclotomic polynomial. See Appendix \ref{app:A}.
\end{proof}

\begin{remark}\label{torsionexists}
The previous lemma together with Lemma \ref{lemma:inerden} shows that if $p \nmid \omega(K)$, then the certificate of the torsion divisibility condition from Section \ref{sec:torsiondiv} always exists.
\end{remark}

\begin{lemma}\label{lemma:pth}
Let $\alpha \in \mathcal{O}_K$. Suppose that $\bar{\alpha} \in \mathcal{O}_K/\mathfrak{q}$ is a $p$-th power for all but a set of prime ideals of density $0$. Then $\alpha$ is a $p$-th power in $\mathcal{O}_K$.
\end{lemma}
\begin{proof}
Since the set of prime ideals $\mathfrak{q}$ of $\mathcal{O}_K$ such that $\bar{\alpha}$ is a $p$-th power is a set of density $1$, from Lemma \ref{lemma:dens1} we have that $\alpha \in (\mathcal{O}_K)^p$.
\end{proof}

\begin{theorem}\label{theo:exists}
Let $x_1, \ldots, x_n$ be a collection of elements in $\mathcal{O}_K$ with the property that if $\prod_{j} x_j ^ {m_j}$ is a $p$-th power in $\mathcal{O}_K$, with $m_j \in \mathbb{N}$, then $p \mid m_j$ for all $j$. Let $S$ be a set of prime ideals of $\mathcal{O}_K$ of natural density equal to $1$. Then, with $r \geq n$,  there exists a collection of prime ideals $\mathfrak{q}_1, \ldots, \mathfrak{q}_r$ in $S$ that do not contain any $x_j$ and with $p \mid N(\mathfrak{q}_i) - 1$ for all $i$, such that for any choice of primitive roots $\zeta_i \in (\mathcal{O}_K/\mathfrak{q}_i)^{\times}$, the logarithmic matrix $M = (\log_{\zeta_i, p} (x_j))_{i,j} \in \mathbb{F}_p ^ {r \times n}$ has full rank.
\end{theorem}
\begin{proof}
Note that for every $j$, we have $x_j \neq 0$. The number of ideals containing $x _j$ is finite (they are in one-to-one correspondence with the ideals of the finite ring $\mathcal{O}_K / \langle x_j  \rangle$ ). Enumerate the elements in the (infinite) set $ \{\mathfrak{q} \in S  \mid p | N(\mathfrak{q}) - 1 \text{ and } \forall j, x_j \not \in \mathfrak{q} \}$ as $\mathfrak{q}_1, \mathfrak{q}_2, \ldots$ and for each $i$ choose a primitive root $\zeta_i \in (\mathcal{O}_K/\mathfrak{q}_i)^{\times}$. Consider the sequence of matrices given by $M_k = (\log_{\zeta_i, p} (x_j))_{1 \leq i \leq k, j} \in \mathbb{F}_p ^ {k \times n}$ .  We have that $\ker M_k \subseteq \mathbb{F}_p ^ n $ and
\begin{align*}
\ker M_0 \supseteq \ker M_1 \supseteq \dots .
\end{align*}
Since $\ker M_i$ is a linear subspace of $\mathbb{F}_p^{n}$, the above chain of containments must stabilize. Let $r$ be such that $\ker M_r = \bigcap_{k = 1} ^ {\infty} \ker M_k$.  Let $v \in \ker M_r$. We have
\begin{align*}
 v \in \ker M_k \text{ for all }k  & \iff \sum_{j} v_j \log_{\zeta_i, p} (x_j) = 0 \text{ for all } i \\ & \iff \log_{\zeta_i, p}(\prod_{j} x_j ^ {v_j}) = 0 \text{ for all } i  \iff \overline{\prod_{j}x_j ^ {v_j}} \in (\mathcal{O}_K/\mathfrak{q}_i)^p \text{ for all } i.
\end{align*}
Furthermore, if $\mathfrak{q}$ is a prime ideal with $p \nmid N(\mathfrak{q}) -1$, then any element in $\mathcal{O}_K/\mathfrak{q}$ is a $p$-th power. Hence, we have that $ \overline{\prod_{j}x_j ^ {v_j}}$ is a $p$-th power modulo every prime ideal $\mathfrak{q} \in S$ such that $\forall j, x_j \not \in \mathfrak{q}$.  Since the set of prime ideals $\mathfrak{q}$ such that $\exists j, x_j \in \mathfrak{q}$ is finite and the complement $S^{c}$ has density $0$, we have by Lemma \ref{lemma:pth} that $\prod_{j}x_j ^ {v_j}$ is a $p$-th power in $\mathcal{O}_K$. By the assumption on the $x_j$, this implies that $p \mid v_j$ for every $j$. Thus $v=0$ and we conclude that $\ker M_{r}$ is trivial and thus the rank of $M_r$ is equal to $n$.
\end{proof}

\begin{theorem}
For any set $S$ of prime ideals of natural density $1$, there exists a collection of prime ideals contained in $S$ satisfying the conditions of Theorem \ref{theo: log} making the matrix \eqref{logmatrix}  full-rank.
\end{theorem}
\begin{proof}
Follows from Lemma \ref{lemma:powers} and Theorem  \ref{theo:exists}.
\end{proof}

\begin{corollary}
The certificate from Section \ref{sec:certunits} always exists.
\end{corollary}
\begin{proof}
Use the previous theorem and Lemma \ref{lemma:inerden}.
\end{proof}

\subsection{A heuristic analysis} \label{heuristic}
In what follows, we interpret the natural density of a set $A$ of prime ideals as the probability that a randomly selected prime ideal belongs to $A$, and write $P(A)$ for $d(A)$. If all the densities involved exist, this "probability" measure is finitely additive albeit not $\sigma$-additive. For simplicity, we will assume that the natural density of all the considered sets exists. Furthermore, for sets $A$ and $B$, we write $P(A \mid B)$ for the conditional probability defined as $P(A \cap B)/P(B)$.

Consider a set $S$ of prime ideals of density $1$ and fix $\alpha \in \mathcal{O}_K \setminus (\mathcal{O}_K)^p$. Using Lemma \ref{lemma:dens1} and Lemma \ref{lemma:rootun}, we have, for a randomly selected prime ideal $\mathfrak{q}$, that
\begin{multline}\label{eq:condprob}
P(\bar{\alpha} \in (\mathcal{O}_K/\mathfrak{q})^p \mid \mathfrak{q} \in S \text{ and } p | N(\mathfrak{q}) - 1)  = P(\bar{\alpha} \in (\mathcal{O}_K/\mathfrak{q})^p \mid p | N(\mathfrak{q}) - 1) \\ = \frac{P(\bar{\alpha} \in (\mathcal{O}_K/\mathfrak{q})^p ) - P (\bar{\alpha} \in (\mathcal{O}_K/\mathfrak{q})^p \mid p \nmid N(\mathfrak{q}) - 1) P(p \nmid N(\mathfrak{q}) - 1) }{P(p | N(\mathfrak{q}) - 1)}  \\
= \frac{(1- \frac{p-1}{p\delta_p}) - 1 \cdot (1 - \frac{1}{\delta_p})}{\frac{1}{\delta_p}} = \frac{1}{p}.
\end{multline}
With $P(\bar{\alpha} \in (\mathcal{O}_K/\mathfrak{q})^p \mid p \nmid N(\mathfrak{q}) - 1) = 1$, since when $p \nmid N(\mathfrak{q}) - 1$, every element of $\mathcal{O}_K/\mathfrak{q}$ is a $p$-th power.

Now fix a collection of elements $x_1, \ldots, x_n$ in $\mathcal{O}_K$ with the property that $\prod_i x_i ^ {m_i} \in (\mathcal{O}_K) ^ p$, with $m_i \in \mathbb{N}$, implies that $p \mid m_i$ for all $i$. In particular, if  $v \in \mathbb{F}_p^{n} \setminus \{0\}$, then $\prod_i x_i ^ {v_i} $ is not a $p$-th power in $\mathcal{O}_K$.
Let $S$ be a set of prime ideals with density $1$ and randomly sample independently $m$ prime ideals $\mathfrak{q}_1, \ldots, \mathfrak{q}_m$ from the set  $ \{\mathfrak{q} \in S  \mid p | N(\mathfrak{q}) - 1 \text{ and } \forall j, x_j \not \in \mathfrak{q} \}$. 
For any choice of primitive roots $\zeta_i \in (\mathcal{O}_K/\mathfrak{q}_i)^{\times}$ form the logarithmic matrix $M = (\log_{\zeta_i, p} (x_j))_{i,j} \in \mathbb{F}_p ^ {m \times n}$. We find that
\begin{align*}
P(\operatorname{rank} M < n) &= P(\exists v \in \mathbb{F}_p^n \setminus \{0\} \text{ such that }  \forall i,  \prod_{j}x_j ^ {v_j} \in (\mathcal{O}_K/\mathfrak{q}_i)^p )  \\
&  \stackrel{(i)}{\leq} \sum_{v \in \mathbb{F}_p ^ n \setminus \{0\}} P( \forall i,  \prod_{j}x_j ^ {v_j} \in (\mathcal{O}_K/\mathfrak{q}_i)^p) \\
&  \stackrel{(ii)}{=}  \sum_{v \in \mathbb{F}_p ^ n \setminus \{0\}} \prod_i P(\prod_{j}x_j ^ {v_j} \in (\mathcal{O}_K/\mathfrak{q_i})^p)  \stackrel{(iii)}{=}   \sum_{v \in \mathbb{F}_p ^ n \setminus \{0\}} \prod_i  \left(\frac{1}{p}\right) = \frac{p^n - 1}{p^m}.
\end{align*}
The inequality $(i)$ could be refined with a more careful handling of the linear subspaces of $\mathbb{F}_p^n$.  For $(ii)$ we crucially use that the sampling of prime ideals is independent, and for $(iii)$ we use (\ref{eq:condprob}).

The probability that $M$ is of full rank is then bounded below by $1 - \frac{p^n - 1}{p^m}$, and rapidly approaches $1$ as $m$ increases. Note that already for $m = n + 2$ the probability of success is larger than $1- 1/p^2 \geq 3/4 $. In this case, we need, on average, less than two independent attempts to get a matrix of full rank.

For the certificate in Section \ref{sec:certunits}, we choose $S$ to be the set of prime ideals of inertia degree one. In practice, to compute this certificate, the selection of the prime ideals $\mathfrak{q}_i$ is not truly random. We select them in increasing order with respect to their norm and continue adding rows to the matrix until it is of full rank. This seems to work well in practice.

\section{Formally verifying LMFDB entries} \label{sec : lmfdb}
The \emph{$L$-functions and modular forms database} \cite{lmfdb} (LMFDB) is a popular data base that contains data of, among other things, invariants of number fields. These include, for a collection of number fields, an integral basis, discriminant, signature, the class number, and the structure of the class group. The framework developed in this work allows us to formally verify some of these database entries.

In \cite{rings}, the ring of integers was certified in Lean for a collection of degree-3 and degree-5 number fields. 
Additionally, the discriminant of the degree-3 examples was also verified, but degree-5, however, was computationally infeasible at that time. We will illustrate our current methods by revisiting this degree-5 collection of number fields. In particular, we can now verify their discriminant efficiently, as well as certify other invariants such as the class group structure. 

We wrote a \href{https://github.com/alainchmt/CertifyingInvariantsNF/blob/v1/InvariantsNFLean.sage}{SageMath script} (which reuses some of the functions defined in \cite{rings} for the ring of integers verification), with entrypoint \texttt{generate\_invariants\_proof\_lean}, that, given a defining polynomial (integral and monic) for a number field $K$, an integral basis, and a naming string, automatically generates Lean files containing full proofs of the class group structure and the class number, together with proofs of the discriminant and signature of $K$. The files defining and certifying the prime ideals of norm below the Minkowski bound, as well as the relations with the set of ideals whose classes are generators of the cyclic factors of the class group, are divided into multiple independent batches, each dealing with the prime ideals above a certain interval of prime numbers. This enables parallelized type-checking if multiple cores are available, and considerably speeds up verification.

To test our tools in a variety of class group structures, we randomly selected from the LMFDB some degree-3 number fields unramified outside primes $p > 200$ and with absolute discriminant between 20,000 and 1,000,000: 100 with trivial class group, 100 with non-trivial cyclic class group, and 100 with non-cyclic class group. Using a MacBook Pro running a 4.6 GHz Apple M5 processor, 10-core CPU/10-core GPU and 24 GB of RAM, we successfully verified the listed invariants for all. The verification timings (wall-clock), are shown in Figure 1. All of these proofs are \lstinline|sorry|-free, only using the three standard axioms of Mathlib: \texttt{propext, choice, Quot.sound}. 

The degree 5 examples of \cite{rings} are the number fields in the LMFDB unramified outside 2,3, and 5, and with a non-monogenic ring of integers.
The Minkowski bound of these ranges between 24 and 18,792, resulting in verification timings (wall-clock) ranging from 30 seconds to 23 minutes (Figure 2)\footnote{The larger examples involve thousands of prime ideals split into hundreds of batches. Assembling the per-batch \lstinline|PrimesBelowBoundCertificateInterval| into \lstinline|PrimesBelowBoundCertificate| may still exceed Lean's preset maximum recursion depth and heartbeat limits, so we raise them both in the relevant file whenever the batch count exceeds 70. }. 
All these 142 degree-5 number fields have cyclic class groups of orders 1,2,3,4, or 5.

\begin{figure}[ht]
  \centering
  \begin{minipage}{0.49\linewidth}
    \centering
    \small{Figure 1}
    \includegraphics[width=\linewidth]{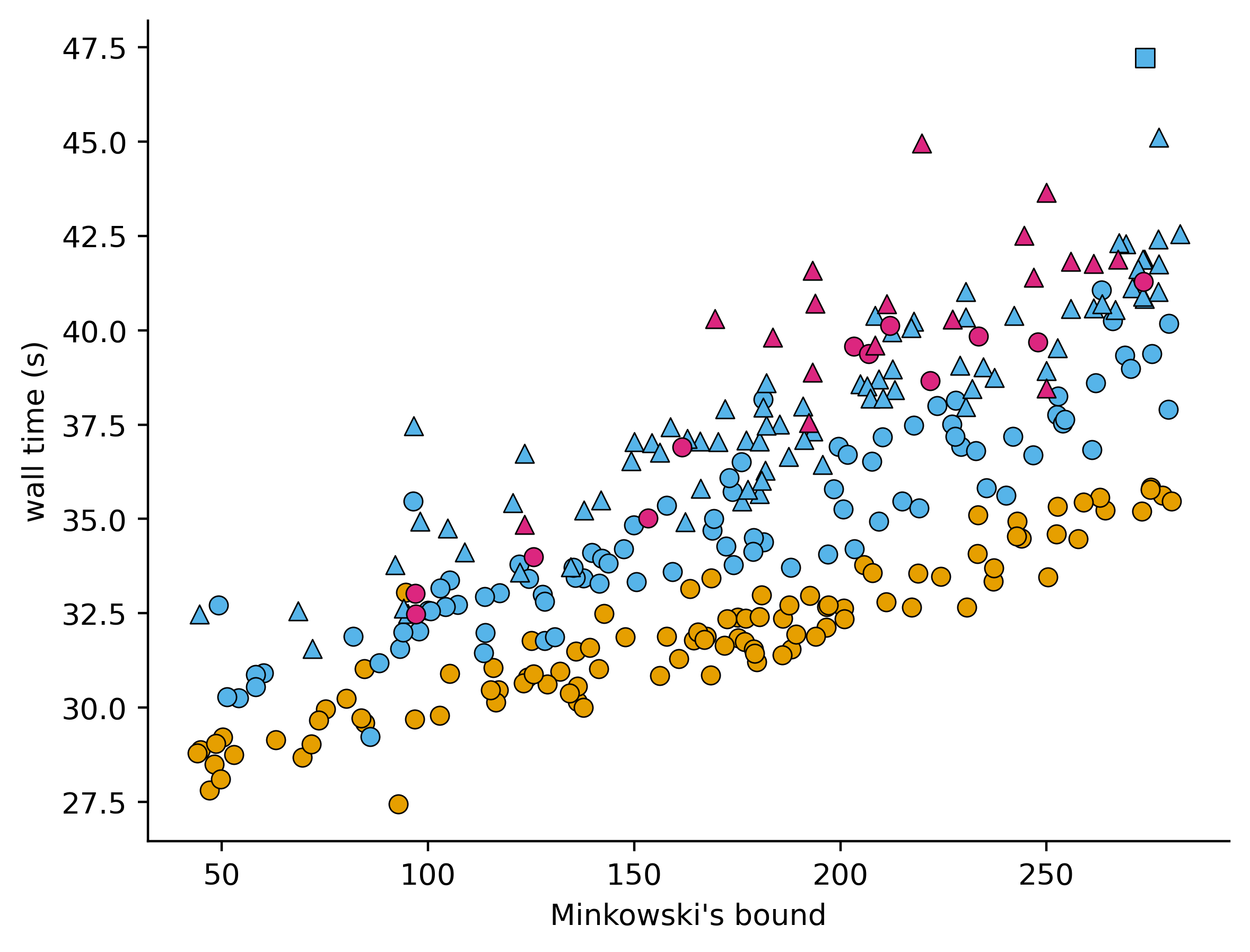}
  \end{minipage}\hfill
  \begin{minipage}{0.49\linewidth}
    \centering
    \small{Figure 2}
    \includegraphics[width=\linewidth]{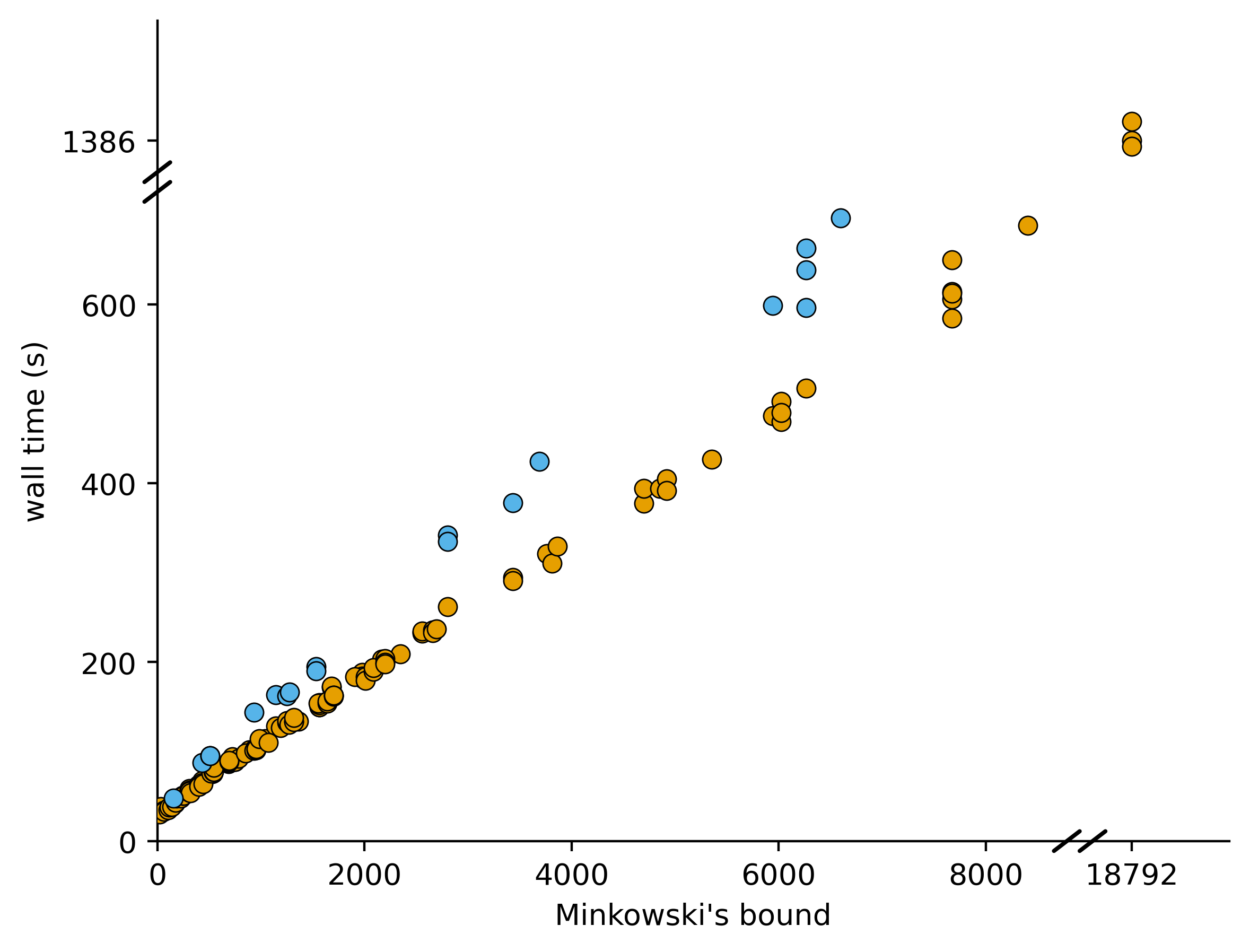}
  \end{minipage}
  {\footnotesize Color indicates number of prime divisors of $h_K$ (orange 0, blue 1, magenta 2), shape indicates number of cyclic factors of $\operatorname{Cl}(\mathcal{O}_K)$ (circle 1, triangle 2, square 3). Note the axis breaks (diagonal marks) in Figure 2.}
\end{figure}
As an example, consider the number field $K$ with defining polynomial $X^5 - X^4 + 3X^2 + 21X + 4$. 
\begin{lstlisting}
def K := AdjoinRoot (map (algebraMap ℤ ℚ) (X^5 - X^4 + 3*X^2 + 21*X + 4 ))
\end{lstlisting}

We automatically generate proofs of the listed invariants using our SageMath script. The summary file is \href{https://github.com/alainchmt/CertifyingInvariantsNF/blob/v1/IdealArithmetic/Examples/NF5_1_3790297_2/Results5_1_3790297_2.lean}{\texttt{Examples/Results5\_1\_3790297\_2}}. It proves that the discriminant and signature are
\begin{lstlisting}
theorem K_discr' : discr K = 3790297 := ...

lemma K_nrComplexPlaces' : InfinitePlace.nrComplexPlaces K = 2 := ...

lemma K_nrRealPlaces' : InfinitePlace.nrRealPlaces K = 1 := ...

\end{lstlisting}

The class group $\operatorname{Cl}(\mathcal{O}_K)$ has order $4$ and is isomorphic to the Klein group $\mathbb{Z}/2\mathbb{Z} \times \mathbb{Z}/2\mathbb{Z}$, as shown in 
\begin{lstlisting}
def class_group_equiv' : 
  (∀ i : Fin 2 , (ZMod (![2, 2] i))) ≃+ Additive (ClassGroup (RingOfIntegers K)) := ...

theorem class_number_K_eq_4' : classNumber K = 4 := ...
\end{lstlisting}

Here, \lstinline|∀ i : Fin k , (ZMod (n i))| is the type of (dependent) functions sending \lstinline|i : Fin k| to an element of type \lstinline|ZMod (n i)|, the integers modulo $n_i$. It corresponds to the group $\mathbb{Z}/n_1\mathbb{Z} \times \ldots \times \mathbb{Z}/n_k\mathbb{Z}$. 

We note that it is possible to describe the isomorphism \lstinline|class_group_equiv'| explicitly using the map from Section \ref{sec : iso class} between the class group of our computational model $\mathcal{O}$ and Mathlib's $\mathcal{O}_K$.
In addition to these examples, we also verified a selection of number fields of degrees 2, 4, and 6 through 10 with modest Minkowski bounds and non-trivial class groups, running into timeout issues with the degree-10 example in the ring of integers verification (solved by increasing Lean's \lstinline|maxHeartbeats|). 

\section{Bottlenecks and optimizations} \label{sec : optimizations}
By far, the largest portion of the verification time is spent on the factor base. For example, for the number field \href{https://github.com/alainchmt/CertifyingInvariantsNF/tree/v1/IdealArithmetic/Examples/NF5_1_337500000_4}{\texttt{Examples/NF5\_1\_337500000\_4}} with Minkowski's bound around $1,147$ and class number $5$, almost $80\%$ of CPU-time goes to the factor base verification, around $15\%$ to the relations, and the rest to ring of integers, discriminant, and saturation calculations. 
For a prime number $p$, the main bottleneck involves the proofs certifying the products of prime ideals $\prod_{i}\mathfrak{p}_i \subseteq \langle p\rangle$. 
For this, most of the time is spent on ideal primality, split between certifying a $\mathbb{Z}$-basis (needed for the norm computation) and checking the primality certificate from Section \ref{sec : primality}, which in turn requires proving irreducibility of a polynomial modulo $p$.  The remaining time is spent chaining ideal multiplications. All of these costs increase with the degree of the number field, and the irreducibility proofs dominate with large $p$.

The dependence on $p$ is illustrated in \href{https://github.com/alainchmt/CertifyingInvariantsNF/tree/v1/IdealArithmetic/Examples/NF5_1_18225000000_1}{\texttt{Examples/NF5\_1\_18225000000\_1}}. For primes $p$ in the range $2-103$, the verification spends around $32\%$ of the time on $\mathbb{Z}$-basis proofs, $36\%$ on irreducibility proofs, $14\%$ on the remaining parts of the primality certificates, and $18\%$ on chaining ideal multiplications. For primes in the range $8,263-8,317$, the proportions become approximately $17\%$, $62\%$, $10\%$, and $11\%$ respectively, with the ratio of the irreducibility proofs to the $\mathbb{Z}$-basis proofs increasing by a factor larger than $3$. 
The certification of a $\mathbb{Z}$-basis can be avoided for primes not dividing the index if one uses the Kummer-Dedekind theorem (see Remark \ref{remarkkummer}). However, this still requires proving irreducibility of polynomials modulo $p$, and therefore gives no asymptotic improvement in $p$. 

During development, we encountered performance issues arising from Lean's typeclass inference. When verifying examples, synthesizing instances for the explicitly constructed subalgebra $\mathcal{O}$, such as \lstinline[mathescape]|Algebra $\mathbb{Z}$ $\mathcal{O}$|, was slow. We mitigate this by explicitly declaring these instances, providing a shortcut for typeclass inference.
Further optimizations to the polynomial, ring of integers, and ideal arithmetic computations could be explored. Bhavik Mehta has suggested ways to optimize these for faster kernel reduction. One possible direction is to replace polymorphic functions and overloaded operations with their underlying concrete implementations, also avoiding repeated unfolding and instance synthesis. Alternative encodings could also be explored. For example, we have used tuples \lstinline[mathescape]|Fin n $\to$ A| to represent collections of elements in \lstinline|A|, as this is convenient for indexing, theorem proving, and integrates well with Mathlib (in particular, the linear algebra part of the library). This worked well in practice, however, alternative data types could be considered, at the potential cost of more laborious proofs. 

Another optimization is a finer split of independent proofs, enabling a more granular parallel verification which can be trivially sped up by increasing the number of cores. Even with highly optimized verification, however, certifying the prime ideals of norm below Minkowski's or Zimmer's bound remains computationally infeasible for number fields of large discriminant. One possible direction is to formalize and use alternative factor base bounds, such as Bach's bound, which hold under GRH.

\section{Discussion}\label{sec : discussion}
\subsection{Alternative approaches}\label{alternative}
Alternative strategies could be explored for certifying the invariants considered in this project. Instead of using the discrete logarithms approach for saturation, the non-principality of an ideal $I$ could be proven by finding a local norm obstruction. Although such certificates are not guaranteed to exist in general, the Hasse norm theorem \cite[Theorem 4.5]{janusz} shows that for number fields cyclic over $\mathbb{Q}$, whenever $N(I)$ is not a global norm, it fails to be a local norm at some completion. Another method is to compute a so-called \emph{reduced} ideal in the class of $I$ and look it up in a precomputed finite list of all the reduced principal ideals (see \cite[pag. 125]{buchmann2}) . 

More analytic and numerical approaches are also possible. 
For a number field $K$ of degree $n$ and defining polynomial $f$, the signature $(r_1, r_2)$ and approximations of the images of $\alpha \in K$ under the embeddings $\sigma : K \to \mathbb{C}$ can be found by approximating the roots of $f$ with rigorous error bounds. 

Such embedding computations can further be used to certify a system of $\mathbb{Z}$-independent units via the \emph{regulator}. Recall that the group of units modulo torsion $U_K = \mathcal{O}_K^{\times}/\mu(\mathcal{O}_K)$ is free of rank $r_1 + r_2 - 1$. The kernel of the logarithmic embedding $ l : \mathcal{O}_K^{\times} \to \mathbb{R}^{r_1 + r_2}$ (see \cite[pag. 57]{janusz})  is $\mu(\mathcal{O}_K)$, and its image is a full lattice in the hyperplane $\sum_i x_i = 0$ of $\mathbb{R}^{r_1 + r_2}$. Projecting onto $\mathbb{R}^{r_1 + r_2 - 1}$ by deleting any coordinate yields a full lattice in $\mathbb{R}^{r_1 + r_2 - 1}$ whose covolume is the \emph{regulator} $R_K$ of $K$. More generally, for a collection of units $\{u_1, \ldots, u_{r_1 + r_2 -1}\}$, their regulator is the covolume of the lattice generated by their logarithmic images in $\mathbb{R}^{r_1 + r_2 - 1}$. It can be computed via a determinant, and it is nonzero if and only if the collection of units is $\mathbb{Z}$-linearly independent. In that case, it is an integer multiple of $R_K$. 

From a collection of ideals whose classes generate $\operatorname{Cl}(\mathcal{O}_K)$, one can easily compute an integer multiple $h'$ of the class number $h_K$. Similarly, given a collection of $r_1 + r_2 -1$ $\mathbb{Z}$-linearly independent units, one can find (up to some error $\pm \epsilon$) their regulator $R'$ and thus an integer multiple of $R_K$. If a sufficiently precise
approximation of $h_KR_K$ is close enough to $h'R'$, one can conclude that $h_K = h'$ and $R_K = R'$, certifying both invariants at once. This idea avoids the saturation step in Section \ref{sec:satclass} and is used in the classical Buchmann algorithm \cite{buchmann}, which employs the class number formula to approximate $h_KR_K$.

The analytic class number formula (\href{https://github.com/leanprover-community/mathlib4/blob/6f1c6456ee863edbcf96c4febc52cd1c8d07487f/Mathlib/NumberTheory/NumberField/DedekindZeta.lean#L72-L87}{formalized} in Lean by Xavier Roblot) writes the residue at $s=1$ of the Dedekind zeta function  $\zeta_K$ of $K$ (see \cite[Definition 5.1]{neuk}) in terms of the signature, the number of roots of unity, the discriminant, and the product $h_KR_K$. Assuming that we have determined the first three invariants (which is much easier than the rest) we could try approximating $\operatorname{res}_{s=1} \zeta_K(s)$ to find $h_KR_K$ up to certain precision. 
Under GRH, there are practical ways \cite{bach2, belabas2} to approximate this residue with rigorous error bounds, by considering prime ideals of norm up to $O(\log |\operatorname{disc}(K)| / \log \log |\operatorname{disc}(K)|) ^2$. 
Preliminary work towards the class number formula approach for certifying regulators in Lean (under GRH) was carried out by some participants, including the first author, during the "Bridging Lean and LMFDB" workshop in Norwich, UK (29 June -- 3 July 2026).

\subsection{Related work}
The method of using a CAS to generate a certificate that is verified inside a proof assistant was first introduced by Harrison and Thery in \cite{skeptic}, known as the skeptic’s approach. It has been applied in various settings. For instance, in primality certificates for integers \cite{gregoire}, computations based on externally generated Grobner bases \cite{grobner}, and the verification of graph-theoretic properties \cite{houseofgraphs}. Unverified computations that are later certified also appear in \cite{paulson, zeta3}. In Lean, prior work bridging Lean and a CAS includes \cite{lewis}, the tactic polyrith, written by Dhruv Bhatia, and sageify \cite{alex}. 

Within algebraic number theory, the certification of the ring of integers of a number field of arbitrary degree in Lean is described in \cite{rings}. This work also formalizes the resultant and discriminant of polynomials via the Sylvester matrix, and its relation with the discriminant of a number field, although without computational optimizations. An explicit formula for the discriminant of cyclotomic number fields was formalized in Lean in \cite{fermat}. Class groups of Dedekind domains were formalized for the first time in \cite{dedekindforma}, and the verification of the class group for a selection of quadratic number fields in Lean was carried out in \cite{alex}. In \cite{subresultant}, algebraic numbers and their arithmetic operations were implemented via their minimal polynomials in Isabelle.
Related work on real closed fields, Sturm's theorem, and pseudo-remainder sequences was discussed in Section \ref{sturmrel}. 

Finally, foundational algorithms in linear algebra and polynomial arithmetic, often used in computational number theory, have been implemented in proof assistants. For instance, in Isabelle and Coq/Rocq \cite{jesus1, jesus2, maxime}, including the LLL basis reduction algorithm \cite{lll}. Very recently, Kim Morrison announced \href{https://github.com/leanprover/hex#hex}{Hex}, a Lean library implementing several of these algorithms, including LLL reduction.

\subsection{Future work}
Several directions for future work exist. For instance, the implementation of the optimizations described in Section \ref{sec : optimizations}. Among these, the formalization of the GRH statement and the derivation of Bach's bound remain interesting formalization challenges. 

An integrated bridge between Lean and our certificates in SageMath could improve the user experience for on-demand certifications. Our current approach generates at once the Lean files from a SageMath script, without interactive integration. A more seamless interface, requiring a substantial amount of engineering effort, would perhaps make this work more accessible. Furthermore, greater use of custom tactics could shorten the generated proofs. 

A further direction is the certification of additional invariants of number fields. The analytic class number formula approach under GRH, outlined in Section \ref{alternative}, offers a possible way to conditionally certify the full unit group. An alternative approach to prove that a subgroup $W$ equals $U_K = \mathcal{O}_K^{\times}/\mu(\mathcal{O}_K)$, which does not need to assume GRH, is to bound the index $[U_K : W]$ by finding a lower bound for the regulator $R_K$ (for instance, via \cite{fiekerR}), and use the certificate of Section \ref{sec:certunits} to prove $p \nmid [U_K : W]$ for every prime $p$ up to this bound. 
Other interesting invariants to certify include the Galois group of a number field.

Ultimately, one of the main motivations for computing number field invariants, such as the class group, is the resolution of certain Diophantine equations. Applying our methods to formally solve some of them is a natural continuation of this work.

\begin{appendices}

    \section{Commutative algebra preliminaries}\label{App:B}

We recall some basic definitions and properties of groups and rings, as can be found e.g. in the undergraduate textbooks \cite{dummit, langundergraduate}. Every finite abelian group $G$ is isomorphic to a direct product $G \cong \mathbb{Z}/n_1\mathbb{Z} \times \dots \times \mathbb{Z}/n_k\mathbb{Z}$ with $n_i \in \mathbb{Z}_{>0}$. An element $x$ in a ring $R$ is a \emph{unit} if there exists a $y$ in $R$ such that $xy = yx = 1$. The set of units of $R$ is a group under multiplication and is denoted by  $R ^ {\times}$. A commutative ring is an \emph{integral domain} if it is not the zero ring and $ab = 0$ implies $a = 0 $ or $b = 0$. An integral domain in which all nonzero elements are units is a \emph{field}. Moreover, every finite integral domain is necessarily a field. 

The following material is standard and can be found in a graduate algebra textbook such as \cite{lang2}. 

\subsection{Ideals}
    Let $R$ be a commutative ring. An \emph{ideal} \emph{I} of $R$ is a nonempty subset of $R$ which is closed under addition and such that for all $x \in R$ and all $i \in I$, we have $xi \in I$. One readily sees that $I$ is an additive subgroup since $-1 \in R$. The quotient of additive groups $R/I$ is equipped with a well-defined multiplication that makes $R/I$ into a ring. The subset $\{0\} \subseteq R$ is an ideal, known as the zero ideal of $R$.
An ideal $I$ of $R$ is \emph{prime} if $I \neq R$ and $xy \in I$ implies $x \in I$ or $y \in I$. Equivalently, $I$ is prime if and only if $R/I$ is an integral domain. Moreover, $I$ is called \emph{maximal} if $R/I$ is a field.  Note that if $R/I$ is finite for every nonzero ideal $I$, then (with $R$ not a field) maximal and nonzero prime ideals of $R$ coincide.

 For an ideal \emph{I} of $R$, there is a canonical map $R \to R/I$ sending $x$ to $x + I$ called \emph{reduction modulo $I$} or \emph{quotient map}. The coset $x + I$ is also denoted by $\bar{x}$ when $I$ is clear from the context. The ideal generated by elements $a_1, \ldots, a_n$ in $R$, denoted by $\langle a_1, \ldots, a_n\rangle$, is defined as the set of elements in $R$ that can be written as $r_1a_1 + \ldots + r_na_n$ with $r_i \in R$. An ideal $I$ is finitely generated if $I = \langle a_1, \ldots, a_n\rangle$ for some $a_i \in R$.
The ring $R$ is \emph{Noetherian} if every ideal of $R$ is finitely generated. An ideal generated by a single element is called \emph{principal}, and an integral domain is a \emph{principal ideal domain (PID)} if every ideal is principal. The canonical examples of a PID, and hence of a Noetherian ring, are $\mathbb{Z}$ and the ring of univariate polynomials $k[X]$ over a field $k$.

\subsection{Dedekind domains}
Let $R$ be a commutative ring and $I$ and $J$ ideals of $R$. Their sum $I + J$ is the ideal $\{x + y \mid x \in I, y \in J\}$, and their product $IJ$ is the smallest ideal containing $\{x y \mid x \in I, y \in J\}$. Under these operations, the collection of all ideals of $R$ forms a commutative semiring.
If $R$ is an integral domain and $F$ its field of fractions, then a subset $I \subseteq F$ is a \emph{fractional ideal} of $R$ if there exists $x \in R\setminus{\{0\}}$ such that $xI \subseteq R$ is an ideal of $R$. Addition and multiplication of (integral) ideals extend in the natural way to fractional ideals. A fractional ideal $I$ is said to be \emph{invertible} if there exists a fractional ideal $J$ such that $IJ = R$. If every nonzero fractional ideal of $R$ is invertible, then $R$ is called a \emph{Dedekind domain}. In a Dedekind domain, every nonzero proper ideal has a unique factorization as a product of prime ideals. Furthermore, every ideal can be written using at most two generators.

\subsection{Modules and Algebras}
Recall that a vector space $V$ over a field $k$ is an abelian group equipped with an action of $k$ satisfying some axioms that ensure the compatibility with both the field operations of $k$ and the group operation of $V$. An $R$-module generalizes this concept by replacing $k$ with an arbitrary commutative ring $R$. The usual notions of linear map and linear independence generalize in a straightforward manner to this more general framework. We highlight two important instances of modules. First, an ideal $I$ of a commutative ring $R$ is an $R$-module with the action given by the ring multiplication. Second, every abelian group $G$ can be viewed as a $\mathbb{Z}$-module by defining the action of $m \in \mathbb{Z}$ on $g \in G$ with the formulas $0 \cdot g = 0$, $(m + 1) \cdot g = m \cdot g + g $, and $(-m) \cdot g = - m \cdot g$ for $m > 0$.

For an $R$-module $M$ and a subset $A$ of $M$, the $R$-submodule consisting of the finite linear combinations $r_1\cdot x_1 + \ldots + r_n  \cdot x_n$ with $r_i \in R$,  $x_i \in A$ and $n \in \mathbb{N}$ is denoted by $\operatorname{span}_{R} A$. We say that $M$ is finitely generated if $M = \operatorname{span}_{R}A$ for some finite subset $A$ of $M$. If $M$ is finitely generated over $R$ and $R$ is Noetherian, then every $R$-submodule of $M$ is also a finitely generated $R$-module. A subset $B$ of $M$ is an $R$-basis of $M$ if $B$ is linearly independent and $M = \operatorname{span}_{R} B$. Unlike in the case of vector spaces, not every module has a basis (for example, $\mathbb{Z}/n\mathbb{Z}$, with $n>0$, is a $\mathbb{Z}$-module but does not have a $\mathbb{Z}$-basis). If it does, it is called a \emph{free} module. A free module $M$ with a finite basis of size $n$, called the \emph{rank} of $M$, is isomorphic to $R ^ n$. Every finitely generated torsion-free module over a PID is free. In particular, every finitely generated torsion-free abelian group has a $\mathbb{Z}$-basis.
A commutative ring $S$ equipped with a ring homomorphism $\varphi : R \to S$ is an $R$-algebra. It has an $R$-module structure via $r \cdot x = \varphi(r)x$, with $r\in R$ and $x \in S$ (thus, every commutative ring is a $\mathbb{Z}$-algebra). An ideal $I$ of $S$ can then be regarded as both an $S$-module and an $R$-module.

\subsection{Finite fields}
Let $F$ be a finite field. The canonical map $\mathbb{Z} \to F$ has kernel equal to $p\mathbb{Z}$ for some prime number $p$ called the \emph{characteristic} of $F$. This induces an injective morphism of fields $\mathbb{Z}/p\mathbb{Z} \to F$, so that $F$ may be regarded as an extension of the finite field $\mathbb{F}_p = \mathbb{Z}/p\mathbb{Z}$. The extension $ \mathbb{F}_p \subseteq F$ is simple, which means that there is $\alpha \in F$ such that $F = \mathbb{F}_p(\alpha)$. The cardinality of $F$ is $p ^ d$, where $d$ is the degree of the extension (the dimension of $F$ as an $\mathbb{F}_p$-vector space).
If $q = p^d$ is a prime power, then the splitting field (unique up to isomorphism) of $X^q - X$ over $\mathbb{F}_p$ is a finite field of size $q$, denoted $\mathbb{F}_q$. Any field with $q$ elements is isomorphic to this one, so we speak of \emph{the} finite field with $q$ elements.

	\section{Frobenius theorem}\label{app:A}
	Let $K$ be a number field and $f \in \mathcal{O}_K[X]$ a monic polynomial of degree $n$ with $n$ distinct roots in a splitting field, and take $L \supseteq K$ to be a splitting field of $f$. Consider a nonzero prime ideal $\mathfrak{q} $ of $\mathcal{O}_K$. Reducing the coefficients of $f$ modulo $\mathfrak{q}$ gives a polynomial $f\pmod{\mathfrak{q}} \in (\mathcal{O}_K/\mathfrak{q})[X]$, which factors into irreducible polynomials. The collection of degrees of these polynomials gives a partition of $n$, called the \emph{factorization type} of $f$ modulo $\mathfrak{q}$. An element $\sigma$ in the Galois group $\operatorname{Gal}(L/K)$ of $f$ can be seen as a permutation on the set of $n$ roots of $f$ in $L$, and thus as an element in the symmetric group $S_n$. The cycle type of $\sigma$ is the partition of $n$ whose parts are the lengths of the disjoint cycles in the decomposition of $\sigma$ in $S_n$. The main result we need is the  Frobenius density theorem, which can be obtained as a corollary of Chebotarev's density theorem, but in fact has a much easier direct proof. 

	\begin{theorem}[Frobenius density theorem]\label{theo:frob}
		Consider $f \in \mathcal{O}_K[X]$ as above with splitting field $L \supseteq K$. Set $G= \operatorname{Gal}(L/K)$. Let $C$ be a partition of $n$, and let $S$ be the set of prime ideals $\mathfrak{q}$ such that the $f$ modulo $\mathfrak{q}$ has factorization type $C$. Then the natural density of $S$ is given by
		\begin{align*}\label{eq:frob}
			\lim_{x \to \infty} \frac{\#\{ \mathfrak{q} \mid N(\mathfrak{q}) \leq x \text{ and } \mathfrak{q} \in S\}}{\#\{\mathfrak{q} \mid N(\mathfrak{q}) \leq x \}} = \frac{\#\{\sigma \in G \mid \sigma \text{ has cycle type } C\}}{\# G}.
		\end{align*}
	\end{theorem}
\begin{proof}
    With Dirichlet density, see \cite[Theorem 2.49]{frob} for a proof via Chebotarev's density theorem or \cite[Theorem 5.2]{janusz} for a direct proof. The result also holds with natural density (see \cite[pag. 24]{frob}).
\end{proof}

Let $p$ be a prime number and $\zeta_p$ a primitive $p$-th root of unity in an algebraic closure of $K$. We set $\delta_p = [K(\zeta_p) : K]$. Since $\mathbb{Q}(\zeta_p)/\mathbb{Q}$ has cyclic Galois group, one can show that $\delta_p$ divides $p-1$.

	\begin{lemma*}
		Let $\alpha \in \mathcal{O}_K \setminus  (\mathcal{O}_K) ^ p$. Then the natural density of prime ideals of $\mathcal{O}_K$ for which $\bar{\alpha} \in \mathcal{O}_K/\mathfrak{q}$ is a $p$-th power is equal to $1-(p-1)/(p\delta_p)$.
	\end{lemma*}
	\begin{proof}
		Since $\alpha$ is not a $p$-th power, it can be shown that the polynomial $f = X^p - \alpha$ over $\mathcal{O}_K$ is irreducible.  For a prime ideal $\mathfrak{q}$, we have that $\bar{\alpha} \in \mathcal{O}_K / \mathfrak{q}$ is a $p$-th power if and only if $X ^ p - \bar{\alpha}$ has a linear factor in $(\mathcal{O}_K / \mathfrak{q})[X]$. The extension  $L = K(\sqrt[p]{\alpha}, \zeta_p)$, is the splitting field of $f$. Since $p$ and $\delta_p$ are coprime ($\delta_p$ divides $p-1$), the degree of $L$ over $K$ equals $p\delta_p$. The elements in $\operatorname{Gal}(L/K)$ with a cycle of length $1$ in their decomposition are those that fix a root of $f$. One checks that there are $p(\delta_p - 1) + 1$ of those elements. Therefore, by Theorem \ref{theo:frob}, the density of prime ideals of $\mathcal{O}_K$ where $\bar{\alpha}$ is a $p$-th power is $(p(\delta_p - 1) + 1)/p\delta_p = 1 - (p-1)/p\delta_p$.
	\end{proof}

	\begin{lemma*}
		The natural density of the set  $ \{\mathfrak{q} \mid p | N(\mathfrak{q}) - 1\}$ of prime ideals of $\mathcal{O}_K$ equals $1/\delta_p$.
	\end{lemma*}
	\begin{proof}
		The condition $p \mid N(\mathfrak{q}) - 1$ is equivalent to $(\mathcal{O}_K/\mathfrak{q})^{\times}$ containing a primitive $p$-th root of unity. This corresponds to the $p$-th cyclotomic polynomial having a linear factor modulo $\mathfrak{q}$ and $p \not \in \mathfrak{q}$. Consider $L = K(\zeta_p)$ of degree $\delta_p$, the splitting field of the $p$-th cyclotomic polynomial over $K$. The only element in $\operatorname{Gal}(L/K)$ fixing at least one of the roots of this polynomial is the identity. Hence, by applying the Frobenius density theorem, and since there are only finitely many $\mathfrak{q}$ with $p \in \mathfrak{q}$, we have that the natural density of $ \{\mathfrak{q} \mid p | N(\mathfrak{q}) - 1\}$ is $1/\delta_p$.
	\end{proof}

\end{appendices}


	\bibliographystyle{plain}
	\bibliography{references}

\end{document}